\documentclass[12pt]{article}
\usepackage[utf8]{inputenc}
\usepackage[english]{babel}
\usepackage{amsmath, amssymb, amsthm}
\usepackage{natbib}
\usepackage{graphicx}
\usepackage{bm}
\usepackage{authblk}
\usepackage{algorithm}
\usepackage{algorithmic}

\theoremstyle{proposition}
\newtheorem{proposition}{Proposition}

\theoremstyle{definition}

\theoremstyle{remark}

\newcommand{\bmu}{\bm{\mu}}
\newcommand{\bSigma}{\bm{\Sigma}}

\newcommand{\by}{\bm{y}}
\newcommand{\bx}{\bm{x}}
\newcommand{\bj}{\bm{j}}
\newcommand{\bC}{\bm{C}}
\newcommand{\bX}{\bm{X}}
\newcommand{\bk}{\bm{k}}
\newcommand{\btau}{\bm{\tau}}
\begin{document}
	
\title{Model-based estimation and imputation with torus missing values}

\author[1]{Luca Greco (\texttt{l.greco@unifortunato.eu})}
\author[2]{Lucia Filippozzi (\texttt{lucia.filippozzi@unitn.it})}
\author[2]{Claudio Agostinelli (\texttt{claudio.agostinelli@unitn.it})}

\affil[1]{University Giustino Fortunato, Benevento, Italy}
\affil[2]{Department of Mathematics, University of Trento, Trento, Italy.}

\date{}
\maketitle
	
\begin{abstract}
This paper addresses the problem of parameter estimation and model-based imputation for multivariate circular data lying on a $p$-dimensional torus in the presence of missing values. 
Actually, the periodic nature of the sample space invalidates conventional imputation techniques designed for Euclidean data. Then, we propose a general framework for maximum likelihood estimation under the wrapped elliptically symmetric family of distributions, with particular interest in the  multivariate wrapped normal distribution, when the missing data mechanism is ignorable. 
The methodology leverages the conditional properties of the elliptically symmetric distributions on the unwrapped space, embedding the imputation of missing torus data into an Expectation-Maximization algorithm that treats both the wrapping coefficients and the missing entries as latent variables. Derivation of both  the E and M steps is detailed and a working algorithm is discussed. Imputation methods are also taken into account.
	
{\bf Keywords}: Circular data; EM algorithm; MAR; MCAR; Elliptically symmetric; Wrapped normal.
\end{abstract}
	
\section{Introduction}
\label{sec:intro}
	
Multivariate circular data arise across a wide range of scientific disciplines. In structural bioinformatics, the conformational space of protein and RNA backbones is characterized by sequences of dihedral torsion angles \citep{mardia2007protein, Eltzner2018, greco2023finite}. In environmental sciences, wind directions recorded at multiple monitoring stations or at successive time points define circular observations whose dependence structure is of primary interest \citep{jona2012spatial, lagona2012model}. In movement ecology, the simultaneous analysis of animal headings and turning angles naturally leads to multivariate circular responses \citep{rivest2016general, ranalli2020model}.
This type of data can be represented as points on the surface of a $p$-dimensional torus $\mathbb{T}^p = [0,2\pi)^p$. The topology is not a convenience but is a constraint since angles (or directions) are periodic. The intrinsic geometry of the torus, where each coordinate is identified modulo $2\pi$, imposes fundamental constraints on statistical modeling and inference that distinguish torus data from conventional multivariate observations in the Euclidean space \citep{MardiaJupp2000, pewsey2013circular}. 
	
The family of wrapped elliptically symmetric distributions (WES) is a flexible model for torus data \citep[][and references therein]{agostinelli2024weighted}.
A WES distribution arises by component-wise reduction modulo $2\pi$ of a $p$-variate random vector $\bx$ belonging to an elliptically symmetric family, i.e.  $\by = \bx \mod 2\pi$, where $\bx$  has density 
\begin{equation} \label{es_dens}
m_p(\bx; \bmu, \bSigma) = |\bSigma|^{-\frac{1}{2}} h(d^2(\bx; \bmu, \bSigma))
\end{equation}
with $\bmu \in \mathbb{R}^p$ the location vector,  $\bSigma \in \mathbb{R}^{p \times p}$ a positive definite matrix, $d^2(\bx; \bmu, \bSigma)=(\bx-\bmu)^\top\bSigma^{-1}(\bx-\bmu)$ the squared Mahalanobis distance (MD), and $h(t)>0$  a scalar real valued function, that is the density generator determining the shape of the distribution and normalized so that \eqref{es_dens} integrates to one.
The resulting WES probability density function is obtained by summing \eqref{es_dens} over all wrapping coefficients $\bj = (j_1, \ldots, j_p)^\top \in \mathbb{Z}^p$:
\begin{equation} \label{eq:wn_density}
m_p^\circ(\by; \bmu, \bSigma) = \sum_{\bj \in \mathbb{Z}^p} m_p(\by + 2\pi \bj; \bmu, \bSigma), 
\end{equation}
where $\bmu$ is now an element of the torus $\mathbb{T}^p$, which ensures identifiability. 
In practice, the infinite sum over $\bj \in \mathbb{Z}^p$ in \eqref{eq:wn_density} is truncated to a finite set $\mathcal{C}_J = \{-J, -J+1, \ldots, J\}^p$ for a sufficiently large $J$, as the terms decay rapidly for concentrated distributions \citep{nodehi2020, greco2023finite}.  
The Wrapped Normal (WN) is obtained by setting $h(t)=(2\pi)^{-p/2}\exp\left(-\frac{t}{2}\right)$, whereas the Wrapped Student $t$ with (known)  degrees of freedom $\nu$ corresponds to $h(t)=c_{p,\nu}\left(1+\frac{t}{\nu}\right)^{-\frac{p+\nu}{2}}$, with $c_{p,\nu}=\Gamma\left(\frac{\nu+p}{2}\right)/\left[\Gamma\left(\frac{\nu}{2}\right)(\nu\pi)^{p/2}\right]$.
	
The WES family is particularly amenable to likelihood-based inference. In particular, this is true for the WN distribution, since it inherits many desirable properties from the Normal distribution. In the paper, the treatment of the WES  family will be presented in general terms but particular emphasis and attention will nonetheless be devoted to the WN distribution.
Finite mixtures of WN components have been proposed for model-based clustering of heterogeneous torus data \citep{greco2023finite}, demonstrating the potential of this distributional family for unsupervised learning on non-Euclidean manifolds. The WN has additionally provided a convenient framework to develop robust inference on the torus \citep{saraceno2021robust, greco2021robust, agostinelli2024weighted}. 
	
Missingness is a ubiquitous complication in virtually all empirical investigations, arising from instrument failure, recording errors, study design, or simply the stochastic nature of data collection. Multivariate circular data are no exception: wind directions recorded by networks of meteorological stations can be affected by sensor failures, maintenance periods and transmission gaps. Similar incomplete records are common in movement ecology, where failed or rejected positional fixes leave gaps in the sequences of headings and turning angles derived from telemetry tracks, whereas in structural bioinformatics dihedral angles cannot be computed for residues that are not resolved in the experimental structure.
	
In Euclidean settings, the analysis of incomplete data is supported by an extensive literature \citep[][for a review]{ van2018flexible, little2019statistical}. Despite the maturity of the missing data literature, the treatment of incomplete multivariate circular observations has remained largely unexplored, to the best of our knowledge. For circular data, solutions to handle missing values cannot ignore the intrinsic periodicity in the data but should account for the torus topology. There are some notable exceptions \citep{lagona2012model, bulla2012multivariate}, even if their approaches rely on the von Mises distribution, which presents substantial analytical challenges \citep{mardia2012mixtures, mardia2014some, greco2023finite}. 
By contrast, the WES family preserves a direct connection to elliptically symmetric distributions, inheriting many desirable properties and making it particularly well-suited in the development of likelihood-based inference also in the presence of missing data, through a principled Expectation-Maximization (EM) based approach.
	
Throughout this paper, we operate under the assumption of an ignorable missing data mechanism \citep[Ch.~6]{little2019statistical}. Let $\by = (\by_{\text{obs}}, \by_{\text{mis}})$ denote the complete data, partitioned into observed and missing components, and let $\bm{r}$ be a binary random vector indicating which entries of $\by$ are observed.  
The missing data mechanism is characterised by the conditional distribution $p(\bm{r} \mid \by, \bm{\psi})$, where $\bm{\psi}$ is a vector of nuisance parameters.
Under the missing completely at random (MCAR) mechanism,
\begin{equation*}
p(\bm{r} \mid \by, \bm{\psi}) = p(\bm{r} \mid \bm{\psi}) \quad \forall \, \by,
\end{equation*}
that is, the probability of missingness does not depend on either the observed or the missing values.  
Under the less restrictive missing at random (MAR) mechanism,
\begin{equation*}
p(\bm{r} \mid \by, \bm{\psi}) = p(\bm{r} \mid \by_{\text{obs}}, \bm{\psi}) \quad \forall \, \by_{\text{mis}},
\end{equation*}
missingness may depend on the observed data, but, conditional on these, it is independent of the unobserved values.
MCAR is recovered as the special case in which this dependence vanishes. When even the MAR condition fails, the mechanism is termed missingnot at random (MNAR), and the missingness process must be explicitly modeled.
A missing data mechanism is said to be ignorable for likelihood-based inference when, in addition to MAR, the parameters $\bm \theta$ indexing the measurement model and the nuisance parameters $\bm{\psi}$ governing the missingness process are distinct, in the sense that their joint parameter space factorizes  \citep[Ch.~6]{little2019statistical}.	
Under ignorability, inference can be legitimately based on the observed-data likelihood
\begin{equation*}
L(\bm{\theta} \mid \by_{\text{obs}}) = \int p(\by_{\text{obs}}, \by_{\text{mis}} \mid \bm{\theta}) \, d\by_{\text{mis}},
\end{equation*}
without explicit specification of $p(\bm{r} \mid \by, \bm{\psi})$.  
	
In the context of circular data, the ignorability assumption is both practically reasonable and methodologically convenient: it allows us to focus on modeling circular measurements, while treating the missingness process as a nuisance that can be disregarded for the purposes of estimation and imputation.  Furthermore, the EM algorithm provides a principled framework for maximum likelihood estimation under an ignorable missing data mechanism.

In the following, we develop maximum likelihood estimation of the parameters of the WES distribution from partially observed torus data under an ignorable missing data mechanism. Our approach builds upon the data-augmentation perspective that has been proven successful for complete-data fitting \citep{nodehi2020, greco2023finite, agostinelli2024weighted}. The key innovation lies in the recognition that the missing components, conditional on the observed ones and on a candidate vector of wrapping coefficients, obey an elliptically symmetric distribution.
This insight enables the derivation of a principled EM algorithm that simultaneously treats two distinct sources of incompleteness: the unknown wrapping coefficients and the missing angular measurements. 
	
The remainder of the paper is organized as follows. Section \ref{sec:background} reviews the EM algorithm to fit a wrapped distribution from complete data, first, and, second, to handle missing values in elliptically symmetric distributions in the Euclidean case. Section \ref{sec:methodology} presents a principled EM algorithm for incomplete torus data under the WES model, combining the previous results, and delivers imputation strategies for the missing values.
Section \ref{sec:simulations} reports the results of Monte Carlo numerical studies designed to assess the finite-sample performance of the proposed estimation procedure with emphasis on the WN distribution under various missingness patterns. Section \ref{sec:applications} illustrates the methodology on data from protein bioinformatics, where missing values are artificially introduced to benchmark imputation accuracy. Section \ref{sec:discussion} concludes with a discussion of limitations, extensions, and directions for future research. Appendix \ref{appendix} contains further details on the introduced algorithm with a particular interest on computational complexity and complete results of the numerical studies. 
	
\section{Background}
\label{sec:background}
	
\subsection{Maximum likelihood estimation with complete torus data}
\label{sec:em_wn}
	
Before addressing the problem of missing values, we review maximum likelihood estimation (MLE) for the WES family when no observations are missing. Consider an i.i.d. sample of size $n$,  $\by= (\by_1, \ldots, \by_n)$, from a circular random variable with density function as in \eqref{eq:wn_density}.  
MLE can be obtained solving fixed point equations \citep{agostinelli2024weighted} of the form
\begin{align} \label{MLE}
\bmu & = \frac{\sum_{i=1}^n \sum_{\bj \in \mathbb{Z}^p} v_{i\bj} (\by_i + 2\pi \bj)}{\sum_{i=1}^n \sum_{\bk \in \mathbb{Z}^p} v_{i\bk}} \\
\bSigma & = - \frac{2}{n} \sum_{i=1}^n \sum_{\bj \in \mathbb{Z}^p} v_{i\bj}(\by_i + 2\pi \bj-\bmu )  (\by_i + 2\pi \bj-\bmu )^\top \nonumber \ .
\end{align}
with 
\begin{equation} \label{vijMLE}
v_{i\bj} =  \frac{h^\prime(d^2_{i\bj})}{\sum_{\bk\in\mathbb{Z}^p} h(d^2_{i\bk})}
\end{equation}
and $d^2_{i\bj}=(\by_i + 2 \pi \bj-\bmu)^\top \bSigma^{-1}(\by_i + 2 \pi \bj-\bmu)$. The Expectation-Maximization \citep[EM,][]{dempster1977maximum, mclachlan2008algorithm} algorithm provides a natural and computationally convenient framework to MLE, exploiting a data-augmentation scheme in which the unknown wrapping coefficients are treated as latent variables \citep{Fisher1994, coles1998inference, nodehi2020}.  
Then, the complete data consist of the pairs $(\by_i, \bj_i)$, with complete-data log-likelihood
\begin{equation} \label{eq:complete_ll_wn}
\ell_c(\bmu, \bSigma) = \sum_{i=1}^n \log m_p(\by_i + 2\pi \bj_i; \bmu, \bSigma) = \sum_{i=1}^n \sum_{\bj\in\mathbb{Z}^p}w_{i\bj} \log m_p(\by_i + 2\pi \bj; \bmu, \bSigma) \ ,
\end{equation}
where $w_{i\bj}=1$ when $\bj=\bj_i$, that is $\bj_i$ is the wrapping coefficient vector for the i-th observation, and $w_{i\bj}=0$ otherwise.
	
In the E-step, given current parameter estimates $\bmu^{(t)}$ and $\bSigma^{(t)}$, we compute the conditional expected value of the complete-data log-likelihood in \eqref{eq:complete_ll_wn}, obtaining
\begin{equation} \label{eq:q_function_wn}
Q(\bmu, \bSigma \mid \bmu^{(t)}, \bSigma^{(t)}) = \mathbb{E}_{\bj}\left[ \ell_c(\bmu, \bSigma \mid \by, \bmu^{(t)}, \bSigma^{(t)})  \right].
\end{equation}
This evaluation only requires computing the conditional probabilities
\begin{equation} \label{eq:posterior_weights}
\hat w_{i\bj}^{(t)} = \frac{m_p(\by_i + 2\pi \bj; \bmu^{(t)}, \bSigma^{(t)})}{\sum_{\bk \in \mathbb{Z}^p} m_p(\by_i + 2\pi \bk; \bmu^{(t)}, \bSigma^{(t)})}=\frac{h(d^{2; (t)}_{i\bj})}{\sum_{\bk\in\mathbb{Z}^p} h(d^{2; (t)}_{i\bk})} , \qquad \bj \in \mathcal{C}_J,
\end{equation}
where the infinite sum over $\mathbb{Z}^p$ is truncated to  $\mathcal{C}_J$. A value of $J = 1,2$ is typically sufficient for concentrated data \citep{nodehi2020, greco2023finite}.  
These weights represent the conditional probability that observation $\by_ i$ has wrapping coefficient vector $\bj$.
	
The M-step then updates the parameters by maximizing \eqref{eq:q_function_wn}, that is, after making explicit the use of truncation, by solving the complete-data score equations:
\begin{align*}
\bmu^{(t+1)} &= \frac{ \sum_{i=1}^n \sum_{\bj \in \mathcal{C}_J}  \hat w_{i\bj}^{(t)} \zeta_{i\bj} (\by_i + 2\pi \bj)}{\sum_{i=1}^n \sum_{\bj \in \mathcal{C}_J} \hat w_{i\bj}^{(t)} \zeta_{i\bj} }, \\
\bSigma^{(t+1)} &= -\frac{2}{n}\sum_{i=1}^n \sum_{\bj \in \mathcal{C}_J}  \hat w_{i\bj}^{(t)} \zeta_{i\bj}  (\by_i + 2\pi \bj - \bmu^{(t+1)}) (\by_i + 2\pi \bj - \bmu^{(t+1)})^\top. 
\end{align*}
where $\zeta_{i\bj} =\frac{h^\prime(d_{i\bj}^2)}{h(d_{i\bj}^2)}$. It is worth noting that $\hat w_{i\bj}^{(t)} \zeta_{i\bj}^{(t)} = v_{i\bj}^{(t)} $, proving that the EM is a principled solution to solve the fixed point equations in (\ref{MLE}) according to the updating scheme
\begin{align*}
\bmu^{(t+1)} &= \frac{ \sum_{i=1}^n \sum_{\bj \in \mathcal{C}_J} \hat v_{i\bj}^{(t)} (\by_i + 2\pi \bj)}{\sum_{i=1}^n \sum_{\bj \in \mathcal{C}_J} \hat v_{i\bj}^{(t)} }, \\
\bSigma^{(t+1)} &= -\frac{2}{n}\sum_{i=1}^n \sum_{\bj \in \mathcal{C}_J} \hat v_{i\bj}^{(t)} (\by_i + 2\pi \bj - \bmu^{(t+1)}) (\by_i + 2\pi \bj - \bmu^{(t+1)})^\top. 
\end{align*}
It is straightforward to obtain $-2\zeta_{i\bj}^{(t)}=1$ in the WN case, then simplifying the M-step updates to 
\begin{align*}
\bmu^{(t+1)} &= \frac{1}{n} \sum_{i=1}^n \sum_{\bj \in \mathcal{C}_J} \hat w_{i\bj}^{(t)} (\by_i + 2\pi \bj), \\
\bSigma^{(t+1)} &= \frac{1}{n} \sum_{i=1}^n \sum_{\bj \in \mathcal{C}_J} \hat w_{i\bj}^{(t)} (\by_i + 2\pi \bj - \bmu^{(t+1)}) (\by_i + 2\pi \bj - \bmu^{(t+1)})^\top. 
\end{align*}

\subsection{MLE for elliptically symmetric distributions with missing values}
\label{sec:em_normal}

The family of elliptically symmetric distributions forms the backbone of the WES model, and its conditional properties \citep[][Ch. 2]{fang1990symmetric} are central to our treatment of missing torus data.  
Partition $\bx_i = (\bx_{i,\text{obs}}, \bx_{i,\text{mis}})$ according to observed and missing entries and, analogously,
\[
\bmu = (\bmu_{\text{obs}}, \bmu_{\text{mis}}),
\quad
\bSigma = \begin{pmatrix}
\bSigma_{\text{obs},\text{obs}} & \bSigma_{\text{obs},\text{mis}} \\
\bSigma_{\text{mis},\text{obs}} & \bSigma_{\text{mis},\text{mis}}
\end{pmatrix}.
\]
Under the ignorability assumptions discussed in Section~\ref{sec:intro}, the observed-data log-likelihood is
\begin{equation} \label{eq:observed_likelihood_normal}
\ell(\bmu, \bSigma \mid \bx_{\text{obs}}) = \sum_{i=1}^n \log m_{p_i}(\bx_{i,\text{obs}}; \bmu_{i,\text{obs}}, \bSigma_{i,\text{obs},\text{obs}}),
\end{equation}
where $p_i = \dim(\bx_{i,\text{obs}})$ and $m_{p_i}(\bx_{i,\text{obs}}; \bmu_{i,\text{obs}}, \bSigma_{i,\text{obs},\text{obs}})$ denotes the marginal density of the observed components, which is still in the elliptically symmetric family.
Direct maximization of \eqref{eq:observed_likelihood_normal} is cumbersome due to the observation-specific marginal distributions. The EM algorithm provides a computationally convenient alternative. Actually, 
whenever $h$ is completely monotone, Bernstein's theorem \citep{widder1941laplace, feller1971introduction, fang1990symmetric} guarantees the Gaussian scale-mixture representation
\begin{equation}
h(u) = \int_0^\infty (2\pi)^{-p/2}\tau^{p/2}\,e^{-\tau u/2}\, dG(\tau), \quad \bX\mid \tau \sim N_p(\bmu,\bSigma/\tau),\ \ \tau\sim G,
\label{eq:kelker}
\end{equation}
which suggests augmenting each observation with an unobserved scale $\tau_i$. 
Conditionally on $\tau_i$ the complete-data density is Gaussian, so that
$$\bx_{i, \text{mis}}\mid \bx_{i, \text{obs}}, \tau_i\sim N_{p-p_i}\left(\bmu_{i,\text{mis}\mid \text{obs}}, \bSigma_{i,\text{mis}\mid \text{obs}} /\tau_i\right)$$ with
\begin{eqnarray} \label{eq:cond_mean}
\bmu_{\text{mis}\mid \text{obs}}&=& \bmu_{\text{mis}} + \bSigma_{\text{mis},\text{obs}} \bSigma_{\text{obs},\text{obs}}^{-1} (\bx_{\text{obs}} - \bmu_{\text{obs}}), \\ 
\bSigma_{\text{mis}\mid \text{obs}} &=& \bSigma_{\text{mis},\text{mis}} - \bSigma_{\text{mis},\text{obs}} \bSigma_{\text{obs},\text{obs}}^{-1} \bSigma_{\text{obs},\text{mis}} \nonumber
\end{eqnarray}
The complete-data log-likelihood is then
\begin{eqnarray}
\ell_c(\bmu,\bSigma) &=& -\frac{n}{2}\log|\bSigma| - \frac{1}{2}\sum_{i=1}^n \tau_i\,(\bx_i-\bmu)^\top\bSigma^{-1}(\bx_i-\bmu) \label{eq:complete-loglik} \\
& = & - \frac{n}{2} \log |\bSigma| - \frac{1}{2} \textrm{trace}\left[ \bSigma^{-1} \sum_{i=1}^n\tau_i\left( \bC_{i} - 2 \bx_i \bmu^\top + \bmu \bmu^\top \right) \right] \nonumber 
\end{eqnarray}
where $\bC_i = \bx_i \bx_i^\top$, so that $\ell_c$ is linear in the sufficient statistics
\[
T_0=  \sum_{i=1}^n \tau_i, \quad T_1 = \sum_{i=1}^n \tau_i\bx_i, \quad T_2 = \sum_{i=1}^n \tau_i\bx_i \bx_i^\top = \sum_{i=1}^n \tau_i\bC_i \ . 
\]

In the E-step, the expected value of the complete-data log-likelihood \eqref{eq:complete-loglik}, conditional on the observed data and the current parameter estimates $\bmu^{(t)}, \bSigma^{(t)}$, 
reduces to computing $\hat{\tau}_i^{(t)} =\mathbb{E}[\tau_i\mid \bx_{i,\text{obs}}, \bmu^{(t)}, \bSigma^{(t)}]$ and 
\begin{eqnarray} \label{eq:exp_obs} 
\mathbb{E}[\tau_i\bx_i \mid \bx_{i,\text{obs}}, \bmu^{(t)}, \bSigma^{(t)}]
&=&\hat\tau_i^{(t)}(\bx_{i,\text{obs}},\hat\bx_{i,\text{mis}})  =\hat\tau_i^{(t)}\left(\bx_{i,\text{obs}}, \mu_{i,\text{mis}\mid \text{obs}}^{(t)}\right) \\
&=&\hat\tau_i^{(t)} \hat\bx_i^{(t)}, \nonumber \\
\mathbb{E}[\tau_i\bx_i \bx_i^\top \mid \bx_{i,\text{obs}}, \bmu^{(t)}, \bSigma^{(t)}] 
&=& \hat\tau_i^{(t)} \begin{bmatrix}
\bx_{i,\text{obs}} \bx_{i,\text{obs}}^\top & \bx_{i,\text{obs}} \hat \bx_{i,\text{mis}} ^\top \\
\hat \bx_{i,\text{mis}} \bx_{i,\text{obs}}^\top & \left(\hat\tau_i^{(t)}\right)^{-1}\bSigma_{i, \text{mis}\mid \text{obs}} + \hat \bx_{i,\text{mis}} \hat \bx_{i,\text{mis}}^\top
\end{bmatrix} \nonumber  \\
&=& \hat\tau_i^{(t)} \hat{\bC}_i^{(t)} , \nonumber
\end{eqnarray}
where the subfix ${\it i}$ denotes that observed and missing components correspond to the entries of $\bx_i$.
Evaluating the expectation in \eqref{eq:exp_obs} requires computing the first two conditional moments of $\bx_{i} = (\bx_{i,\text{obs}}, \bx_{i,\text{mis}})$ given the observed components. 
Crucially, the conditional expectation of the outer product $\bx_{i,\text{mis}} \bx_{i,\text{mis}}^\top$ is not simply the outer product of the conditional means, but must also include the conditional covariance.

The expectation of the latent scale variable conditional on the observed components $\mathbb{E}[\tau_i\mid \bx_{i,\text{obs}}]$ depends on the marginal elliptical distribution induced by the observed coordinates. Let  \(r_i=p-p_i\) denote the number of missing components. Marginalization of an elliptically symmetric density yields the marginal generator \citep{fang1990symmetric, kano1994consistency}
\begin{equation}
\label{eq:kano}
h_{p_i}(u)
=
\frac{\pi^{r_i/2}}{\Gamma(r_i/2)}
\int_u^\infty
(s-u)^{r_i/2-1}h(s)\,ds,
\qquad r_i>0,
\end{equation}
with \(h_{p_i}=h\) when \(r_i=0\). This representation makes explicit that the marginal generator may depend on the missingness pattern.
In the Gaussian case, the marginal generator \eqref{eq:kano} reduces exactly to the Gaussian generator in the observed dimension, that is,
$
h_{p_i}(u)=(2\pi)^{-p_i/2}e^{-u/2}.
$
Analogously, conditioning on the observed components yields again an elliptically symmetric distribution for the missing ones, with location $\bmu_{i,\text{mis}\mid\text{obs}}$ and scale matrix $\bSigma_{i,\text{mis}\mid\text{obs}}$ as in \eqref{eq:cond_mean}, but with conditional generator
\begin{equation}\label{eq:cond_gen}
h_{\text{mis}\mid\text{obs}}(u)=\frac{h\big(u+d^2_{i,\text{obs}}\big)}{h_{p_i}\big(d^2_{i,\text{obs}}\big)},
\end{equation}
which depends on the observed components through the partial (observed) MD $d^2_{i,\text{obs}}$ where 
\begin{equation*}
d_{i,\text{obs}}^2 = (\bx_{i,\text{obs}}-\bmu_{i,\text{obs}})^\top\bSigma_{i,\text{obs},\text{obs}}^{-1}(\bx_{i,\text{obs}}-\bmu_{i,\text{obs}}). 
\end{equation*}
The normalizing constant in \eqref{eq:cond_gen} is the marginal generator \eqref{eq:kano}. Indeed, $|\bSigma|=|\bSigma_{i,\text{obs},\text{obs}}|\,|\bSigma_{i,\text{mis}\mid\text{obs}}|$ and
\begin{equation*}
d^2(\bx_i;\bmu,\bSigma)=d^2_{i,\text{obs}}+d^2(\bx_{i,\text{mis}};\bmu_{i,\text{mis}\mid\text{obs}},\bSigma_{i,\text{mis}\mid\text{obs}}),
\end{equation*}
so that
\begin{multline*}
\int_{\mathbb{R}^{r_i}} |\bSigma_{i,\text{mis}\mid\text{obs}}|^{-\frac12}\, h\big(d^2_{i,\text{obs}}+d^2(\bx_{i,\text{mis}};\bmu_{i,\text{mis}\mid\text{obs}},\bSigma_{i,\text{mis}\mid\text{obs}})\big)\, d\bx_{i,\text{mis}}\\
= \frac{\pi^{r_i/2}}{\Gamma(r_i/2)}\int_0^\infty t^{r_i/2-1}\, h\big(t+d^2_{i,\text{obs}}\big)\,dt
= h_{p_i}\big(d^2_{i,\text{obs}}\big),
\end{multline*}
and $h_{\text{mis}\mid\text{obs}}$ is a density generator in dimension $r_i$. Under the scale-mixture representation \eqref{eq:kelker}, \eqref{eq:cond_gen} is the Gaussian generator in dimension $r_i$ averaged with respect to the conditional distribution of $\tau_i$ given $\bx_{i,\text{obs}}$, consistently with the conditional normal distribution of $\bx_{i,\text{mis}}$ given $\bx_{i,\text{obs}}$ and $\tau_i$ above and with \eqref{eq:tau-hat-missing} below. In the Gaussian case $h_{\text{mis}\mid\text{obs}}$ reduces to the Gaussian generator in dimension $r_i$, whatever the value of $d^2_{i,\text{obs}}$.
Then, under the Gaussian scale-mixture representation, Bayes' rule gives
\begin{equation}
\label{eq:tau-hat-missing}
\mathbb{E}[\tau_i\mid \bx_{i,\text{obs}}]
= \frac{\int_0^\infty \tau^{p_i/2+1}e^{-\tau d^2_{i,\text{obs}}/2} dG(\tau)}
{\int_0^\infty \tau^{p_i/2}e^{-\tau d^2_{i,\text{obs}}/2} dG(\tau)}
= -2 \frac{h_{p_i}'(d_{i,\mathrm{obs}}^2)}{h_{p_i}(d_{i,\mathrm{obs}}^2)}.
\end{equation}
Then, the conditional expected value of the complete-data log-likelihood is
\begin{equation}
Q(\theta\mid\theta^{(t)}) = -\frac{n}{2}\log|\bSigma| - \frac{1}{2}\operatorname{tr}\!\left[\bSigma^{-1}\left(\hat T_2^{(t)}- 2\hat T_1^{(t)}\bmu^\top+\hat T_0^{(t)}\bmu\bmu^\top\right)\right],
\label{eq:Q-missing}
\end{equation}
with $\theta=(\bmu, \bSigma)$ and $\hat T_k$ ($k=0,1,2$) being the $T_k$ evaluated at $\hat \tau_i^{(t)}$ and $\hat x_i^{(t)}$, yielding the closed-form M-step
\begin{equation}
\hat\bmu^{(t+1)}  = \frac{\hat T_1^{(t)}}{\hat T_0^{(t)} }, \quad \hat\bSigma^{(t+1)}  = \frac{1}{n}\left(\hat T_2^{(t)}  - \frac{\hat T_1^{(t)} \left(\hat T_1^{(t)}\right) ^\top}{\hat T_0^{(t)}} \right).
\label{eq:M-step-missing}
\end{equation}
The Gaussian family is closed under marginalization and gives $\hat\tau_i=1$ irrespective of the missingness pattern, since the marginal generator is proportional to $\exp(-u/2)$ for every observed dimensions. 
The Student $t_p(\nu)$ family is likewise closed under marginalization, with marginal generator proportional to $(1+u/\nu)^{-(\nu+p_{\text{obs}})/2}$, i.e. a $t_{p_{\text{obs}}}(\nu)$ generator with the same $\nu$. Equation \eqref{eq:tau-hat-missing} then gives the closed-form weights
\begin{equation}
\hat\tau_i = \frac{\nu+p_{i}}{\nu+d_{i,\text{obs}}^2}.
\label{eq:t-weight-missing}
\end{equation}
For a generator $h$ not closed under marginalization, $h_{p_{\text{obs}}}$ in \eqref{eq:kano}  does not belong to the same parametric family as $h$ and must be evaluated explicitly  (analytically or by one-dimensional quadrature) for each distinct missingness pattern occurring in the data.

\section{EM for WES estimation with missing torus values}
\label{sec:methodology} 
The conditional formulas \eqref{eq:cond_mean} and the expectations in \eqref{eq:exp_obs}, \eqref{eq:tau-hat-missing} within the scale mixture representation  in \eqref{eq:kelker} are the key building blocks that will be extended to the WES setting, where the wrapping coefficients constitute an additional latent layer.

Under the ignorability assumptions stated in Section~\ref{sec:intro}, inference can be based on the observed-data loglikelihood
\begin{equation} \label{eq:observed_likelihood_wn}
\ell(\bmu, \bSigma \mid \by_{\text{obs}}) = \sum_{i=1}^n \log m_{p_i}^\circ(\by_{i,\text{obs}}; \bmu_{i,\text{obs}}, \bSigma_{i,\text{obs},\text{obs}}),
\end{equation}
where $\by_{i,\text{obs}}$ denotes the observed torus components for the $i$-th observation and $m_{p_i}^\circ(\by_{i,\text{obs}}; \bmu_{i,\text{obs}}, \bSigma_{i,\text{obs},\text{obs}})$ is its marginal wrapped density, obtained by summing over the wrapping coefficients of the observed dimensions only. Direct maximization of \eqref{eq:observed_likelihood_wn} is intractable due to the two layers of latent structure. However, it is possible to recast the problem in an extended principled EM algorithm. 
Let the complete data now consist of the triples $(\by_{i,\text{obs}}, \by_{i,\text{mis}}, \bj_i)$, where it also holds that $\bj_i = (\bj_{i,\text{obs}}, \bj_{i,\text{mis}})$. Allowing the scale mixture representation in \eqref{eq:kelker} and introducing the further vector of unknowns $\btau$, the complete-data log-likelihood can be written
\begin{align} \label{eq:complete_ll_wn_na}
\ell_c(\bmu, \bSigma & \mid \by_{\text{obs}}, \by_{\text{mis}}, \bj, \btau) = \nonumber \\
& = -\frac{n}{2}\log|\bSigma|-\frac{1}{2} \sum_{i=1}^n \sum_{\bj\in\mathcal{C}_J} w_{i\bj} \tau_{i\bj} d^2_{i\bj}(\bmu, \bSigma)\\
&=  - \frac{n}{2} \log |\bSigma| - \frac{1}{2} \textrm{trace}\left[ \bSigma^{-1} \sum_{i=1}^n \sum_{\bj\in\mathcal{C}_J} w_{i\bj}\tau_{i\bj}\left( \bC_{i\bj} - 2 \bx_{i\bj} \bmu^\top + \bmu \bmu^\top \right) \right] .\nonumber
\end{align}
The EM algorithm can deliver MLE by iteratively computing the expectation of \eqref{eq:complete_ll_wn_na} 
with respect to the joint distribution of the latent variables $(\by_{\text{mis}}, \bj, \btau)$ given the observed data $\by_{\text{obs}}$. 
The joint expectation over the latent space is factorized as
\begin{equation} \label{eq:q_function_wn_na}
Q(\bmu, \bSigma \mid \bmu^{(t)}, \bSigma^{(t)}) = 
\mathbb{E}_{\bj} \mathbb{E}_{\by_{\text{mis}}, \btau \mid \bj} \left[ \ell_c(\bmu, \bSigma \mid \by_{\text{obs}}, \by_{\text{mis}}, \bj)\mid \by_{\text{obs}}, \bmu^{(t)}, \bSigma^{(t)} \right] .
\end{equation}
The key insight is that, conditional on the scale mixture component and the wrapping coefficients $\bj$, the unwrapped observations $\bx_{i\bj}= \by_i + 2\pi \bj$ follow a multivariate normal distribution as in \eqref{eq:kelker}, with $\bx_{i\bj} = (\bx_{i\bj,\text{obs}}, \bx_{i\bj,\text{mis}})$, as well.

The inner conditional expectation is linear in the sufficient statistics $T_{0, \bj}=\sum_{i=1}^n \tau_{i\bj} , T_{1,\bj} = \sum_{i=1}^n \tau_{i\bj} \bx_{i\bj}$ and $T_{2,\bj} = \sum_{i=1}^n \tau_{i\bj} \bx_{i\bj} \bx_{i\bj}^\top=\sum_{i=1}^n \tau_{i\bj} \bC_{i\bj}$. Therefore, its evaluation only requires using  the first two conditional moments in \eqref{eq:exp_obs}, given the candidate wrapping coefficient vector.

The outer expectation over the wrapping coefficients only involves the observed components. Indeed, the unwrapped missing components $\boldsymbol{x}_{i,\mathrm{mis}}$ are themselves latent, so that $\boldsymbol{j}_{i,\mathrm{mis}}$ is absorbed by the integration over the missing dimensions and does not enter the observed-data log-likelihood \eqref{eq:observed_likelihood_wn}. For each $i$, let $\mathcal{C}_{J,i}=\{-J,\dots,J\}^{p_i}$ denote the truncated set of wrapping coefficient vectors for the $p_i$ observed components, with generic element $\boldsymbol{j}_{\mathrm{obs}}$. The conditional probabilities in \eqref{eq:posterior_weights} are then replaced by
\begin{equation} \label{eq:posterior_weights_obs}
\hat{w}^{(t)}_{i\boldsymbol{j}_{\mathrm{obs}}}
=\frac{m_{p_i}\big(\boldsymbol{y}_{i,\mathrm{obs}}+2\pi\boldsymbol{j}_{\mathrm{obs}};\boldsymbol{\mu}^{(t)}_{i,\mathrm{obs}},\boldsymbol{\Sigma}^{(t)}_{i,\mathrm{obs},\mathrm{obs}}\big)}
{\sum_{\boldsymbol{k}_{\mathrm{obs}}\in \mathcal{C}_{J,i}} m_{p_i}\big(\boldsymbol{y}_{i,\mathrm{obs}}+2\pi\boldsymbol{k}_{\mathrm{obs}};\boldsymbol{\mu}^{(t)}_{i,\mathrm{obs}},\boldsymbol{\Sigma}^{(t)}_{i,\mathrm{obs},\mathrm{obs}}\big)},
\qquad \boldsymbol{j}_{\mathrm{obs}}\in \mathcal{C}_{J,i},
\end{equation}
which sum to one over $\mathcal{C}_{J,i}$. Hence, the $(2J+1)^{p}$ configurations in $\mathcal{C}_J$ collapse into the $(2J+1)^{p_i}$ configurations in $\mathcal{C}_{J,i}$; when all components are observed, $\mathcal{C}_{J,i}=\mathcal{C}_J$ and \eqref{eq:posterior_weights_obs} reduces to the complete-data weights \eqref{eq:posterior_weights}. Then, given $\hat{\boldsymbol{x}}^{(t)}_{i\boldsymbol{j}_{\mathrm{obs}}}=\big(\boldsymbol{y}_{i,\mathrm{obs}}+2\pi\boldsymbol{j}_{\mathrm{obs}},\,\boldsymbol{\mu}^{(t)}_{i,\mathrm{mis}|\mathrm{obs};\boldsymbol{j}_{\mathrm{obs}}}\big)$, we obtain
\begin{align*}
\hat{T}^{(t)}_0 &= \sum_{i=1}^n \sum_{\boldsymbol{j}_{\mathrm{obs}}\in \mathcal{C}_{J,i}} \hat{w}^{(t)}_{i\boldsymbol{j}_{\mathrm{obs}}}\,\hat{\tau}^{(t)}_{i\boldsymbol{j}_{\mathrm{obs}}},\\
\hat{T}^{(t)}_1 &= \sum_{i=1}^n \sum_{\boldsymbol{j}_{\mathrm{obs}}\in \mathcal{C}_{J,i}} \hat{w}^{(t)}_{i\boldsymbol{j}_{\mathrm{obs}}}\,\hat{\tau}^{(t)}_{i\boldsymbol{j}_{\mathrm{obs}}}\,\hat{\boldsymbol{x}}^{(t)}_{i\boldsymbol{j}_{\mathrm{obs}}},\\
\hat{T}^{(t)}_2 &= \sum_{i=1}^n \sum_{\boldsymbol{j}_{\mathrm{obs}}\in \mathcal{C}_{J,i}} \hat{w}^{(t)}_{i\boldsymbol{j}_{\mathrm{obs}}}\,\hat{\tau}^{(t)}_{i\boldsymbol{j}_{\mathrm{obs}}}\,\hat{\boldsymbol{C}}^{(t)}_{i\boldsymbol{j}_{\mathrm{obs}}},
\end{align*}
and M-step updates as in \eqref{eq:M-step-missing}. Notably, in the WN case the M-step reduces to computing
\begin{align*}
\boldsymbol{\mu}^{(t+1)} &= \frac{1}{n}\sum_{i=1}^n \sum_{\boldsymbol{j}_{\mathrm{obs}}\in \mathcal{C}_{J,i}} \hat{w}^{(t)}_{i\boldsymbol{j}_{\mathrm{obs}}}\,\hat{\boldsymbol{x}}^{(t)}_{i\boldsymbol{j}_{\mathrm{obs}}},\\
\boldsymbol{\Sigma}^{(t+1)} &= \frac{1}{n}\sum_{i=1}^n \sum_{\boldsymbol{j}_{\mathrm{obs}}\in \mathcal{C}_{J,i}} \hat{w}^{(t)}_{i\boldsymbol{j}_{\mathrm{obs}}}\,\hat{\boldsymbol{C}}^{(t)}_{i\boldsymbol{j}_{\mathrm{obs}}} - \boldsymbol{\mu}^{(t+1)}\boldsymbol{\mu}^{(t+1)\top}.
\end{align*}
The mean vector is always projected onto the torus by taking each component modulo $2\pi$, i.e., $\bmu^{(t+1)} \leftarrow \bmu^{(t+1)} \mod 2\pi$, to ensure identifiability. 
More details on the  structure of the proposed EM algorithm with emphasis on the WN case are given in the Appendix, with a particular interest on computational complexity. It is worth stressing that the current implementation accounts for missingness patterns rather than casewise iterations in expectations evaluations during the E-step.

The proposed algorithm is a genuine instance of the EM algorithm and therefore inherits its fundamental convergence properties. In particular, each iteration increases the observed-data log-likelihood, and under  standard regularity conditions the sequence of updates $\{\bmu^{(t)}, \bSigma^{(t)}\}_{t \geq 0}$ generated by the EM algorithm converges to a stationary point of the observed log-likelihood. These conditions are satisfied provided the covariance matrices remain positive definite and the vector means are confined to the torus. Moreover, it is worth noting that the proposed EM algorithm shares and expands the basic ideas of the solution first proposed in \cite{ghahramani1995learning} in the context of Gaussian mixture modeling with missing values.

\subsection{Imputation}
The WES  family of distributions is not closed under conditioning and hence the following result holds.
\begin{proposition}\label{propo1}
Let $\by$ be a $p$-variate circular random variable whose density $m^\circ_p (\by; \bmu,\bSigma; h)$ is an element of the WES family with generator $h$.
Let $\by=(\by_1,  \by_2)$ with partial dimensions $p_1$ and $p_2$, respectively, and partition $\bmu$ and $\bSigma$ accordingly. The conditional distribution of $\by_2\mid \by_1$ is a mixture of WES distributions with density function of the form
\begin{equation*}
m^\circ(\by_2 \mid \by_1)=\sum_{\bj_1\in \mathbb{Z}^{p_1}} w_{\bj_1}\, m^\circ_{p_2}(\by_2; \bmu_{2\mid 1; \bj_1}, \bSigma_{2\mid 1}; h_{\bj_1})
\end{equation*}
with 
\begin{equation*}
\bmu_{2\mid 1; \bj_1}= \bmu_2+\bSigma_{21}\bSigma_{11}^{-1}(\by_1 + 2\pi\bj_1 - \bmu_1), \quad \bSigma_{2\mid 1}=\bSigma_{22}-\bSigma_{21}\bSigma_{11}^{-1} \bSigma_{21}^\top ,
\end{equation*}
component-specific conditional generators
\begin{equation*}
h_{\bj_1}(u)=\frac{h\big(u+d^2_{1,\bj_1}\big)}{h_{p_1}\big(d^2_{1,\bj_1}\big)}, \qquad d^2_{1,\bj_1}=(\by_1 + 2\pi\bj_1 - \bmu_1)^\top\bSigma_{11}^{-1}(\by_1 + 2\pi\bj_1 - \bmu_1),
\end{equation*}
that is, the conditional generator \eqref{eq:cond_gen} evaluated at $d^2_{1,\bj_1}$, and weights
\begin{equation*}
w_{\bj_1} = \frac{m_{p_1}(\by_{1} + 2\pi \bj_1; \bmu_{1}, \bSigma_{11}; h_{p_1})}{\sum_{\bk_1 \in \mathbb{Z}^{p_1}} m_{p_1}(\by_{1} + 2\pi \bk_1; \bmu_{1}, \bSigma_{11}; h_{p_1})},
\end{equation*}
where $h_{p_1}$ is the marginal generator \eqref{eq:kano} with $r_i=p_2$.
\end{proposition}

\begin{proof}
For $\bx=(\bx_1,\bx_2)\in\mathbb{R}^p$, the partition of $\bSigma$ gives $|\bSigma|=|\bSigma_{11}|\,|\bSigma_{2\mid 1}|$ and
\begin{equation*}
d^2(\bx;\bmu,\bSigma)=d^2(\bx_1;\bmu_1,\bSigma_{11})+d^2\big(\bx_2;\bmu_{2\mid 1}(\bx_1),\bSigma_{2\mid 1}\big)
\end{equation*}
\citep[see e.g.][]{mardia1979multivariate},
with $\bmu_{2\mid 1}(\bx_1)=\bmu_2+\bSigma_{21}\bSigma_{11}^{-1}(\bx_1-\bmu_1)$. Hence, setting $\bx_1=\by_1+2\pi\bj_1$ and $\bx_2=\by_2+2\pi\bj_2$,
\begin{align*}
& m_p(\by_1 + 2\pi \bj_1,\, \by_2 + 2\pi \bj_2; \bmu, \bSigma; h)\\
&\qquad = |\bSigma_{11}|^{-\frac12}|\bSigma_{2\mid 1}|^{-\frac12}\,
h\big(d^2_{1,\bj_1}+d^2(\by_2 + 2\pi \bj_2;\bmu_{2\mid 1;\bj_1},\bSigma_{2\mid 1})\big)\\
&\qquad = m_{p_1}(\by_1 + 2\pi \bj_1; \bmu_1, \bSigma_{11}; h_{p_1})\,
m_{p_2}(\by_2 + 2\pi \bj_2; \bmu_{2\mid 1;\bj_1}, \bSigma_{2\mid 1}; h_{\bj_1}),
\end{align*}
where $h_{\bj_1}$ is the generator of the conditional distribution of $\by_2+2\pi\bj_2$ given $\by_1+2\pi\bj_1$,
which is a density generator in dimension $p_2$, as shown in Section \ref{sec:em_normal}.
The conditional density is obtained as the ratio of the joint wrapped density to the marginal wrapped density:
\begin{align*}
m^\circ(\by_{2} \mid \by_{1}) &= 
\frac{
\sum_{\bj_{1}}
\sum_{\bj_{2}}
m_p(
\by_{1} + 2\pi \bj_{1},\,
\by_{2} + 2\pi \bj_{2};
\bmu, \bSigma; h
)
}{
\sum_{\bk_{1}}
m_{p_{1}}(
\by_{1} + 2\pi \bk_{1};
\bmu_{1},
\bSigma_{11}; h_{p_1}
)
}\\
&= \sum_{\bj_{1}}
\frac{
m_{p_1}(
\by_{1} + 2\pi \bj_{1};
\bmu_{1}, \bSigma_{11}; h_{p_1}
)
}{
\sum_{\bk_{1}}
m_{p_{1}}(
\by_{1} + 2\pi \bk_{1};
\bmu_{1},
\bSigma_{11}; h_{p_1})
}\\
&\qquad\times\sum_{\bj_{2}}
m_{p_2}(
\by_{2} + 2\pi \bj_{2};
\bmu_{2\mid 1;\bj_1}, \bSigma_{2\mid 1}; h_{\bj_1}
).
\end{align*}
\end{proof}
In the WN case, $h(u)=(2\pi)^{-p/2}e^{-u/2}$ and $h_{p_1}(u)=(2\pi)^{-p_1/2}e^{-u/2}$, so that $h_{\bj_1}(u)=(2\pi)^{-p_2/2}e^{-u/2}$ for every $\bj_1$: all mixture components are WN with common covariance matrix $\bSigma_{2\mid 1}$ and differ only in location. For the wrapped Student $t$ with $\nu$ degrees of freedom, the component indexed by $\bj_1$ is a wrapped $t$ with $\nu+p_1$ degrees of freedom and scale matrix $(\nu+d^2_{1,\bj_1})\bSigma_{2\mid 1}/(\nu+p_1)$.

The weights in Proposition \ref{propo1}, truncated to $\mathcal{C}_{J,i}$, are those given in \eqref{eq:posterior_weights_obs}.
Therefore, model-based imputations of the missing multivariate circular data requires averaging over all plausible wrapping configurations, respecting the true mixture nature of the conditional distribution. 
This ensures that the imputed values are coherent with the underlying torus geometry and with the estimated covariance structure. Principled conditional imputations are
 obtained by reducing the average imputed value at convergence modulo $2\pi$:
\begin{equation} \label{eq:imputation}
\hat{\by}_{i,\text{mis}} = \left( \sum_{\bj_{\text{obs}} \in \mathcal{C}_{J,i}} \hat w_{i\bj_{\text{obs}}}^{(\infty)} \, \hat{\bx}_{i\bj_{\text{obs}},\text{mis}}^{(\infty)} \right) \mod 2\pi,
\end{equation}
where the superscript $(\infty)$ denotes quantities evaluated at convergence.

In the WN case, multiple imputations \cite[see][for a similar solution in Gaussian mixture modeling]{di2007imputation}  can be generated by drawing from the conditional distribution of $\by_{i,\text{mis}}$ given $\by_{i,\text{obs}}$, as provided in Proposition \ref{propo1}, evaluated at the final parameter estimates. Since this distribution is a finite mixture indexed by the wrapping coefficients of the observed components only, each draw is obtained in two steps: first $\bj_{\text{obs}}$
 is sampled from a multinomial distribution on $\mathcal{C}_{J,i}$ with probabilities $\hat w_{i\bj_{\text{obs}}}^{(\infty)}$ in \eqref{eq:posterior_weights_obs}; then, given $\bj_{\text{obs}}$,
 the circular missing components are generated from a WN distribution with parameters $\bmu^{(\infty)}_{i,\text{mis}\mid \text{obs}; \bj_{\text{obs}}}$ and $\bSigma^{(\infty)}_{i,\text{mis}\mid \text{obs}}$, that is, by drawing the unwrapped missing components from $N_{r_i}\big(\bmu^{(\infty)}_{i,\text{mis}\mid \text{obs}; \bj_{\text{obs}}}, \bSigma^{(\infty)}_{i,\text{mis}\mid \text{obs}}\big)$ and reducing them modulo $2\pi$.
This provides a flexible tool for subsequent inference \citep{van2018flexible}. Note, however, that draws are conditional on the MLE, so that parameter uncertainty is not propagated and the resulting imputations are improper \citep{rubin1987multiple, van2018flexible}. Proper multiple imputation is beyond the scope of this paper and is left for future work.

\section{Numerical studies}
\label{sec:simulations}
In this section, we investigate the finite sample behavior of the proposed MLE evaluated through the EM algorithm. It is compared with the MLE over the complete data before missingness is introduced, the MLE evaluated after discarding cases with missing entries and the MLE obtained after imputing the data using the column-wise circular mean.  The EM algorithm is initialized from the MLE evaluated over complete cases. 
Data are sampled from a $p-$variate WN with null mean vector $\bmu = {\bf 0}$ and variance-covariance matrix $\bSigma=D^{1/2}RD^{1/2}$, where $R$ is a random correlation matrix with condition number set equal to $20$ and $D=\sigma\mathrm{I}_p$. 
We set the sample size $n=250$, number of dimensions $p=2,5$, $\sigma=\pi/8, \pi/4$.
Incomplete data with missing values are obtained using the function {\tt ampute} available from the {\tt R} package {\tt mice} \citep{mice}. The
rate of missing values is set equal to $\epsilon_{NA}=0.1, 0.2, 0.3, 0.4$. The fractions refer to the number of rows with NA entries or NA cells. Both missing data mechanisms, MCAR and MAR are considered. In the situation with $p=5$, four missingness patterns are allowed with equal probabilities, including one to four missing values; the first variable is supposed to be completely observable.  Moreover, the missingness patterns under MAR only depend on the values of the first variable.
The situation with $40\%$ of cellwise missing values has not been considered since it leads to a very few number of complete rows that badly compromise initialization.
Fitting accuracy is evaluated according to
\begin{itemize}
\item[(i)] the chord circular distance
\begin{equation*}
\Delta(\hat{\bmu}, \bmu_0) =\sum_{h=1}^p\left(\sqrt{2 (1 - \cos(\hat\mu_h-\mu_{0h}))}\right)
\end{equation*}
\item[(ii)] the divergence
\begin{equation*}
\Delta(\hat{\bSigma}, \bSigma) = \textrm{trace}(\hat{\bSigma}\bSigma^{-1})-\log(\textrm{det}(\hat{\bSigma}\bSigma^{-1}))-p.
\end{equation*}
\end{itemize}
averaged over $N = 500$ Monte Carlo trials. All the Figures are given in the Appendix. 
Across the considered simulation settings, the proposed estimator generally exhibits a better performance than the benchmark procedures, with the largest gains observed for covariance estimation under cellwise missingness.
Similar performances are observed under both MCAR and MAR mechanisms and both casewise and cellwise missingness rates. The effect of missing values is enhanced under MCAR and cellwise missingness rates.
Missing values largely affect variance-covariance estimation: deleting cases with missing entries or imputing them with marginal means can be very misleading. Moreover, under MAR also mean estimation can be particularly negatively affected. 
On the other hand, the proposed EM algorithm always performs satisfactorily and improves over the other methods in all considered settings, leading to reliable fitting when compared with the MLE based on complete data. 

\section{Empirical application: protein data}
\label{sec:applications}

We consider a subset of the TIM8 protein data available from the R package {\tt BAMBI} \citep{bambi}, which comprises $n = 241$ pairs of backbone dihedral angles for the enzyme Triose Phosphate Isomerase. A proportion of $30\%$ rows with missing values was introduced under both the MCAR and MAR mechanism. We compare the MLE obtained from the original complete data with those derived from the proposed EM algorithm, the MLE obtained after discarding incomplete cases and the MLE based on imputed data using marginal circular means.

Figure \ref{fig:a} displays the original data with missing values highlighted as filled black circles, presented in the form of a Ramachandran plot over the torus $[-\pi, \pi) \times [-\pi, \pi)$. Superimposed $95\%$ tolerance ellipses facilitate visual comparison among the fitted models, particularly with respect to the complete-data MLE. Further insights into the comparative performance, effectiveness, and superiority of the proposed EM-based MLE 
 are provided in Table \ref{tab:a}, which reports the chord circular distance $\Delta(\hat\bmu,\hat\bmu_{MLE})$ between the fitted mean vectors obtained under each incomplete-data approach and the complete-data MLE, along with the divergence between the corresponding fitted covariance matrices $\Delta(\hat\bSigma, \hat\bSigma_{MLE})$, averaged over 200 replications, along with standard errors. 

Figure \ref{fig:b} illustrates the imputed data obtained using both conditional (unwrapped) means and stochastic imputation. The results indicate the extent to which random imputation preserves the original data structure more faithfully than conditional mean imputation.

\begin{figure}
\centering
\includegraphics[scale=0.35]{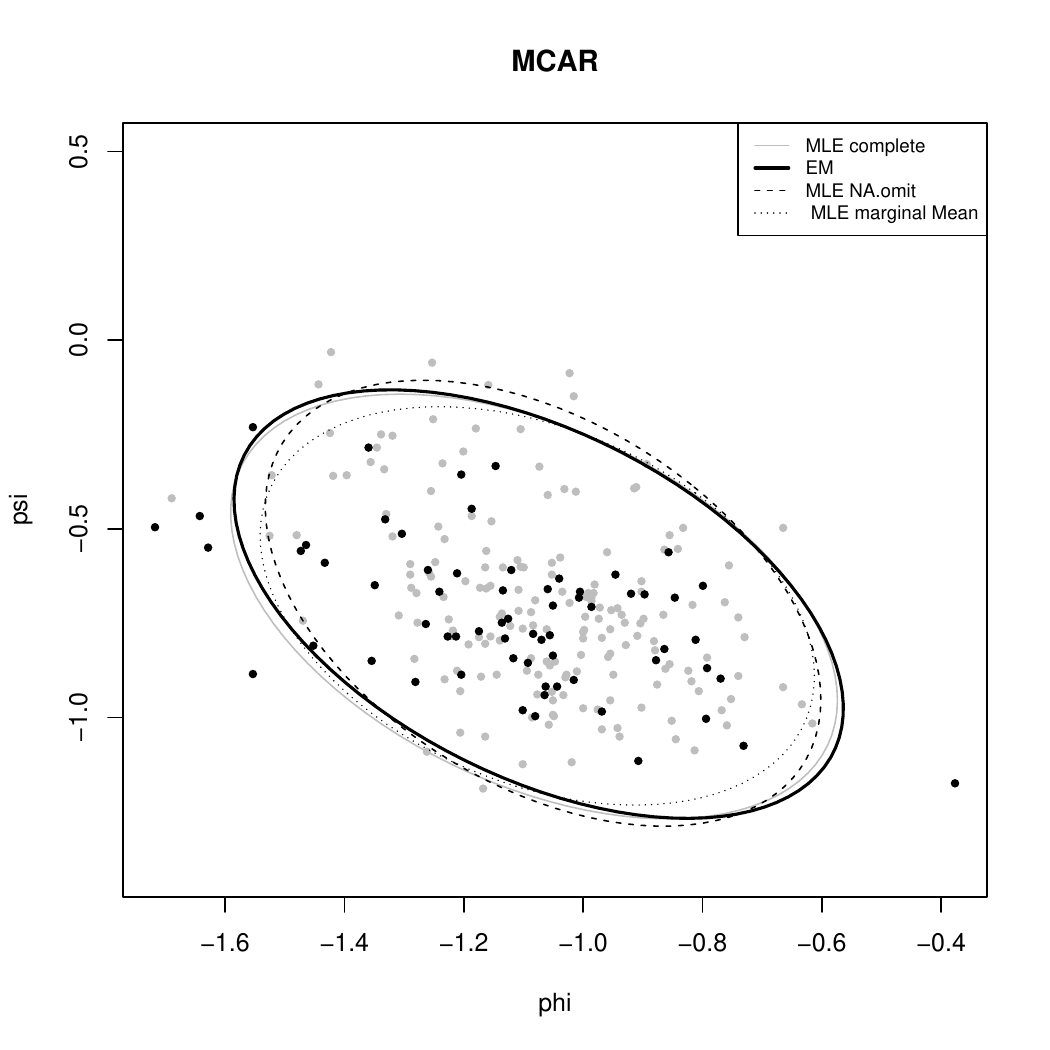}
\includegraphics[scale=0.35]{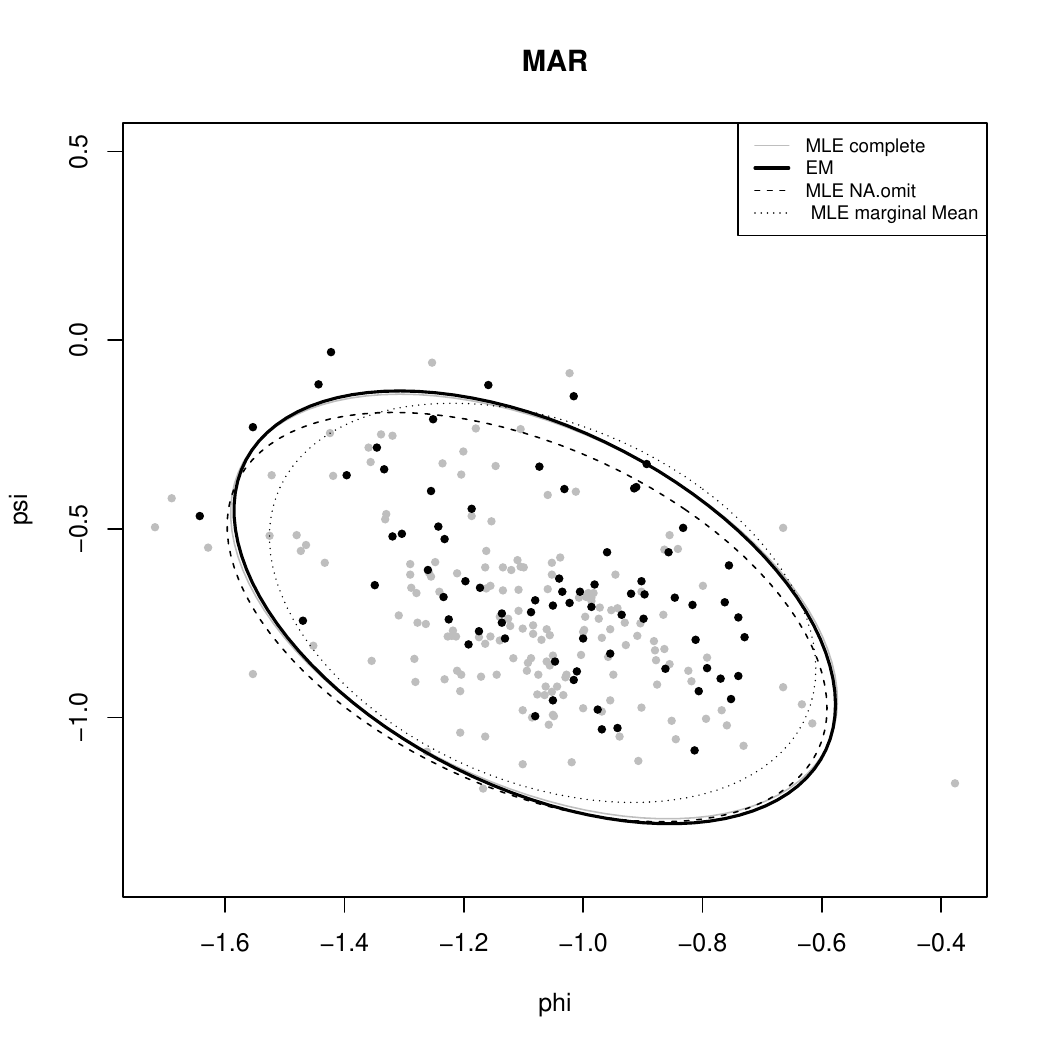}
\caption{TIM8 data. Original and missing values (in black) with $95\%$ tolerance ellipses overimposed from different methods. }
\label{fig:a}
\end{figure}

\begin{figure}
\centering
\includegraphics[scale=0.325]{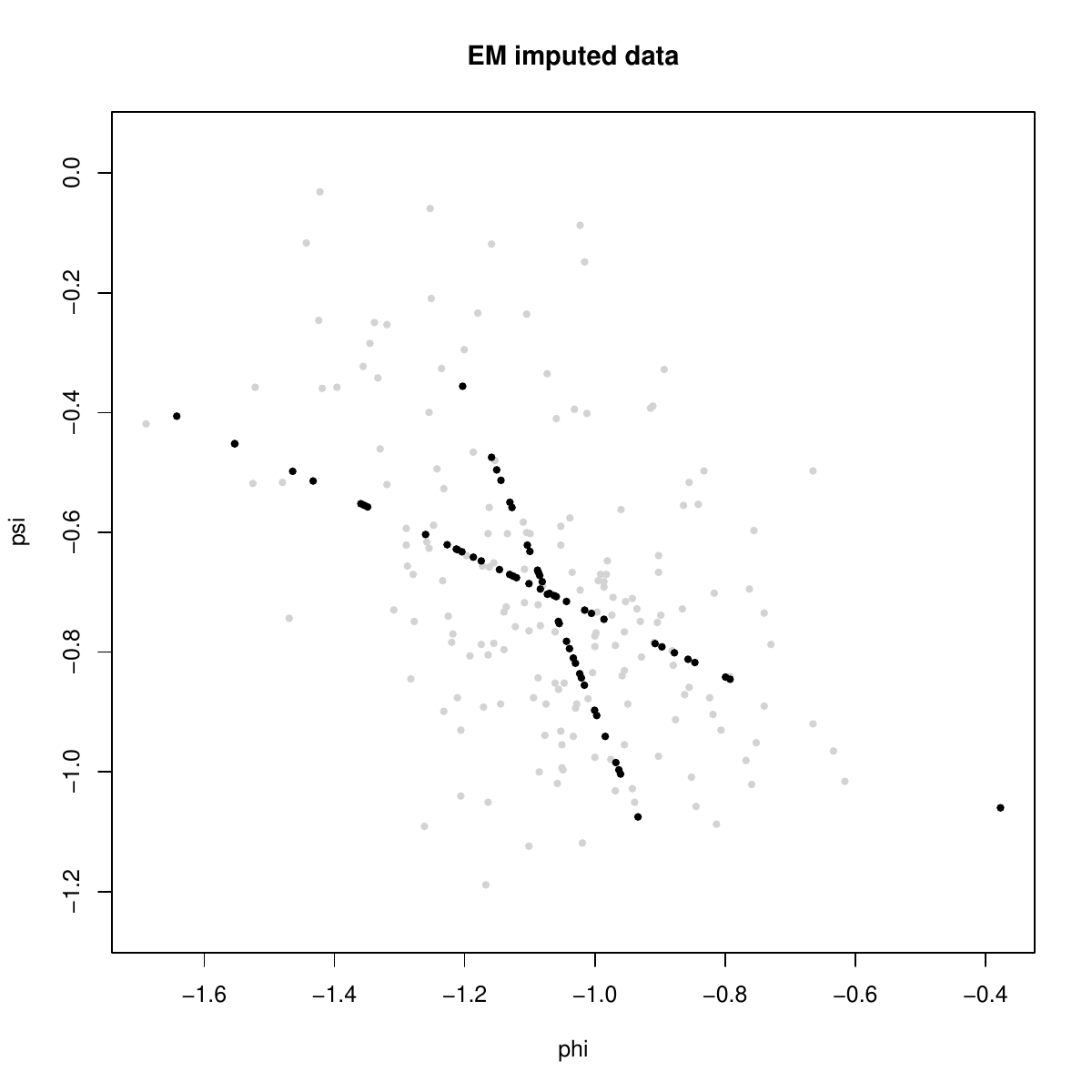}
\includegraphics[scale=0.325]{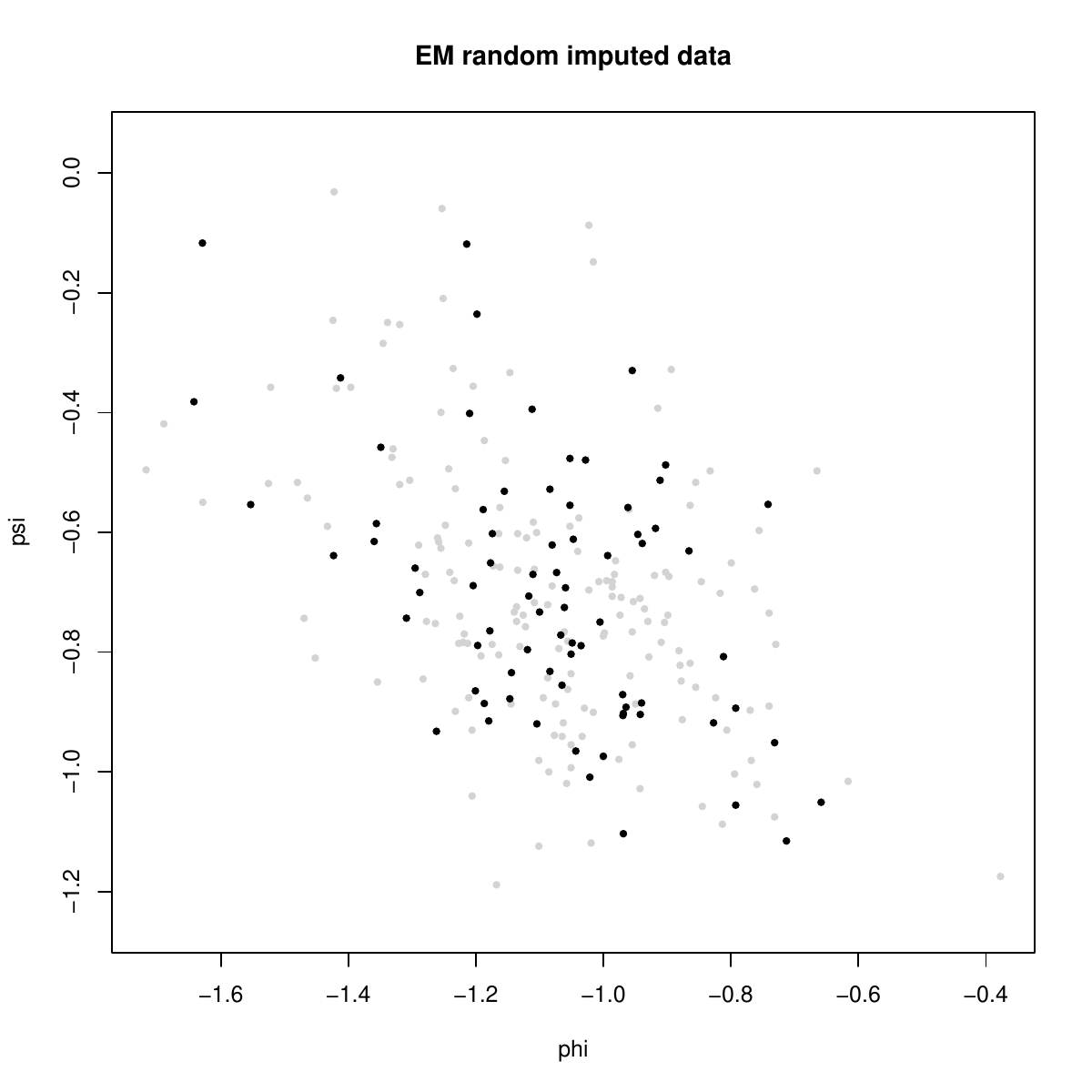}\\
\includegraphics[scale=0.325]{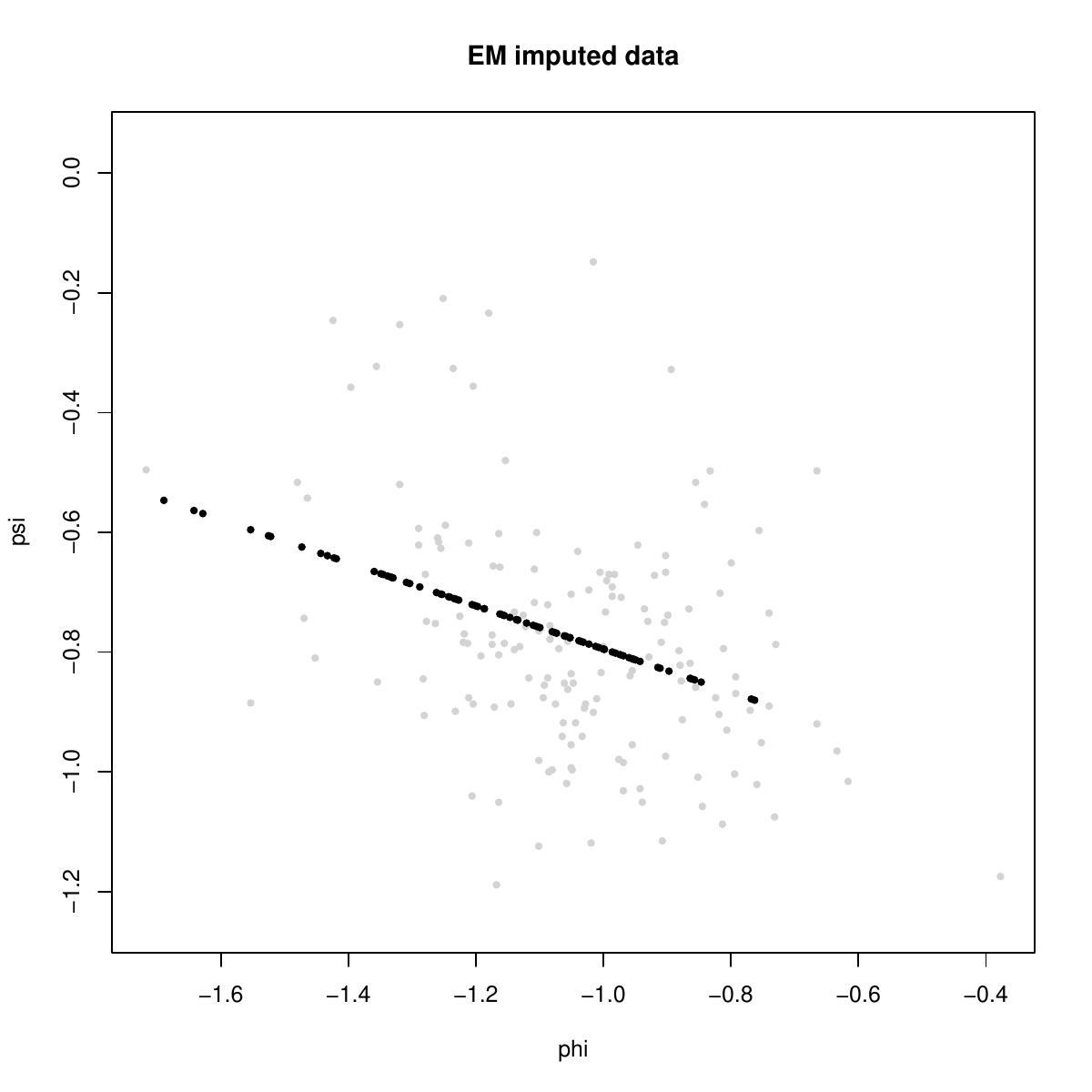}
\includegraphics[scale=0.325]{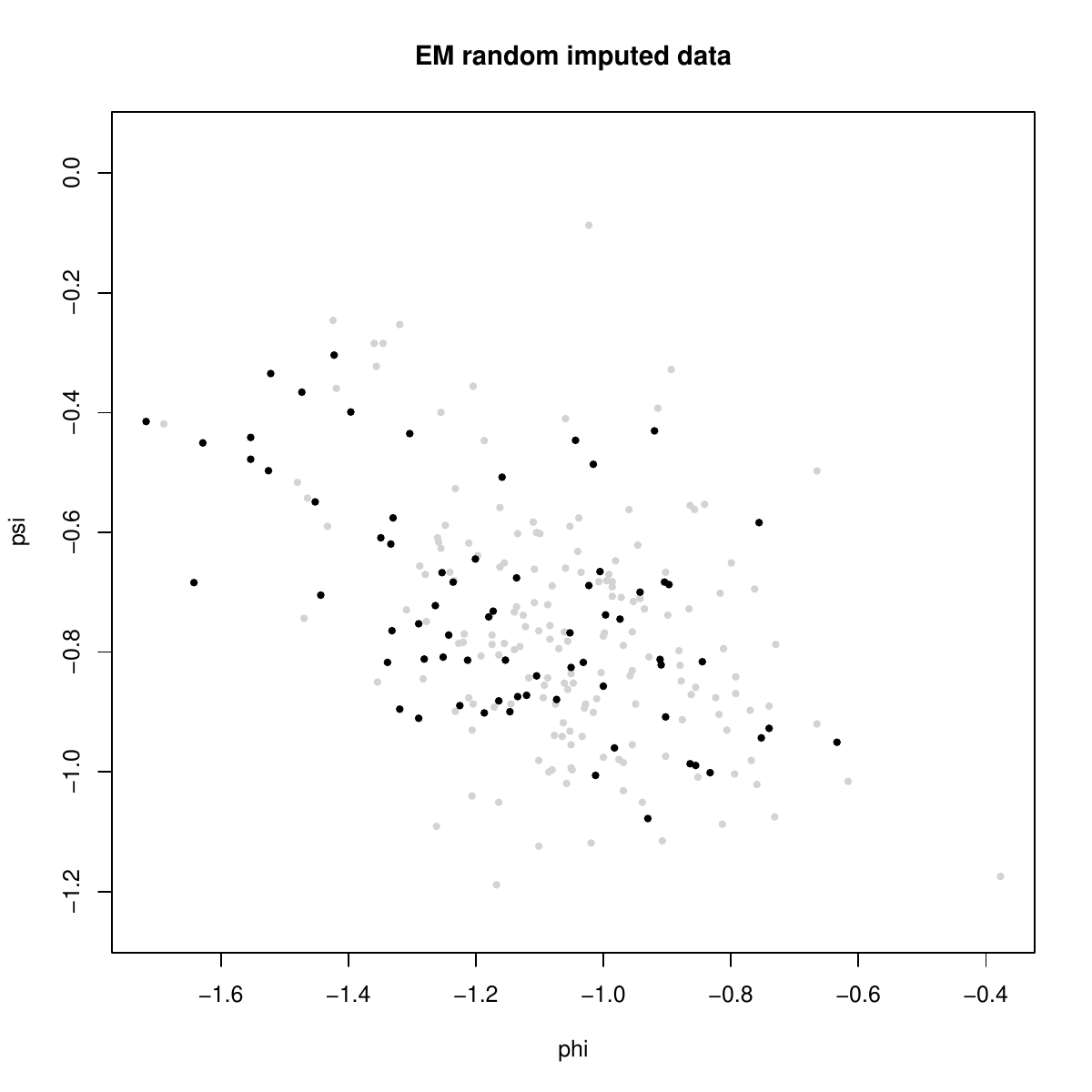}\\
\caption{TIM8 data. Imputed data using conditional means (left) or random draws (right) under MCAR (top panels) and MAR (bottom panels).}
\label{fig:b}
\end{figure}

\begin{table}
\centering
\begin{tabular}{r|cc|cc}
& \multicolumn{2}{c}{MCAR}&\multicolumn{2}{c}{MAR}\\
Method & $\Delta(\hat\bmu,\hat\bmu_{MLE})$ & $\Delta(\hat\bSigma, \hat\bSigma_{MLE})$& $\Delta(\hat\bmu,\hat\bmu_{MLE})$ & $\Delta(\hat\bSigma, \hat\bSigma_{MLE})$\\
\hline
EM     & 0.0064 & 0.0035 & 0.0529 & 0.0551 \\
       &(0.0049)&(0.0029)&(0.0086)&(0.0277)\\
NA.omit& 0.0119 & 0.0061 & 0.0870 & 0.0590 \\
       &(0.0093)&(0.0053)&(0.0143)&(0.0279)\\
Marg. mean& 0.0070 & 0.0277 & 0.0626 & 0.2037 \\
          &(0.0050)&(0.0109)&(0.0091)&(0.0543)\\
\hline
\end{tabular}
\caption{TIM8 data. Comparison between incomplete and complete MLE for different methods and missing mechanisms. Values averaged over 200 replications (with standard errors in brackets).}
\label{tab:a}
\end{table}

\section{Conclusions}
\label{sec:discussion}
In this paper, we introduced a methodological framework for the estimation and imputation of torus data in the presence of ignorable missing values. We developed a general methodology within the Wrapped Elliptically Symmetric family of distributions with completely monotone generators. By leveraging the conditional properties of the elliptically symmetric distributions in the unwrapped space, we successfully embedded the handling of missing responses directly within an EM algorithm. This approach treats both the wrapping coefficients and the missing entries (and scale mixtures) as latent variables simultaneously.
The proposed methodology respects the periodic nature of the $p$-dimensional torus, avoiding the distortions typical of Euclidean methods applied to multivariate circular data. 
The methodology performed satisfactorily both on numerical studies and empirical applications.
One limitation lies in the algorithm becoming computationally demanding in high-dimensional settings: future research may explore different directions.


\appendix
\section{Appendix}
\label{appendix}
\subsection*{Computational details}
Algorithm~\ref{alg:em_wn_na} summarizes the EM algorithm of Section~\ref{sec:methodology} for the WN specification 
The same scheme applies to other members of the WES family, with weights $\hat w_{i\bj_{\text{obs}}}\hat\tau_{i\bj_{\text{obs}}}$ in the sufficient statistics and in the M-step \eqref{eq:M-step-missing}, since $\hat\tau_{i\bj_{\text{obs}}}$ depends on the wrapping coefficients only through the observed components, as in \eqref{eq:tau-hat-missing}. When no observations are missing, the algorithm reduces exactly to the standard EM for complete-data WN estimation.

Before running the algorithm, each observed angle is centred at the corresponding component of the starting location $\bmu^{(0)}$ and mapped to $[-\pi,\pi)$. Then, the algorithm is run on the centred data starting from $\bmu=\bm{0}$, and the final location estimate is shifted back and reduced modulo $2\pi$. In the remainder of this Appendix, $\by_i$ and $\bmu$ denote centered quantities.

The implementation exploits the structure of the missingness patterns. The $N_P$ distinct patterns $\mathcal{P}_1,\dots,\mathcal{P}_{N_P}$ are identified once, together with the sets $\mathcal{C}_{J,m}=\{-J,\dots,J\}^{p_m}$ of wrapping coefficients of the $p_m$ observed components, so that $\mathcal{C}_{J,i}=\mathcal{C}_{J,m}$ for $i\in\mathcal{P}_m$; $n_m$ denotes the number of rows in $\mathcal{P}_m$ and $K_m=(2J+1)^{p_m}$ the cardinality of $\mathcal{C}_{J,m}$. At each iteration and for each pattern, the Cholesky factorization $\bSigma_{\text{obs},\text{obs}}=R_m^\top R_m$ is computed once and shared by the E-step weights and the sufficient statistics. No matrix is explicitly inverted: the regression matrix
\begin{equation*}
\bm{\beta}_m=\bSigma_{\text{mis},\text{obs}}\bSigma_{\text{obs},\text{obs}}^{-1}=\left\{R_m^{-1}\left(R_m^{-\top}\bSigma_{\text{obs},\text{mis}}\right)\right\}^\top
\end{equation*}
is obtained by two triangular solves, through the {\tt backsolve} function in {\tt R} \citep{erre}, which is numerically more stable than explicit inversion.

The weights \eqref{eq:posterior_weights_obs} are evaluated on the $K_m$ configurations in $\mathcal{C}_{J,m}$ rather than on the $(2J+1)^p$ elements of $\mathcal{C}_J$. Writing $\bm{u}_i=R_m^{-\top}(\by_{i,\text{obs}}-\bmu_{\text{obs}})$ and $\bm{v}_{\bj}=R_m^{-\top}\bj_{\text{obs}}$, the squared Mahalanobis distance of $\by_{i,\text{obs}}+2\pi\bj_{\text{obs}}$ is $\|\bm{u}_i\|^2+4\pi\,\bm{u}_i^\top\bm{v}_{\bj}+4\pi^2\|\bm{v}_{\bj}\|^2$, so that the $n_m\times K_m$ log-densities of a pattern require a single matrix product. Weights are normalized on the log scale, values below $10^{-20}$ are set to zero, and the observed-data log-likelihood \eqref{eq:observed_likelihood_wn} is obtained from the same normalizing constants at no additional cost.

At the current estimates, the conditional mean of the completed unwrapped vector is affine in the observed wrapping coefficients: for $i\in\mathcal{P}_m$,
\begin{equation*}
\hat\bx_{i\bj_{\text{obs}}}=\bm{a}_i+2\pi G_m\bj_{\text{obs}},
\quad
\bm{a}_i=\begin{pmatrix}\by_{i,\text{obs}}\\ \bmu_{\text{mis}}+\bm{\beta}_m(\by_{i,\text{obs}}-\bmu_{\text{obs}})\end{pmatrix},
\quad
G_m=\begin{pmatrix} I_{p_m}\\ \bm{\beta}_m\end{pmatrix},
\end{equation*}
where, as in Section~\ref{sec:em_normal}, the components are arranged with the observed ones first, up to a permutation, and $\bm{a}_i$, $G_m$ and $\bm{\beta}_m$ are evaluated at the current estimates, superscripts being omitted. In the WN case, the conditional second moment is $\hat\bC_{i\bj_{\text{obs}}}=\hat\bx_{i\bj_{\text{obs}}}\hat\bx_{i\bj_{\text{obs}}}^\top+\bSigma^{0}_{\text{mis}\mid\text{obs}}$, where $\bSigma^{0}_{\text{mis}\mid\text{obs}}$ is the $p\times p$ matrix equal to $\bSigma_{\text{mis}\mid\text{obs}}$ in the block of the missing components and zero elsewhere, and is therefore quadratic in $\bj_{\text{obs}}$. Hence, the sums over $\mathcal{C}_{J,m}$ only involve the first two moments of the observed wrapping coefficients, $\bar\bj_i=\sum_{\bj_{\text{obs}}}\hat w_{i\bj_{\text{obs}}}\bj_{\text{obs}}$ and $V_m=\sum_{i\in\mathcal{P}_m}\sum_{\bj_{\text{obs}}}\hat w_{i\bj_{\text{obs}}}\bj_{\text{obs}}\bj_{\text{obs}}^\top$, the former entering $\hat T_1$ and both entering $\hat T_2$, and the contributions of pattern $\mathcal{P}_m$ to $\hat T_1$ and $\hat T_2$ are
\begin{align}
&\sum_{i\in\mathcal{P}_m}\left(\bm{a}_i+2\pi G_m\bar\bj_i\right), \label{eq:app_T1}\\
&\sum_{i\in\mathcal{P}_m}\left[\bm{a}_i\bm{a}_i^\top+2\pi\left(\bm{a}_i\bar\bj_i^\top G_m^\top+G_m\bar\bj_i\bm{a}_i^\top\right)\right]+4\pi^2 G_mV_mG_m^\top+n_m\,\bSigma^{0}_{\text{mis}\mid\text{obs}}, \label{eq:app_T2}
\end{align}
respectively. For complete rows $\bm{a}_i=\by_i$ and $G_m$ is the identity matrix, whereas rows with no observed component contribute $n_m\bmu$ and $n_m(\bSigma+\bmu\bmu^\top)$.

To ensure numerical stability when dealing with sparse patterns or localized collinearity, the updated covariance matrix is symmetrized and its eigenvalues $\lambda_h$ are replaced by $\max(\lambda_h,\kappa)$, with $\kappa=10^{-8}$, so that it remains strictly positive definite. The algorithm stops when
\begin{equation*}
\max\left\{\max_{h}\sqrt{2\big(1-\cos(\mu_h^{(t+1)}-\mu_h^{(t)})\big)},\ \max_{h,k}\big|\Sigma_{hk}^{(t+1)}-\Sigma_{hk}^{(t)}\big|\right\}<\varepsilon
\end{equation*}
or after a maximum number of iterations; the observed-data log-likelihood, available at each iteration, is monitored to check the ascent property.

The computational cost per iteration is $O\big(N_P\,p^3+\sum_{m}n_m K_m p_m\big)$, where the first term accounts for the pattern-specific factorizations and the second for the log-densities and the moments of the wrapping coefficients; lower-order terms are $O\big(\sum_m (K_m p_m^2+n_m p^2)\big)$, and the memory required by the weights is $O\big(\sum_m n_m K_m\big)$. The exponential dependence is thus on the number of observed components of each pattern rather than on $p$: for incomplete rows, the cost is reduced by a factor $(2J+1)^{p-p_m}$ with respect to the evaluation over $\mathcal{C}_J$, whereas complete rows have the same cost as in complete-data estimation.

At convergence, the weights $\hat w^{(\infty)}_{i\bj_{\text{obs}}}$ are available from the last E-step. Since $\hat\bx_{i\bj_{\text{obs}},\text{mis}}$ is affine in $\bj_{\text{obs}}$, the conditional imputation \eqref{eq:imputation} reduces to $\hat\by_{i,\text{mis}}=\bmu^{(0)}_{\text{mis}}+\{\bmu_{\text{mis}}+\bm{\beta}_m(\by_{i,\text{obs}}+2\pi\bar\bj_i-\bmu_{\text{obs}})\}\bmod 2\pi$, evaluated at the final estimates, and only requires $\bar\bj_i$. For stochastic imputation, the cumulative weights $\sum_{l\le k}\hat w^{(\infty)}_{i\bj_l}$ over the configurations $\bj_1,\dots,\bj_{K_m}$ of $\mathcal{C}_{J,m}$ are computed once. Each draw then samples $\bj_{\text{obs}}$ by inversion, which requires a single uniform variate per row, and sets $\bx_{i,\text{mis}}=\bmu_{i,\text{mis}\mid\text{obs};\bj_{\text{obs}}}+\bm{\varepsilon}_i$, with $\bm{\varepsilon}_i\sim N_{r_i}(\bm{0},\bSigma_{\text{mis}\mid\text{obs}})$ generated through the Cholesky factor of $\bSigma_{\text{mis}\mid\text{obs}}$. Since the conditional covariance matrix does not depend on $\bj_{\text{obs}}$ in the WN case, this factor is computed once per pattern. The imputed values are finally shifted back by $\bmu^{(0)}_{\text{mis}}$ and reduced modulo $2\pi$.

\begin{algorithm}[h!]
\caption{Pattern-based EM algorithm for the multivariate WN with missing values. Within the algorithm, $\bj$ denotes the vector of wrapping coefficients of the observed components.}
\label{alg:em_wn_na}
\begin{algorithmic}[1]
\REQUIRE Incomplete data $\by_1,\dots,\by_n$, starting values $\bmu^{(0)},\bSigma^{(0)}$, truncation $J$, tolerance $\varepsilon>0$, maximum number of iterations $t_{\max}$.
\STATE Centre the observed angles at $\bmu^{(0)}$ and map them to $[-\pi,\pi)$; set $\bmu\leftarrow\bm{0}$, $\bSigma\leftarrow\bSigma^{(0)}$, $t\leftarrow 0$.
\STATE Identify the patterns $\mathcal{P}_1,\dots,\mathcal{P}_{N_P}$, their observed and missing indices and the sets $\mathcal{C}_{J,m}$.
\REPEAT
\STATE \textbf{E-step, weights.} \label{alg:estep_start}
\FOR{$m=1,\dots,N_P$ such that $p_m>0$}
\STATE Compute the Cholesky factor $R_m$ of $\bSigma_{\text{obs},\text{obs}}$.
\STATE For $i\in\mathcal{P}_m$ and $\bj\in\mathcal{C}_{J,m}$, compute $s_{i\bj}=\log\phi_{p_m}(\by_{i,\text{obs}}+2\pi\bj;\bmu_{\text{obs}},\bSigma_{\text{obs},\text{obs}})$ by one matrix product, $c_i=\log\sum_{\bk\in\mathcal{C}_{J,m}}\exp(s_{i\bk})$ and $\hat w_{i\bj}=\exp(s_{i\bj}-c_i)$.
\ENDFOR
\STATE $\ell^{(t)}\leftarrow\sum_i c_i$, rows with no observed component contributing zero. \label{alg:estep_end}
\STATE \textbf{E-step, sufficient statistics.} Set $\hat T_1\leftarrow\bm{0}$, $\hat T_2\leftarrow\bm{0}$.
\FOR{$m=1,\dots,N_P$}
\IF{$p_m=0$}
\STATE $\hat T_1\leftarrow\hat T_1+n_m\bmu$, $\hat T_2\leftarrow\hat T_2+n_m(\bSigma+\bmu\bmu^\top)$.
\ELSE
\STATE Compute $\bm{\beta}_m$ and $\bSigma_{\text{mis}\mid\text{obs}}$ from $R_m$ by triangular solves, then $\bm{a}_i$, $i\in\mathcal{P}_m$, and $G_m$.
\STATE Compute $\bar\bj_i=\sum_{\bj}\hat w_{i\bj}\bj$, $i\in\mathcal{P}_m$, and $V_m=\sum_{i\in\mathcal{P}_m}\sum_{\bj}\hat w_{i\bj}\bj\bj^\top$.
\STATE Add the contributions \eqref{eq:app_T1} to $\hat T_1$ and \eqref{eq:app_T2} to $\hat T_2$.
\ENDIF
\ENDFOR
\STATE \textbf{M-step.} $\bmu^{\text{new}}\leftarrow\hat T_1/n$ mapped to $[-\pi,\pi)$, $\bSigma^{\text{new}}\leftarrow\hat T_2/n-\bmu^{\text{new}}\bmu^{\text{new}\top}$; symmetrize $\bSigma^{\text{new}}$ and replace its eigenvalues $\lambda_h$ by $\max(\lambda_h,\kappa)$.
\STATE $\delta\leftarrow\max\big\{\max_h\sqrt{2(1-\cos(\mu^{\text{new}}_h-\mu_h))},\ \max_{h,k}|\Sigma^{\text{new}}_{hk}-\Sigma_{hk}|\big\}$.
\STATE $\bmu\leftarrow\bmu^{\text{new}}$, $\bSigma\leftarrow\bSigma^{\text{new}}$, $t\leftarrow t+1$.
\UNTIL{$\delta<\varepsilon$ or $t=t_{\max}$}
\STATE Repeat lines \ref{alg:estep_start}--\ref{alg:estep_end} at the final values, which gives $\hat w^{(\infty)}_{i\bj}$ and the maximized log-likelihood.
\ENSURE $\hat\bmu=(\bmu+\bmu^{(0)})\bmod 2\pi$, $\hat\bSigma=\bSigma$ and the weights $\hat w^{(\infty)}_{i\bj}$ for imputation.
\end{algorithmic}
\end{algorithm}

\section{Further results from numerical studies}
In the bivariate setting, the results under MCAR are summarized in Figure \ref{fig:1} and Figure \ref{fig:2}, whereas those under MAR are given in Figure \ref{fig:3} and Figure \ref{fig:4}, when the missingness rate concerns the number of rows with missing entries. The results from the scenario with cellwise missingness rates are given in Figure \ref{fig:1a} and Figure \ref{fig:2a} under MCAR, and Figure \ref{fig:3} and Figure \ref{fig:4} under MAR. Figure \ref{fig:pattern_p2} aids the comprehension of the missingness patterns in action when $\epsilon_{NA}=0.4$: it is worth noting the differences between casewise and cellwise rates. 
The results concerning the five dimensional setting are given in Figure \ref{fig:5} and Figure \ref{fig:6}, under MCAR, whereas in Figure \ref{fig:7} and Figure \ref{fig:8} under MAR, when the missingness rate is casewise. Results related to cellwise missingness are given in Figure \ref{fig:5a} and Figure \ref{fig:6a} under MCAR, and Figure \ref{fig:7a} and Figure \ref{fig:8a} under MAR. Figure \ref{fig:pattern_p5} shows the considered missingness in the illustrative situation with $\epsilon_{NA}=0.4$.

\begin{figure}[!h]
	\includegraphics[scale=0.325]{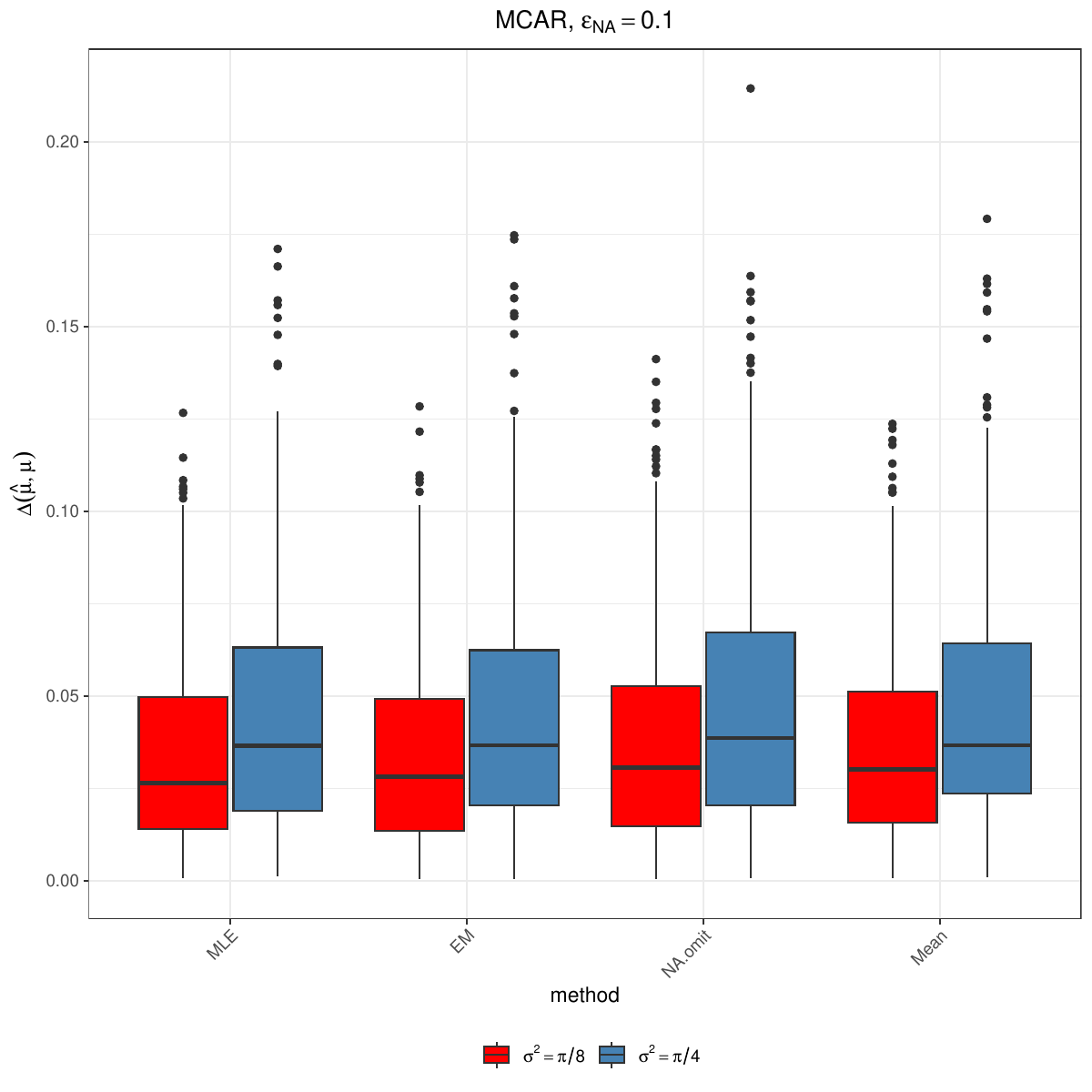}
	\includegraphics[scale=0.325]{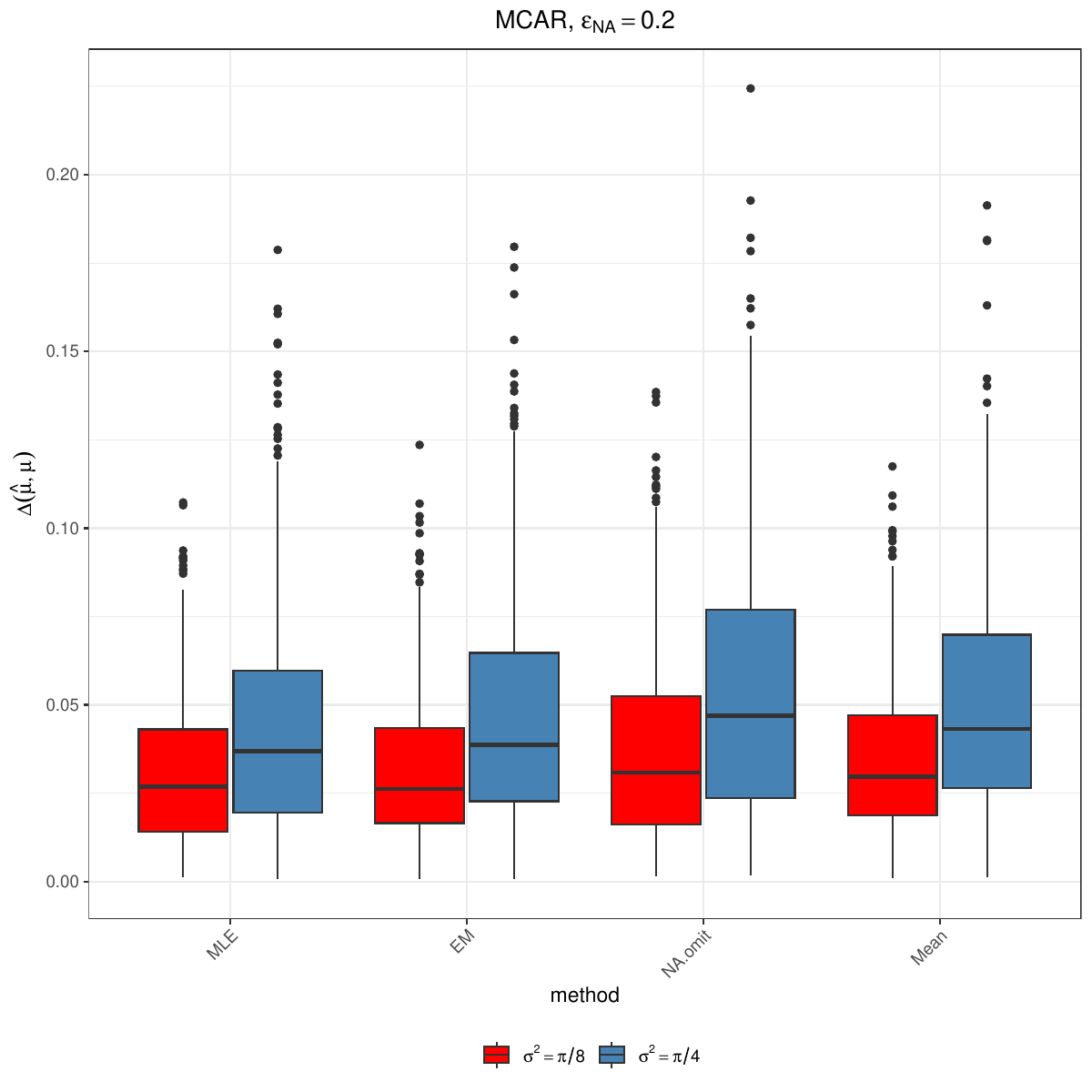}\\
	\includegraphics[scale=0.325]{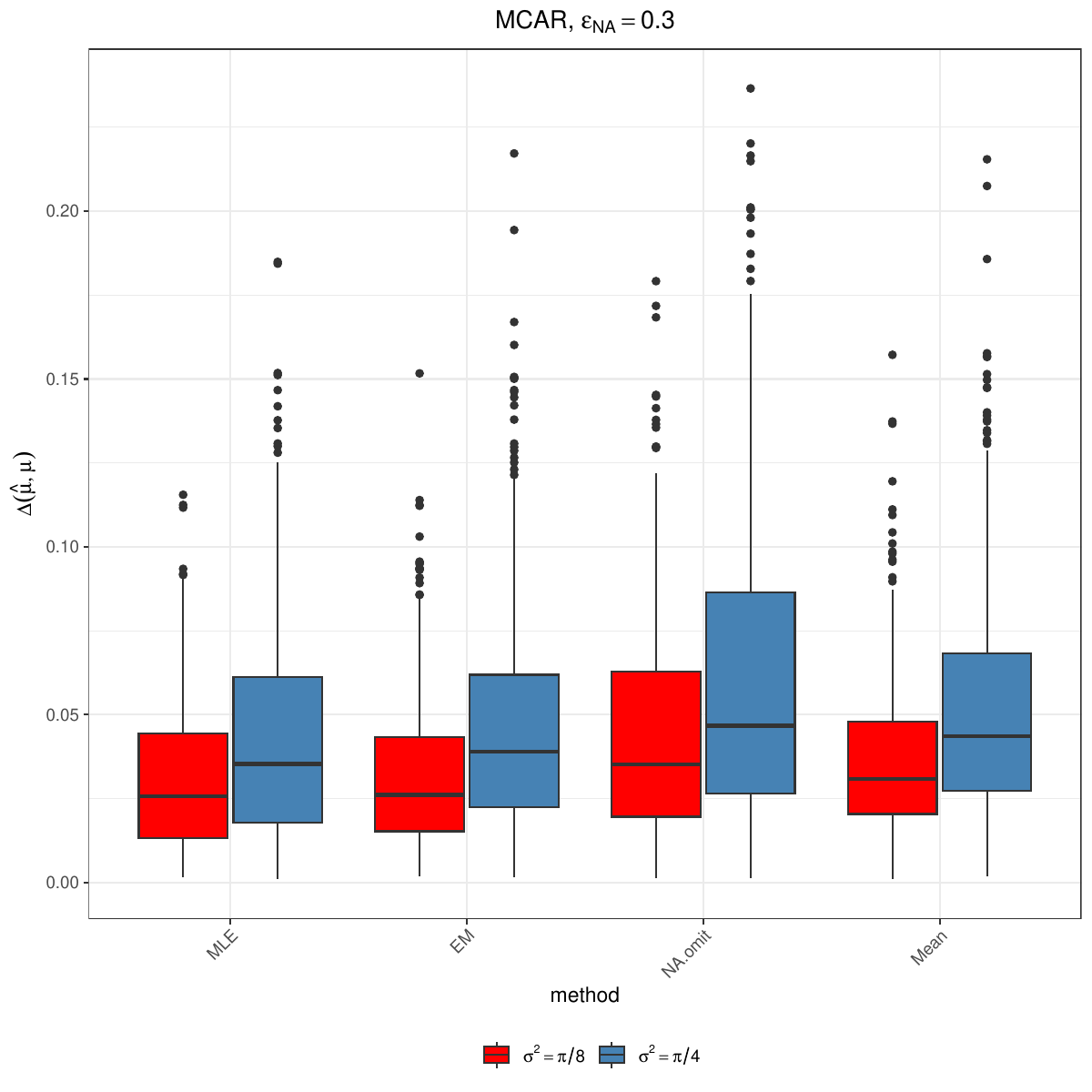}
	\includegraphics[scale=0.325]{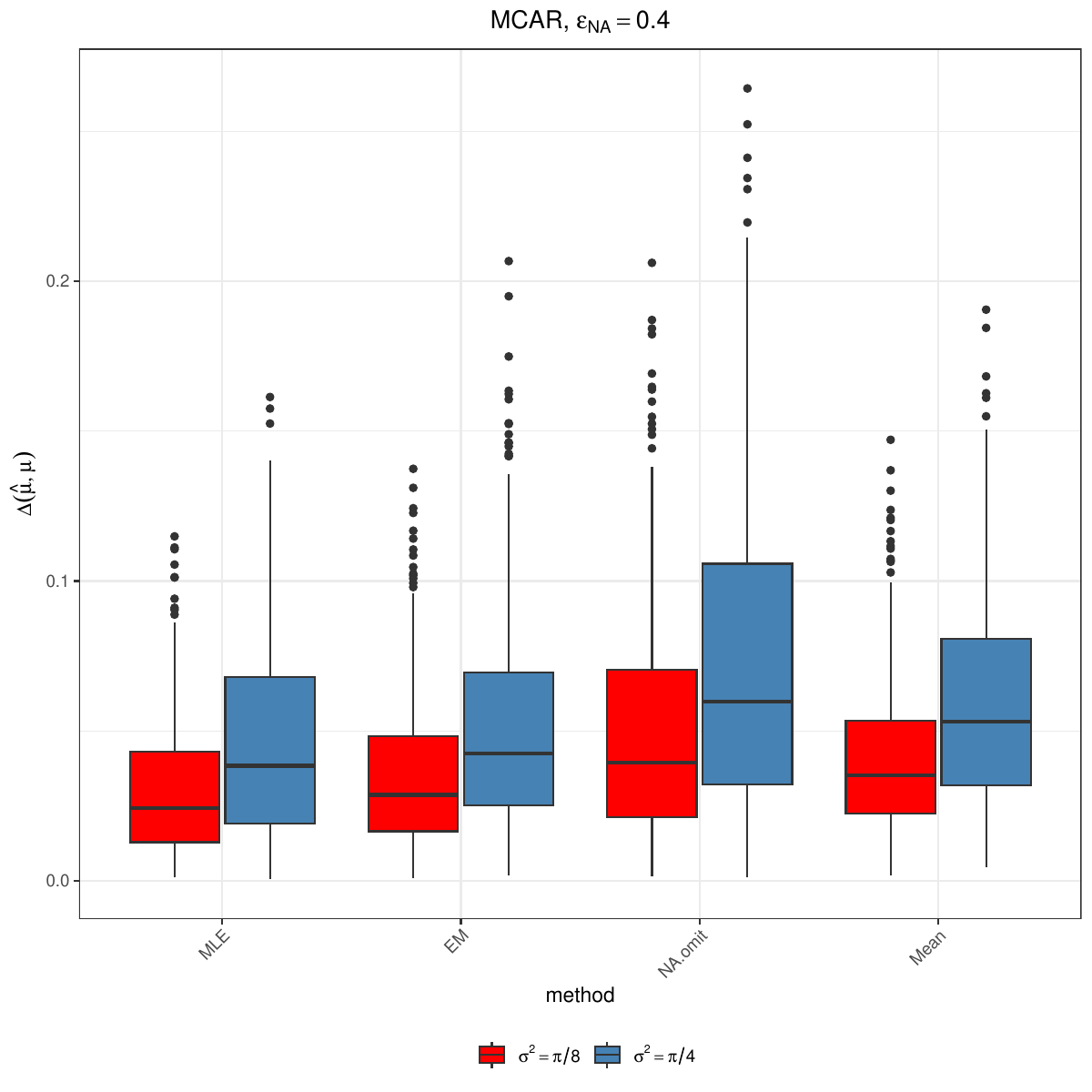}
	\caption{Numerical studies. Distribution of $\Delta(\hat\bmu)$ for different methods uder MCAR when $p=2$. Missingness rate is casewise.}
	\label{fig:1}
\end{figure}

\begin{figure}[!h]
	\includegraphics[scale=0.325]{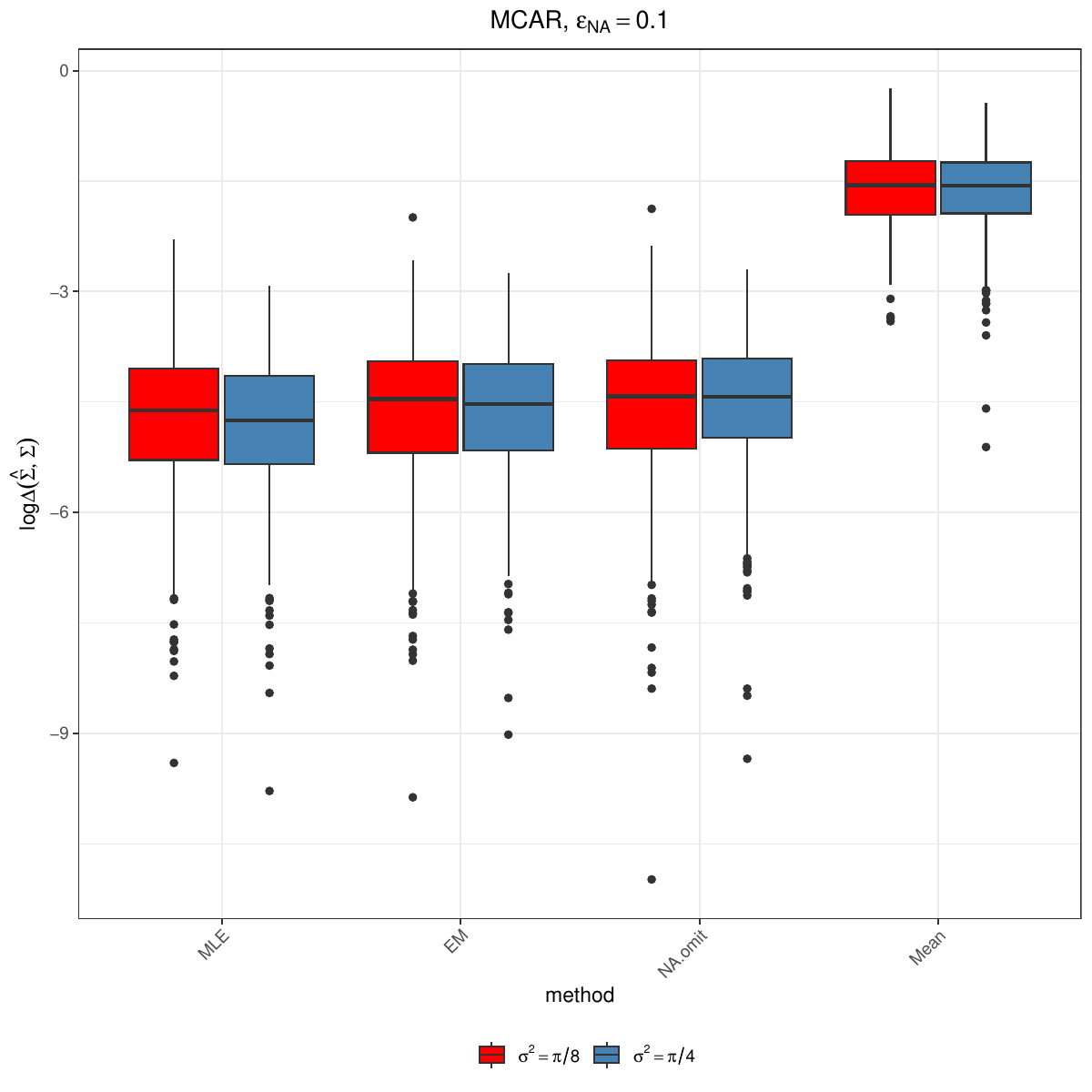}
	\includegraphics[scale=0.325]{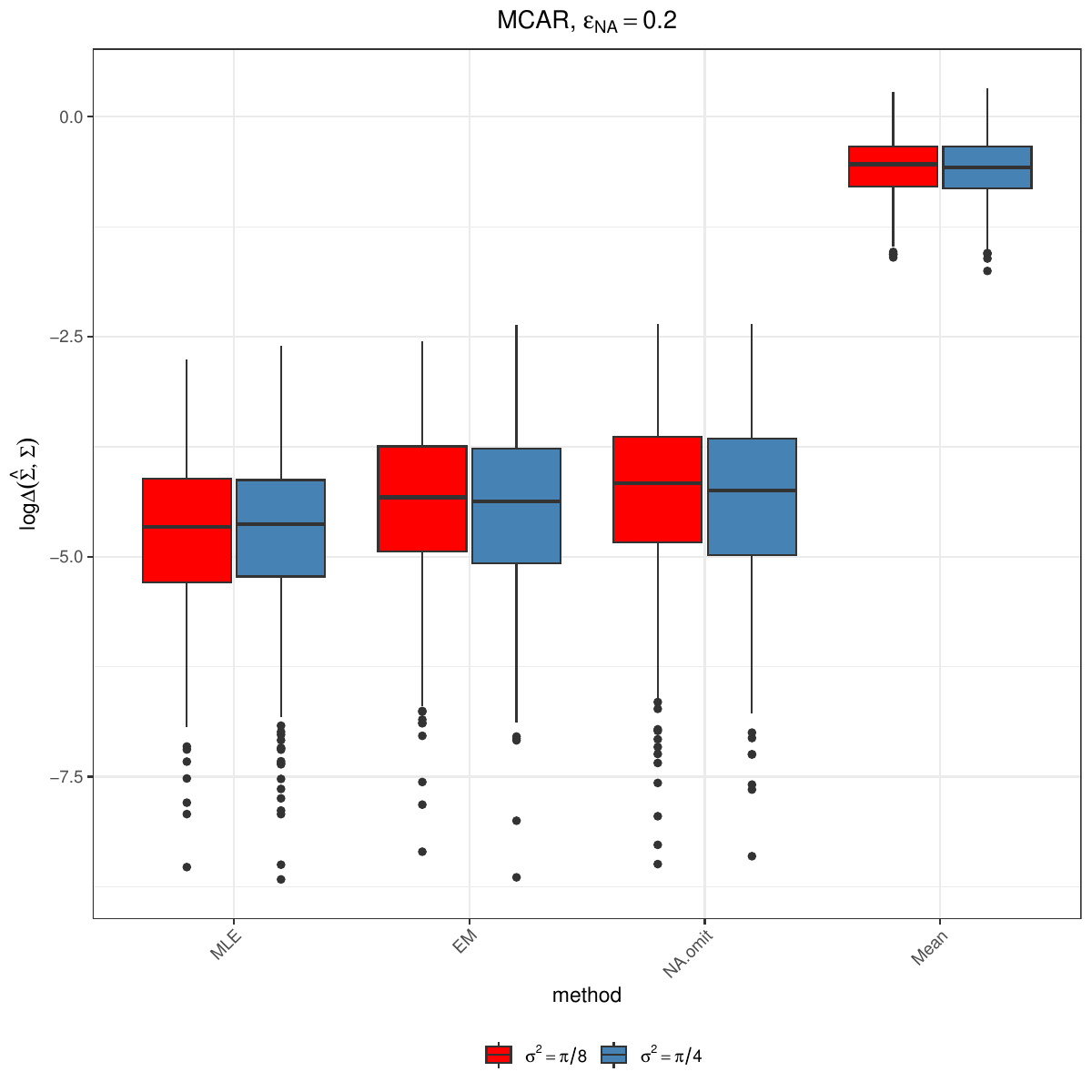}\\
	\includegraphics[scale=0.325]{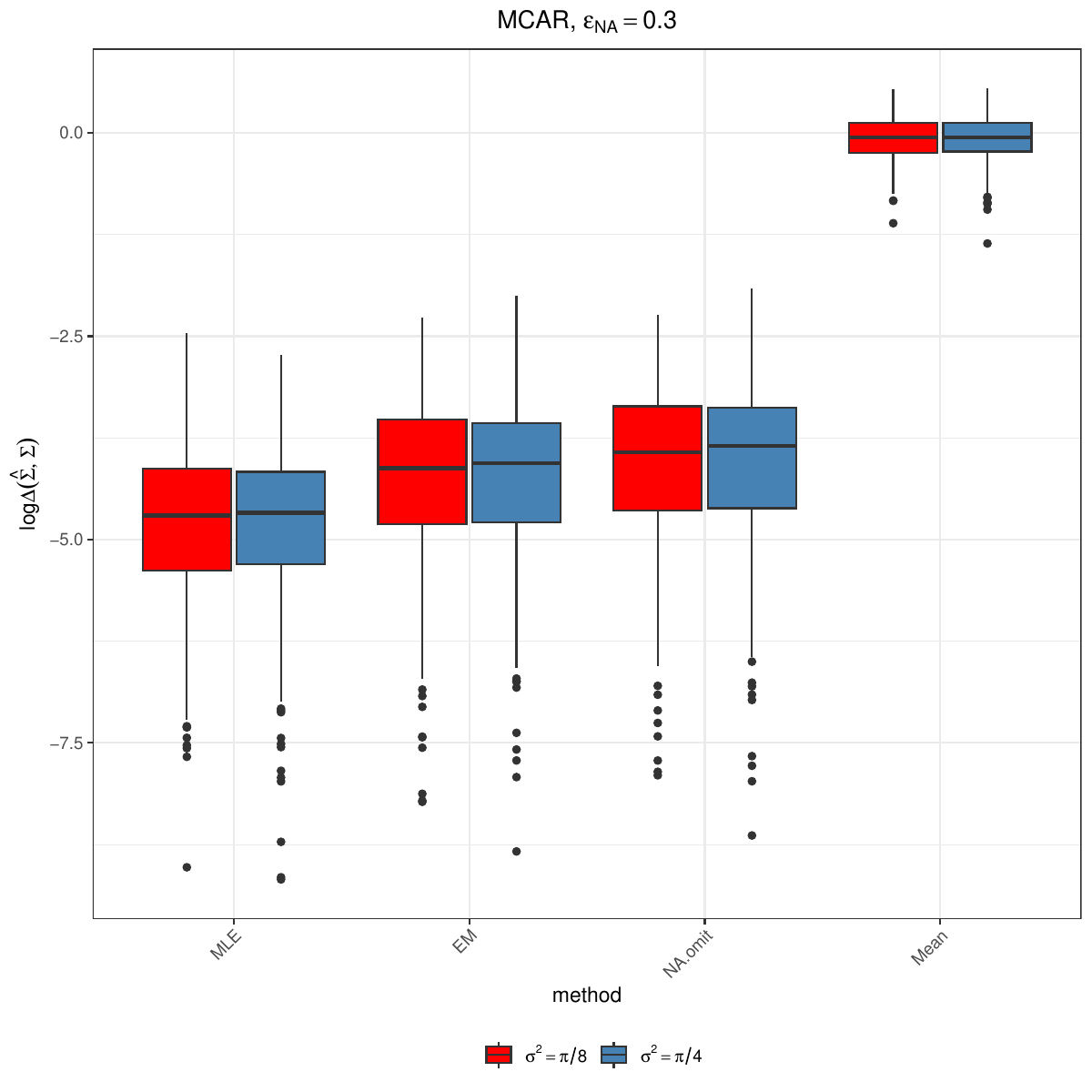}
	\includegraphics[scale=0.325]{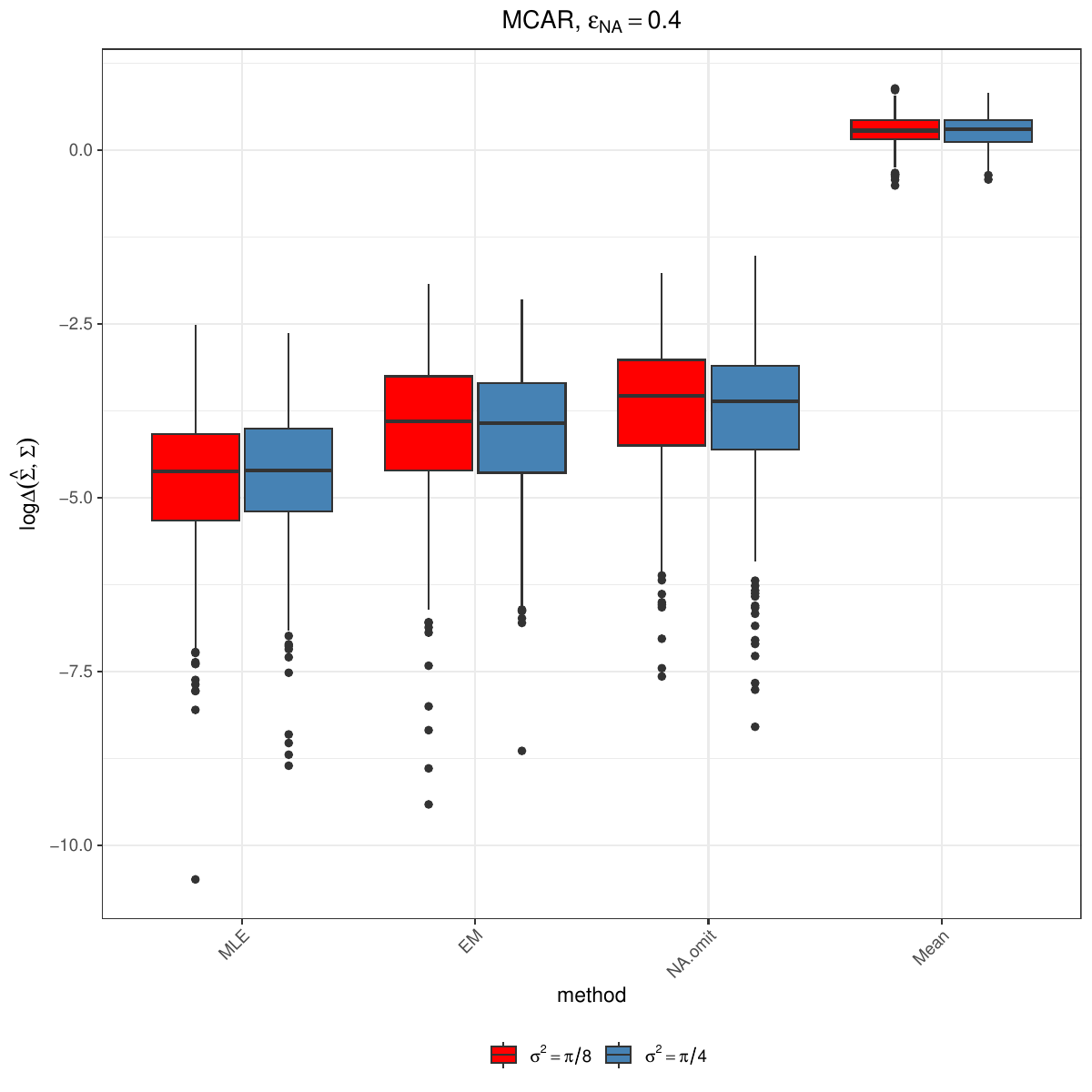}
	\caption{Numerical studies. Distribution of $\log\Delta(\hat\bSigma, \bSigma)$ for different methods uder MCAR when $p=2$. Missingness rate is casewise.}
	\label{fig:2}
\end{figure}

\begin{figure}[!h]
	\includegraphics[scale=0.325]{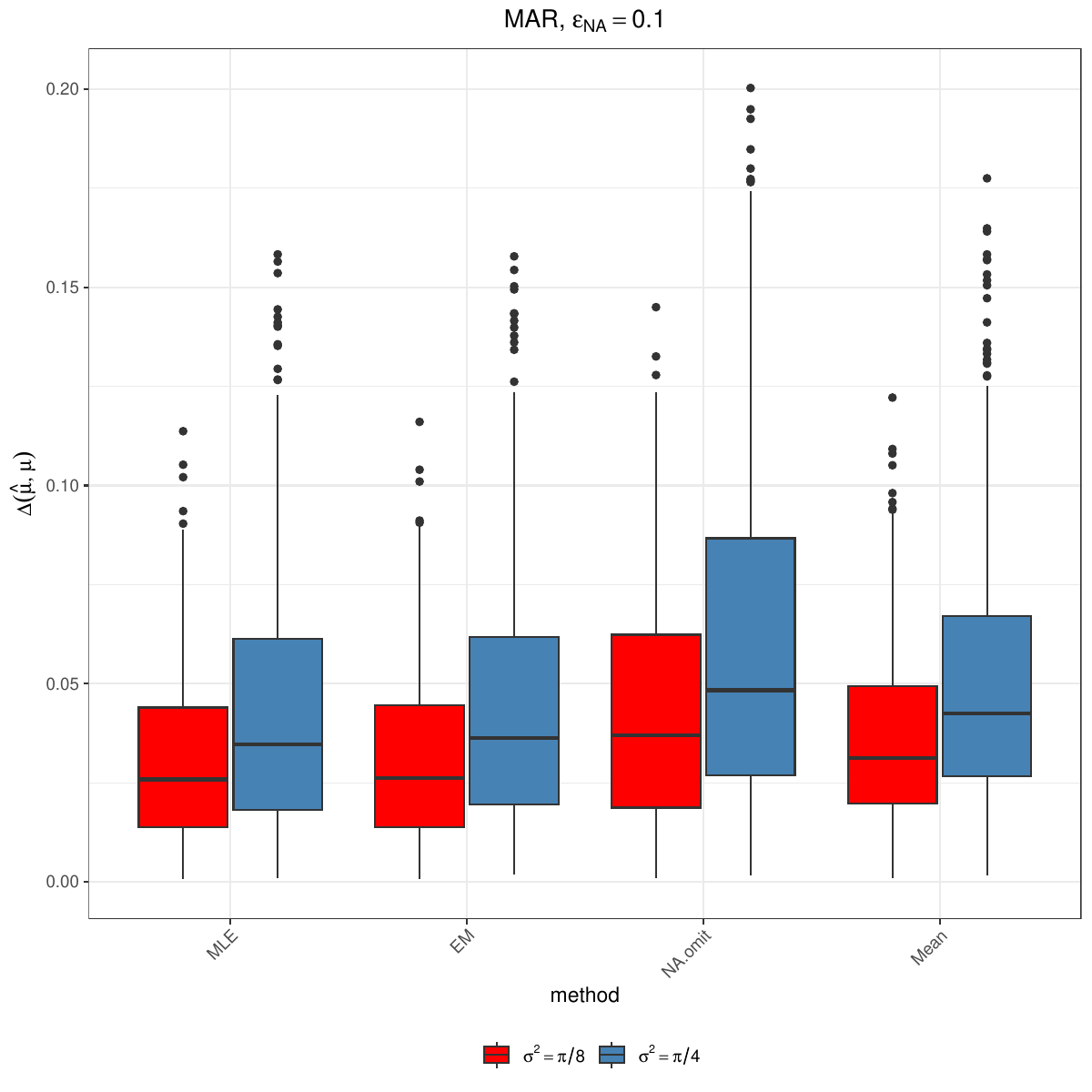}
	\includegraphics[scale=0.325]{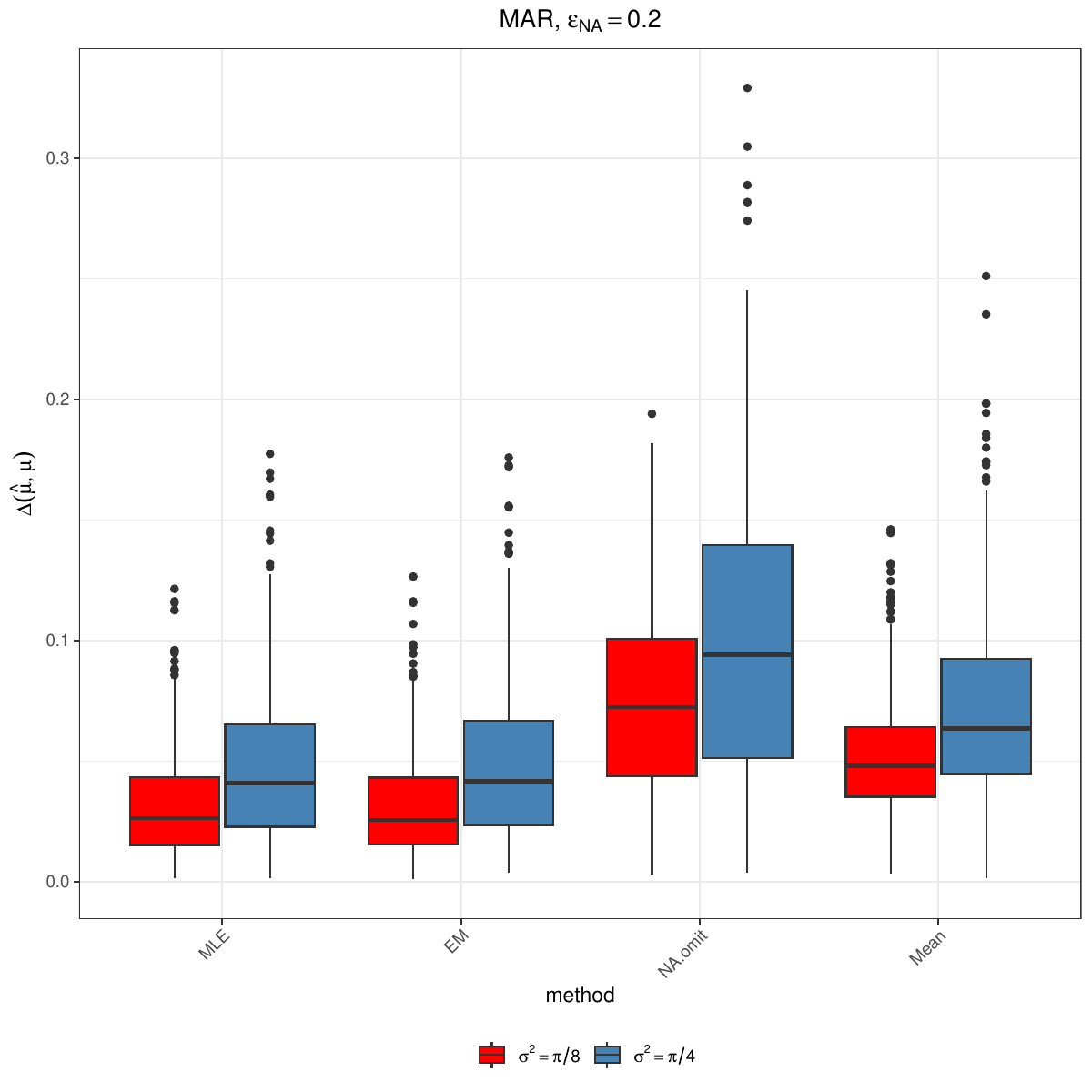}\\
	\includegraphics[scale=0.325]{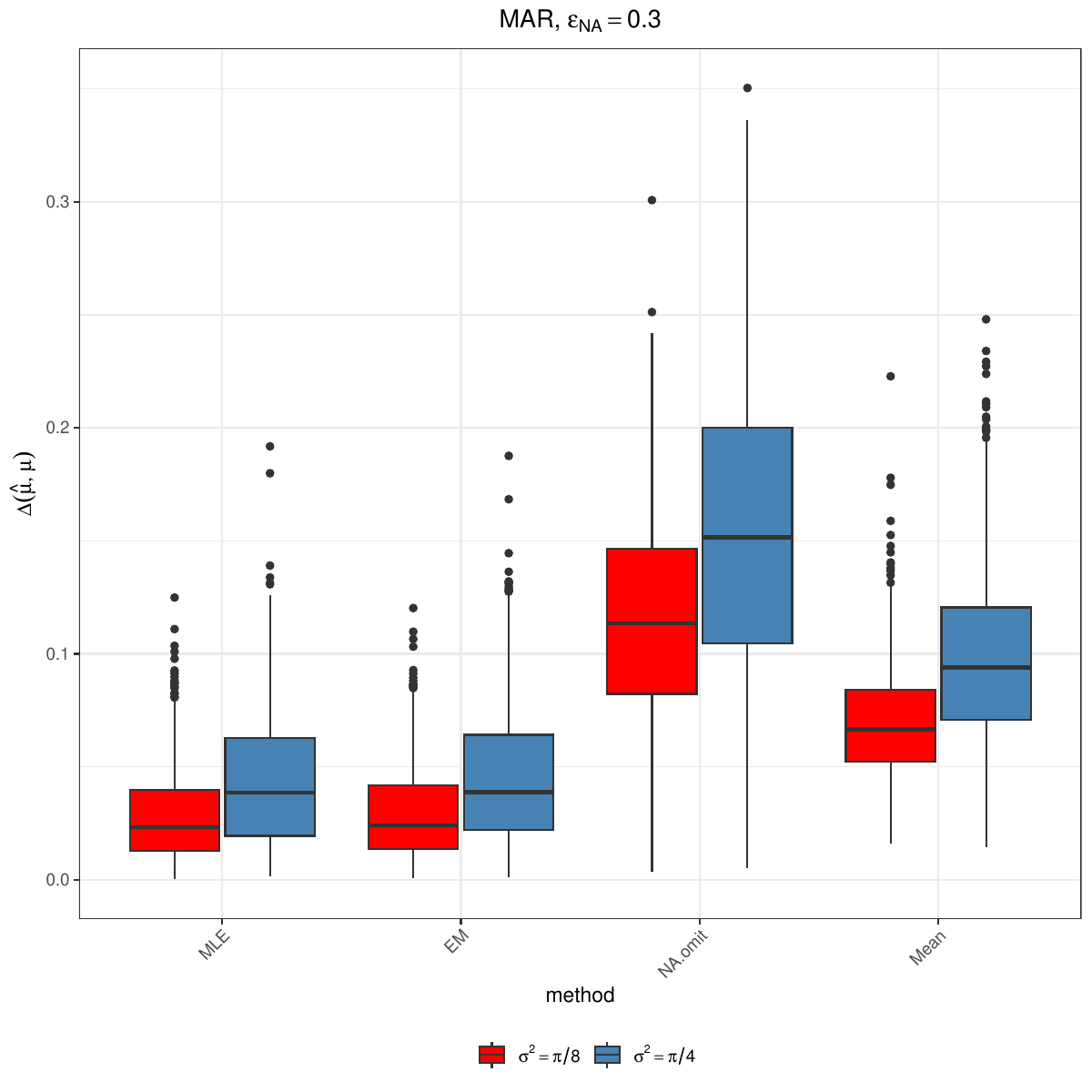}
	\includegraphics[scale=0.325]{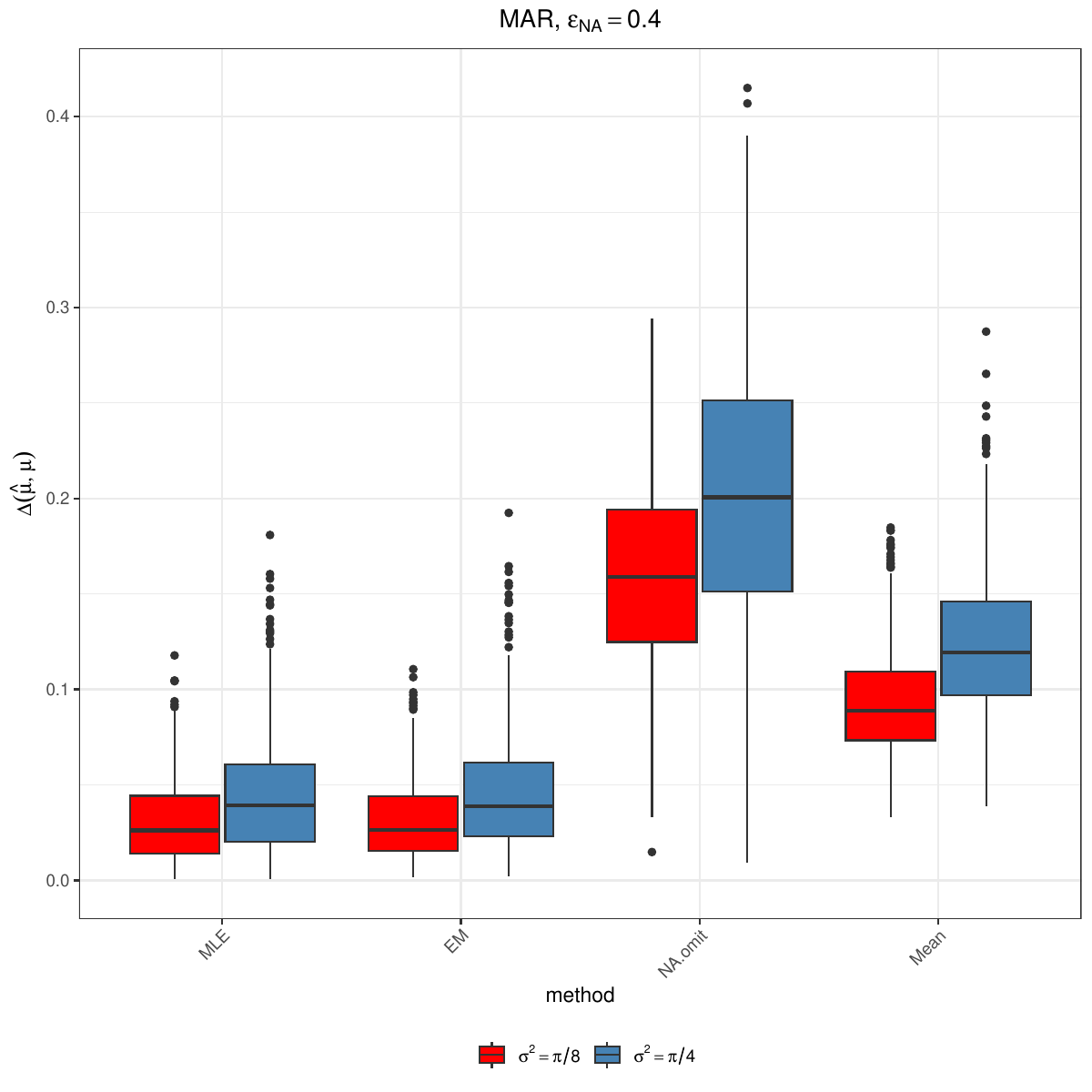}
	\caption{Numerical studies. Distribution of $\Delta(\hat\bmu)$ for different methods uder MAR when $p=2$. Missingness rate is casewise.}
	\label{fig:3}
\end{figure}

\begin{figure}[!h]
	\includegraphics[scale=0.325]{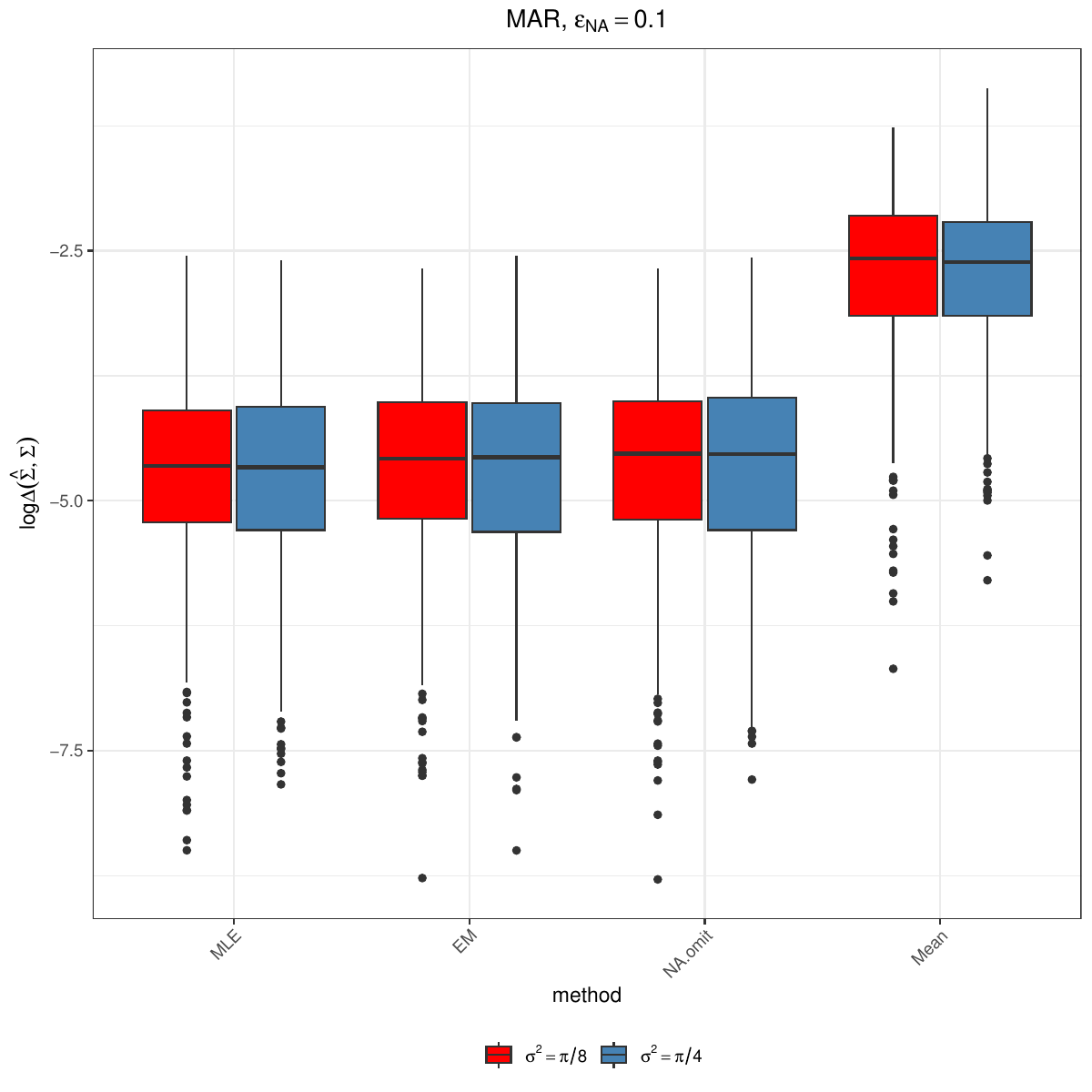}
	\includegraphics[scale=0.325]{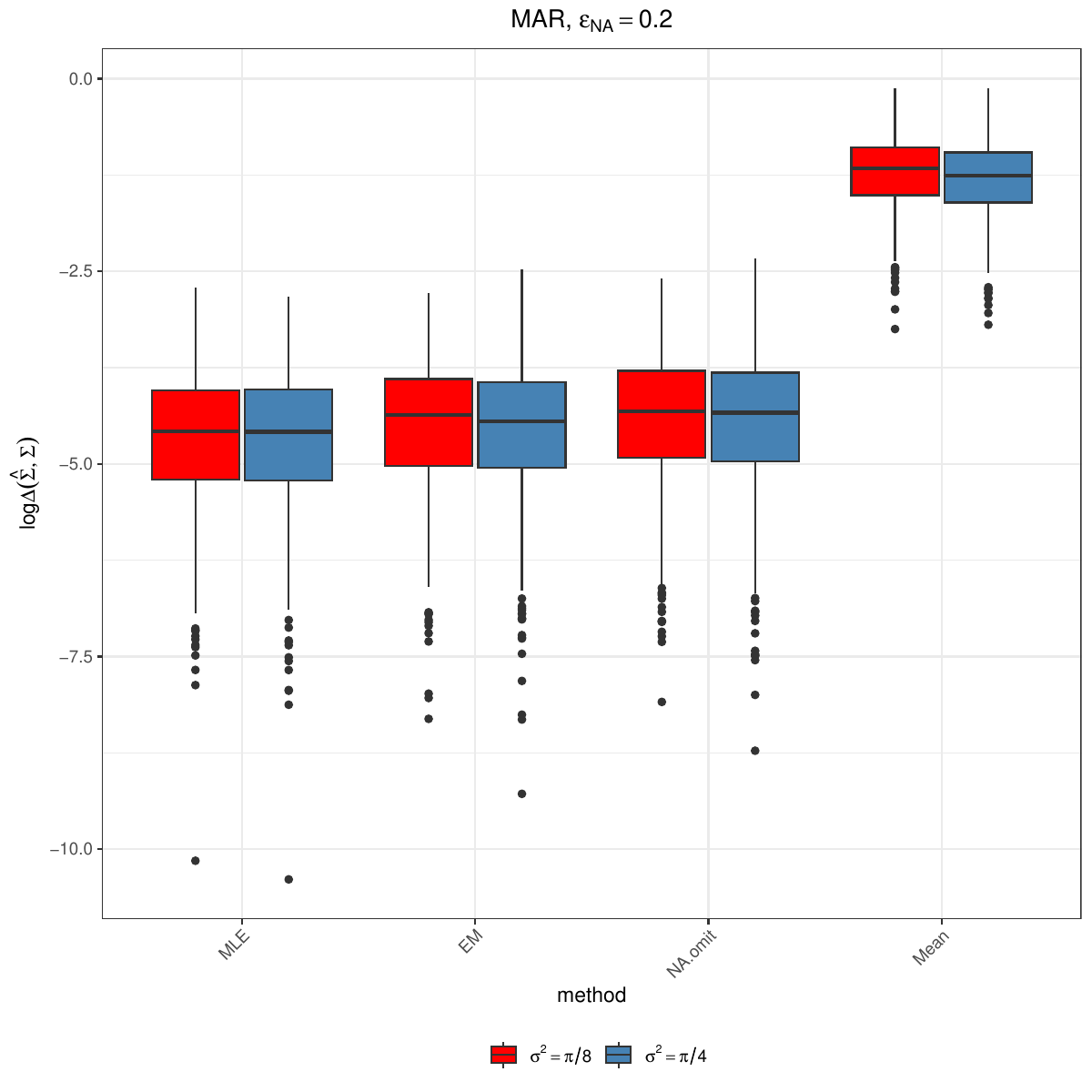}\\
	\includegraphics[scale=0.325]{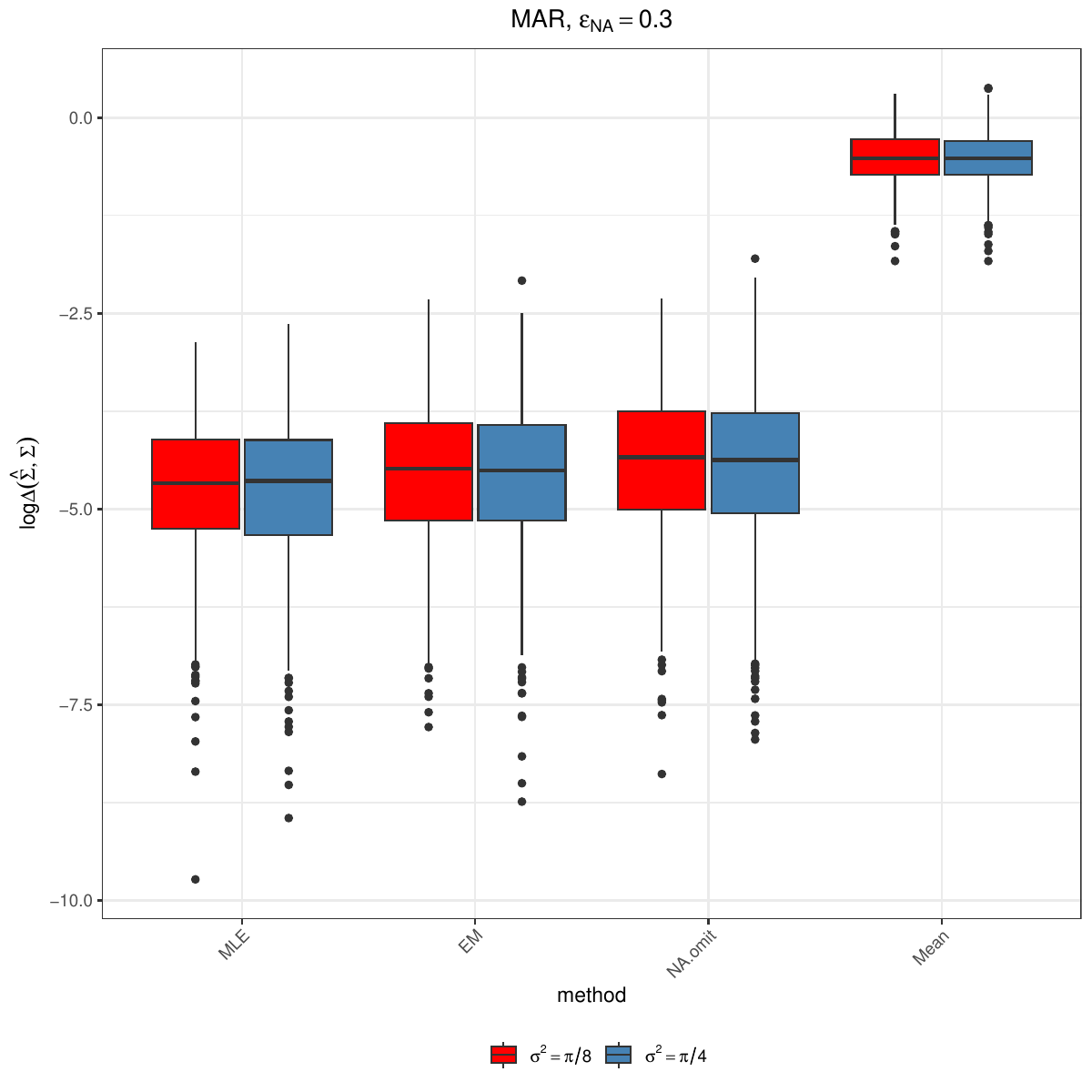}
	\includegraphics[scale=0.325]{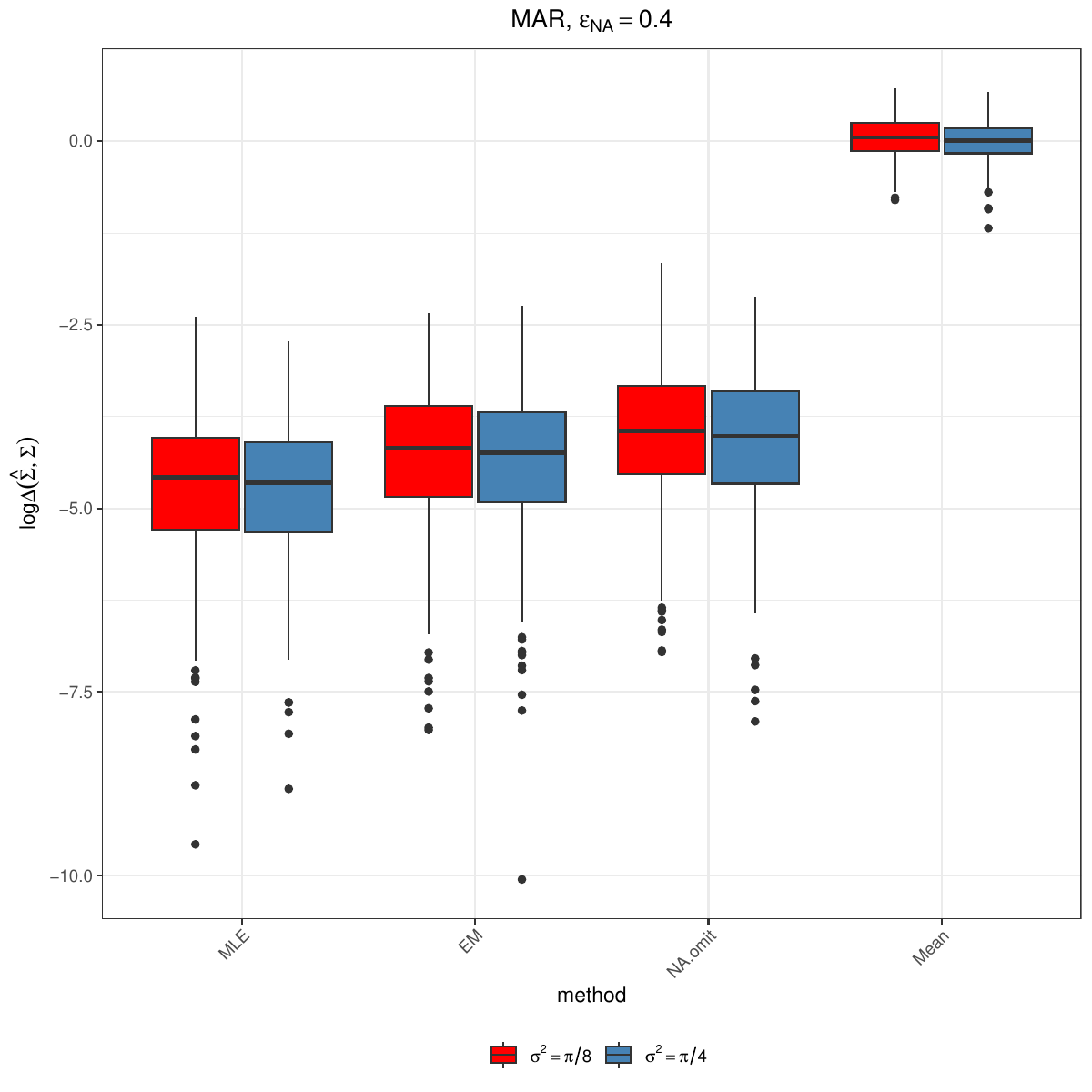}
	\caption{Numerical studies. Distribution of $\log\Delta(\hat\bSigma, \bSigma)$ for different methods uder MAR when $p=2$. Missingness rate is casewise.}
	\label{fig:4}
\end{figure}


\begin{figure}[!h]
	\includegraphics[scale=0.325]{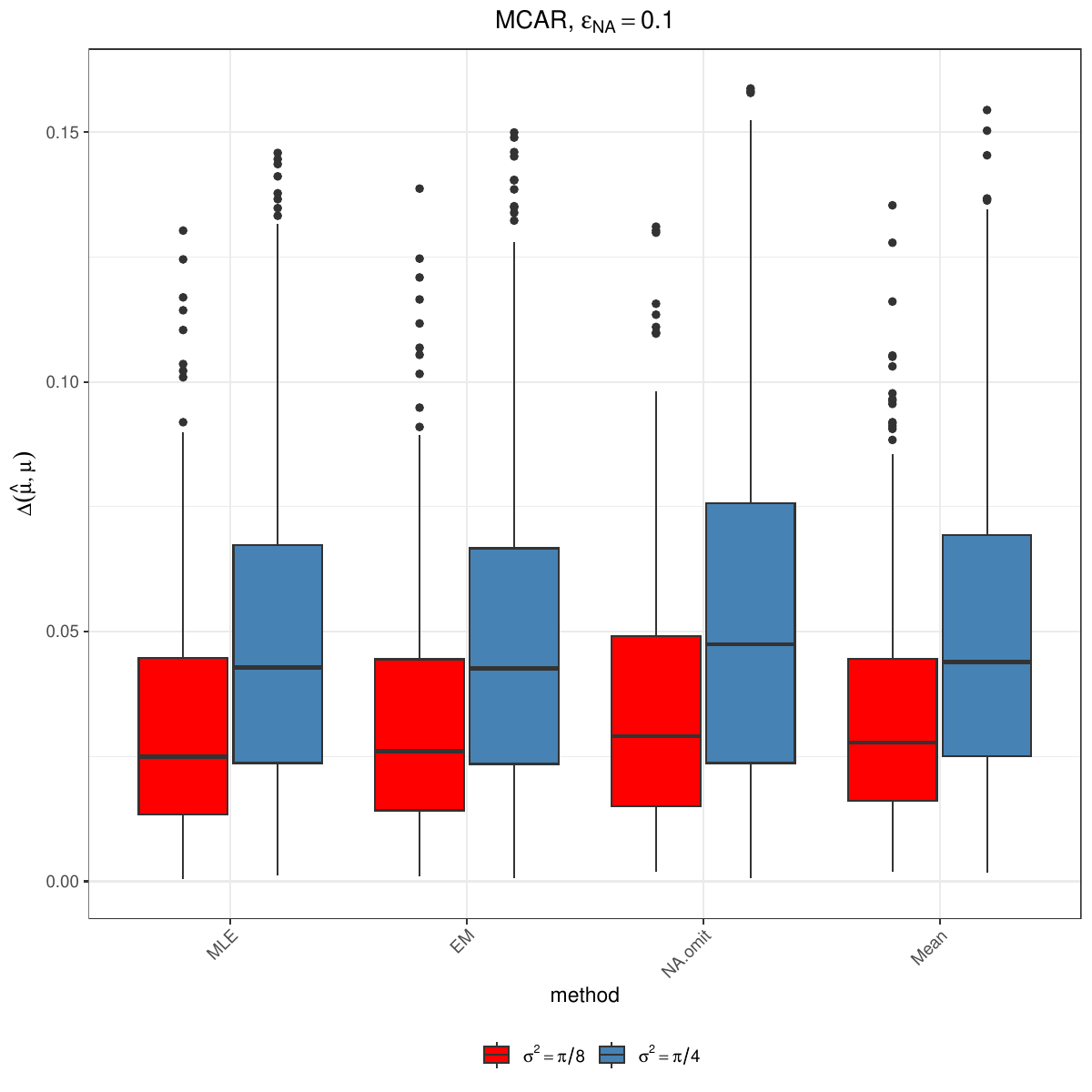}
	\includegraphics[scale=0.325]{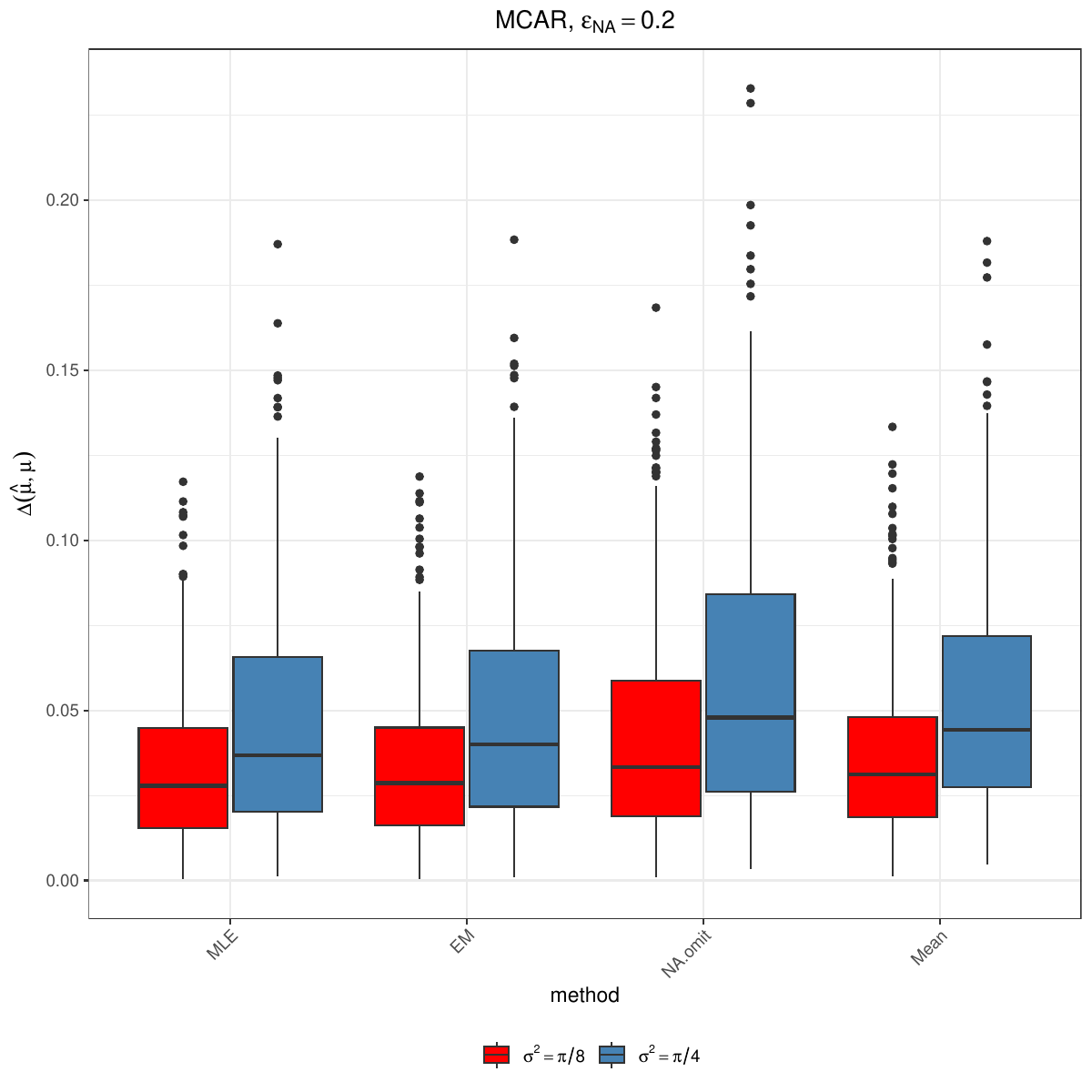}\\
	\includegraphics[scale=0.325]{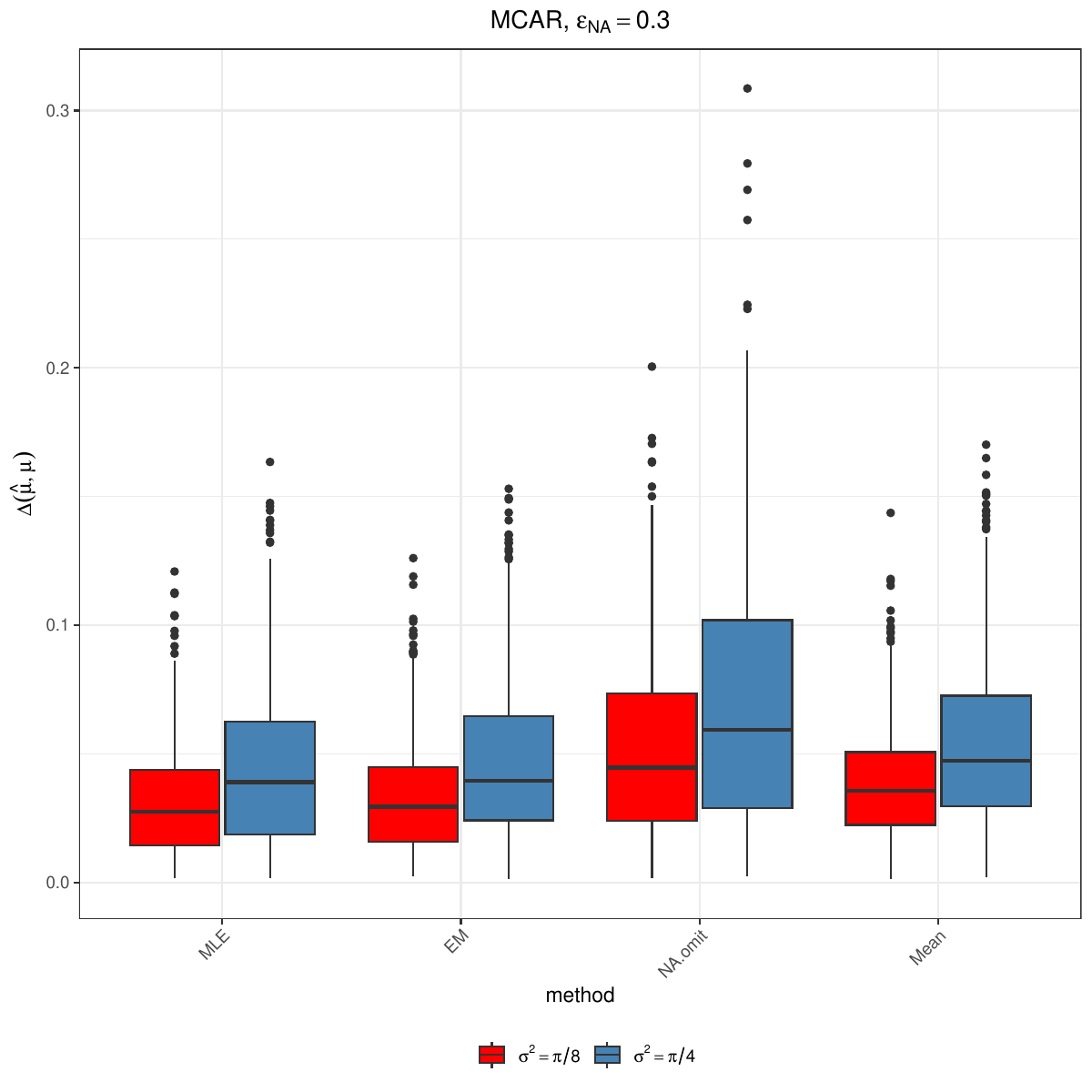}
	\includegraphics[scale=0.325]{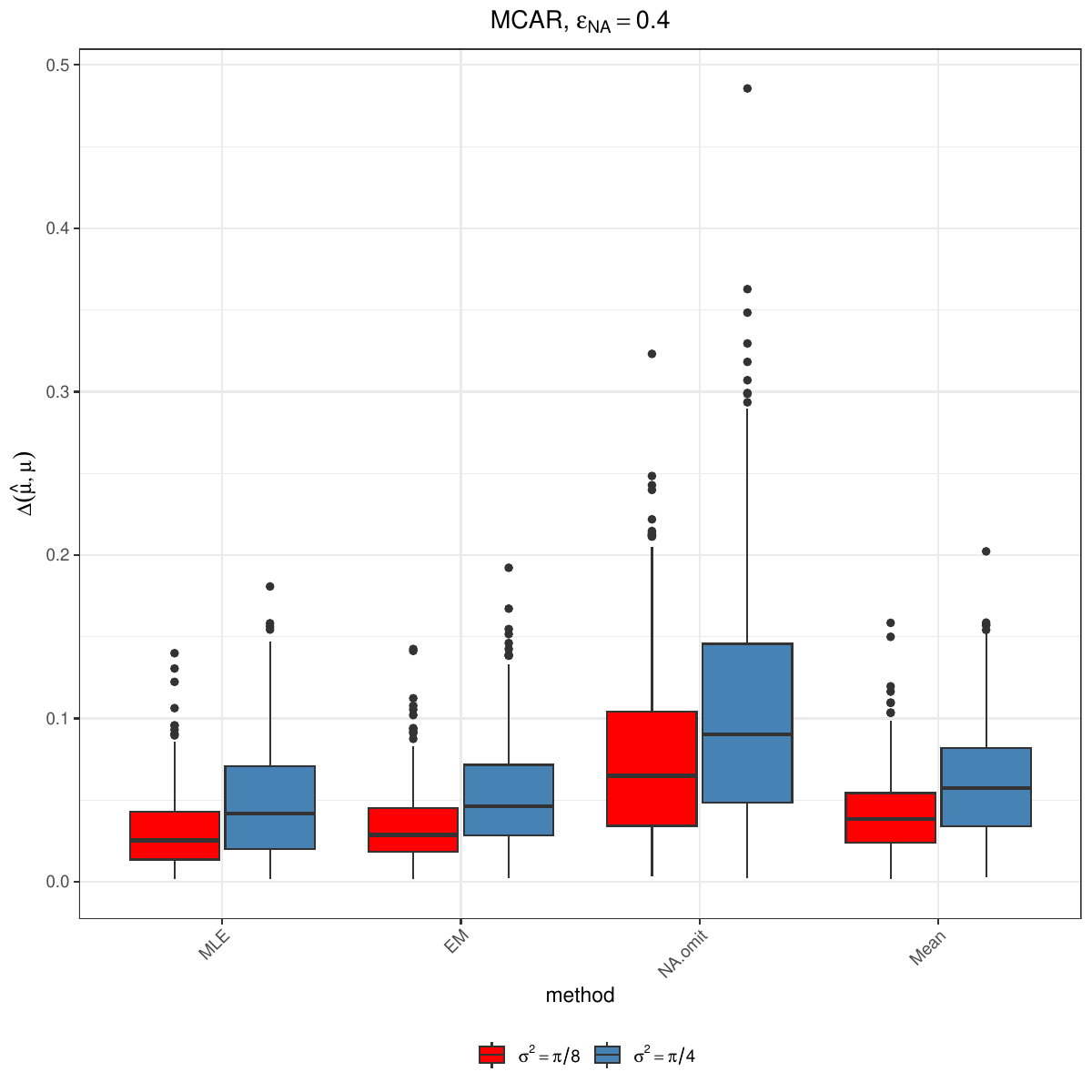}
	\caption{Numerical studies. Distribution of $\Delta(\hat\bmu)$ for different methods uder MCAR when $p=2$. Missingness rate is cellwise.}
	\label{fig:1a}
\end{figure}

\begin{figure}[!h]
	\includegraphics[scale=0.325]{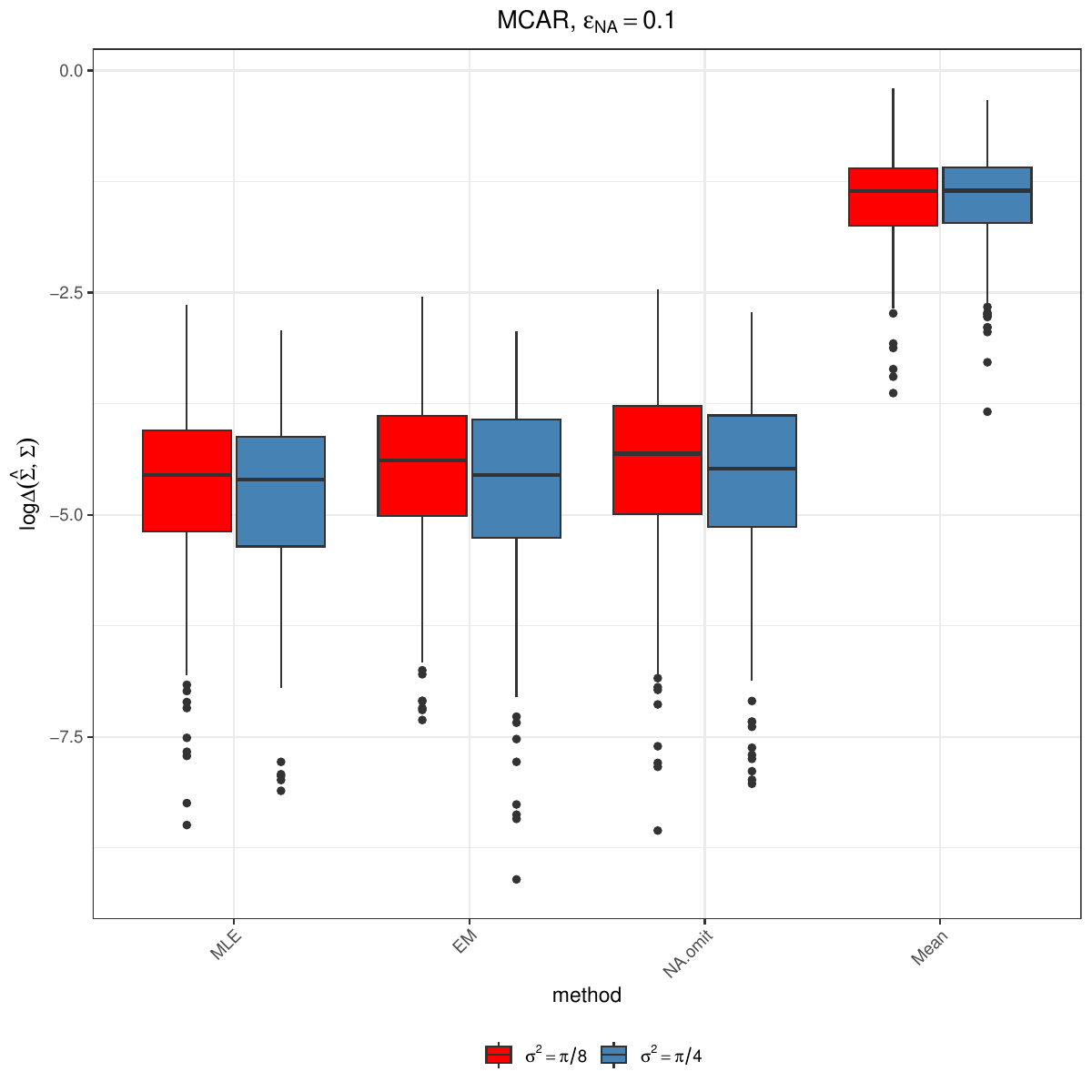}
	\includegraphics[scale=0.325]{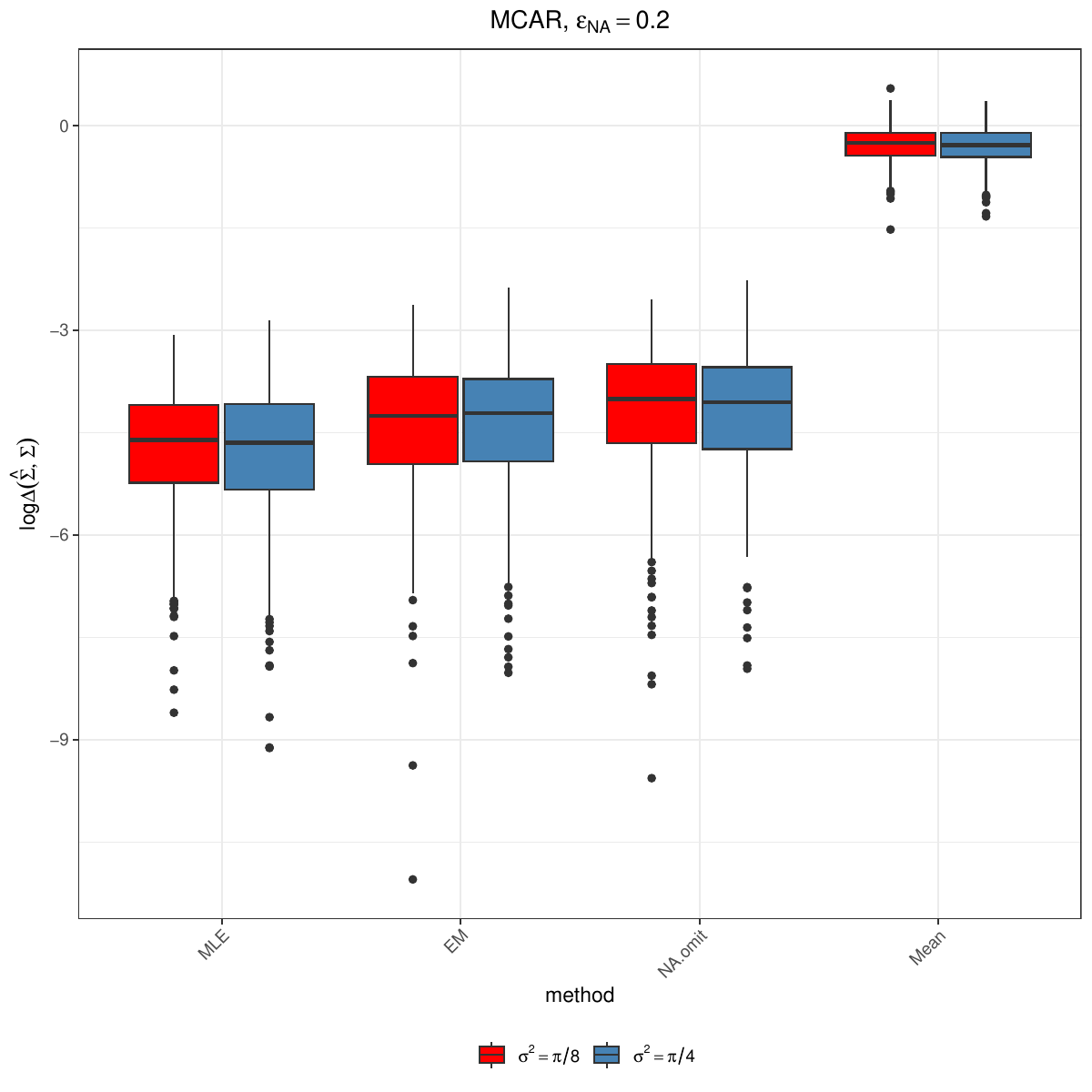}\\
	\includegraphics[scale=0.325]{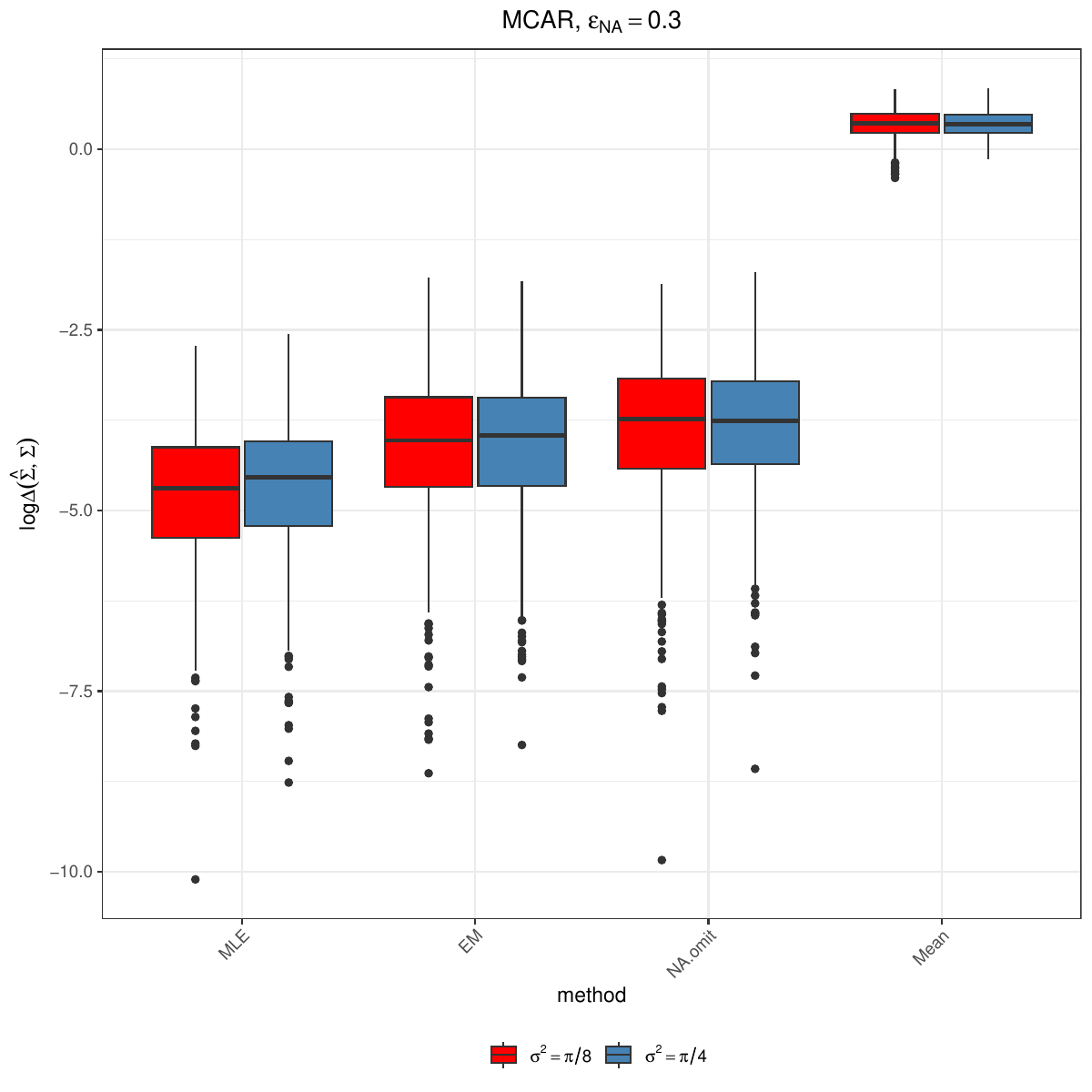}
	\includegraphics[scale=0.325]{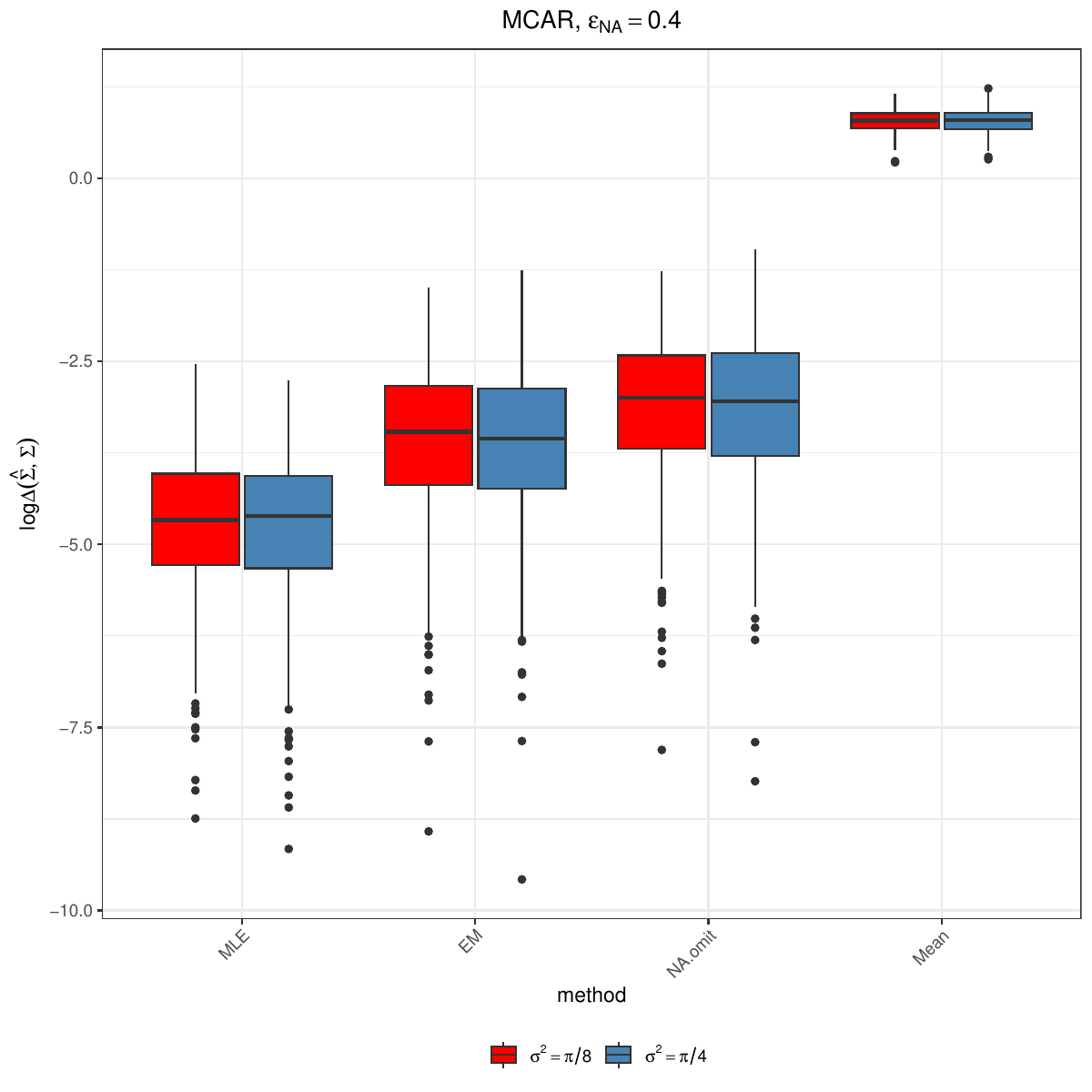}
	\caption{Numerical studies. Distribution of $\log\Delta(\hat\bSigma, \bSigma)$ for different methods uder MCAR when $p=2$. Missingness rate is cellwise.}
	\label{fig:2a}
\end{figure}

\begin{figure}[!h]
	\includegraphics[scale=0.325]{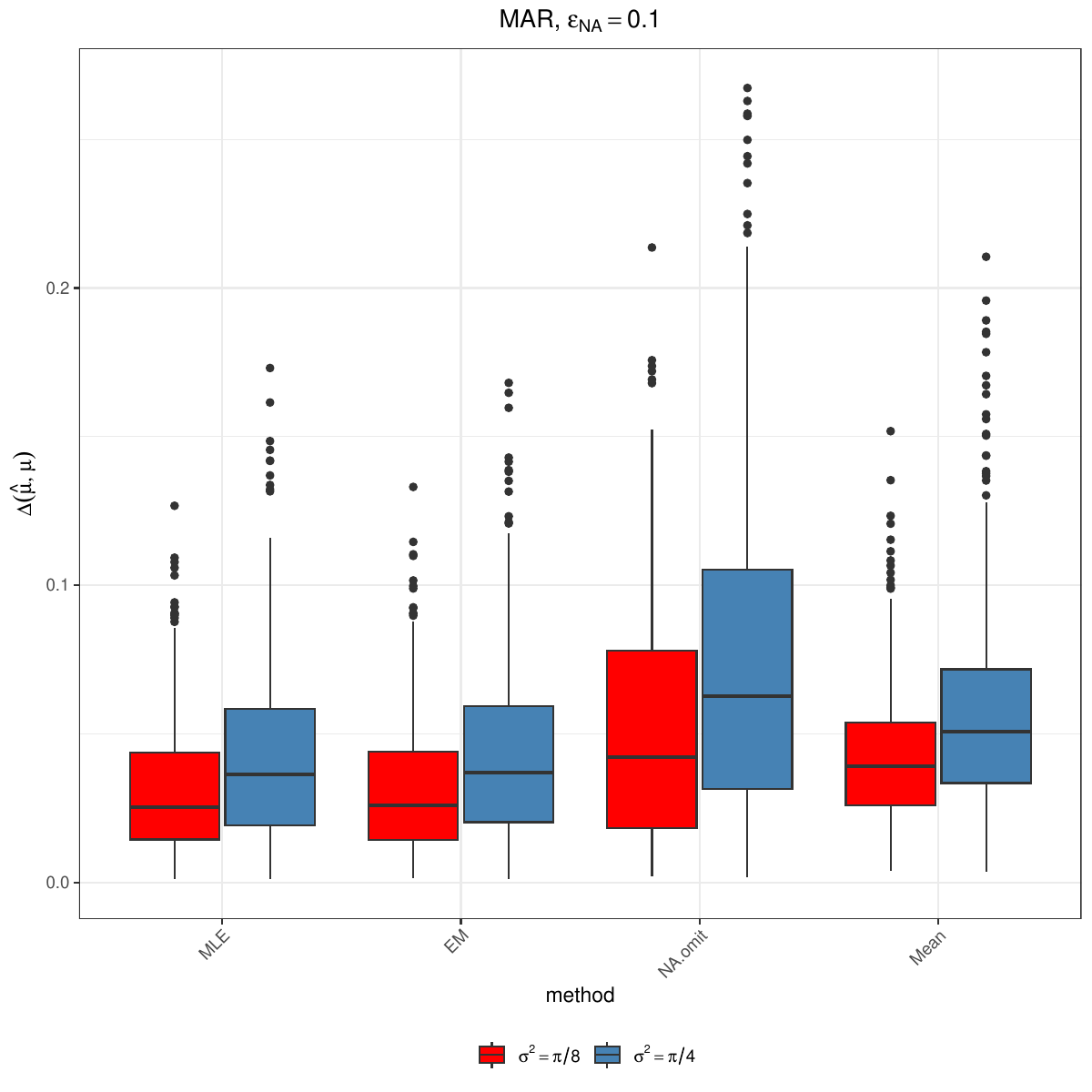}
	\includegraphics[scale=0.325]{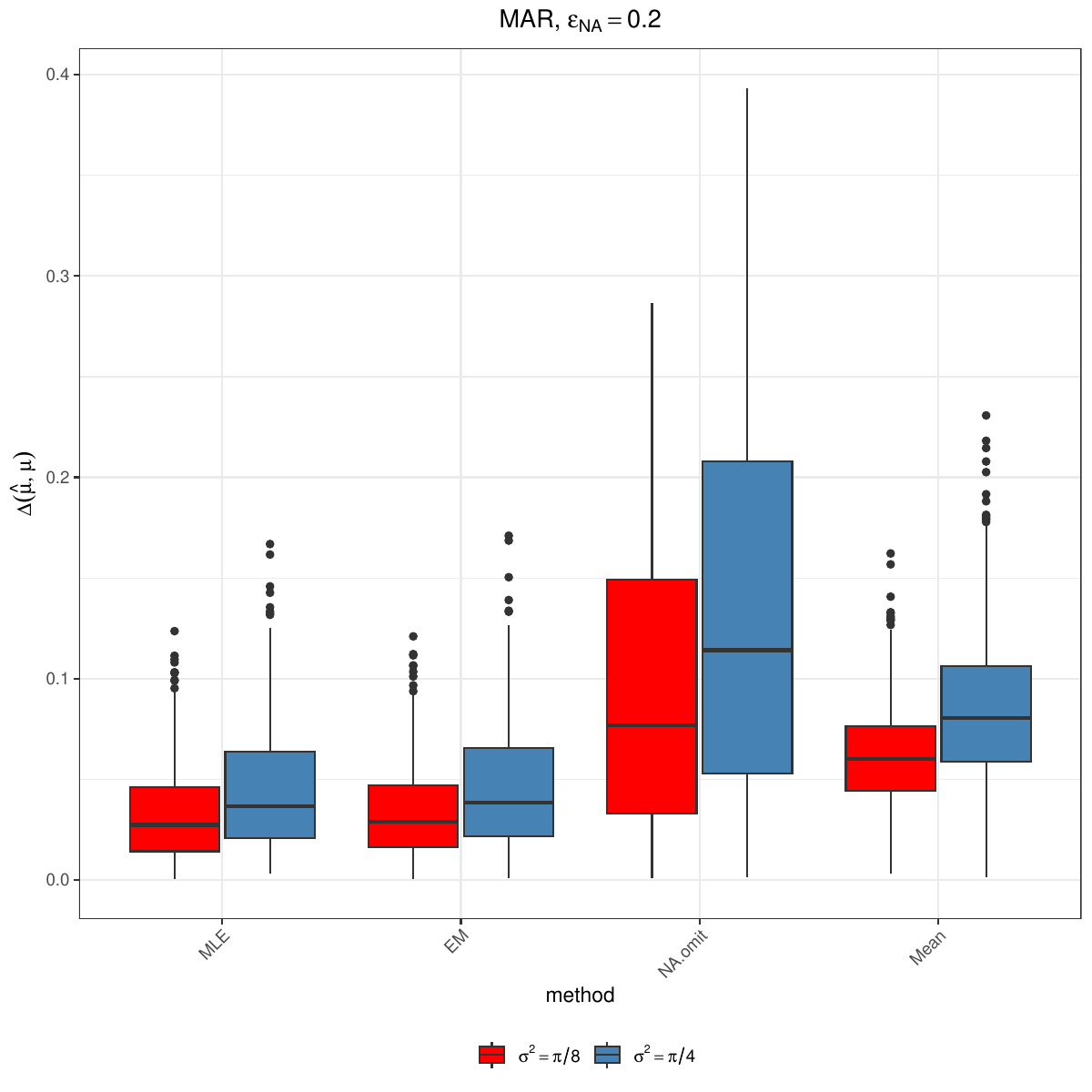}\\
	\includegraphics[scale=0.325]{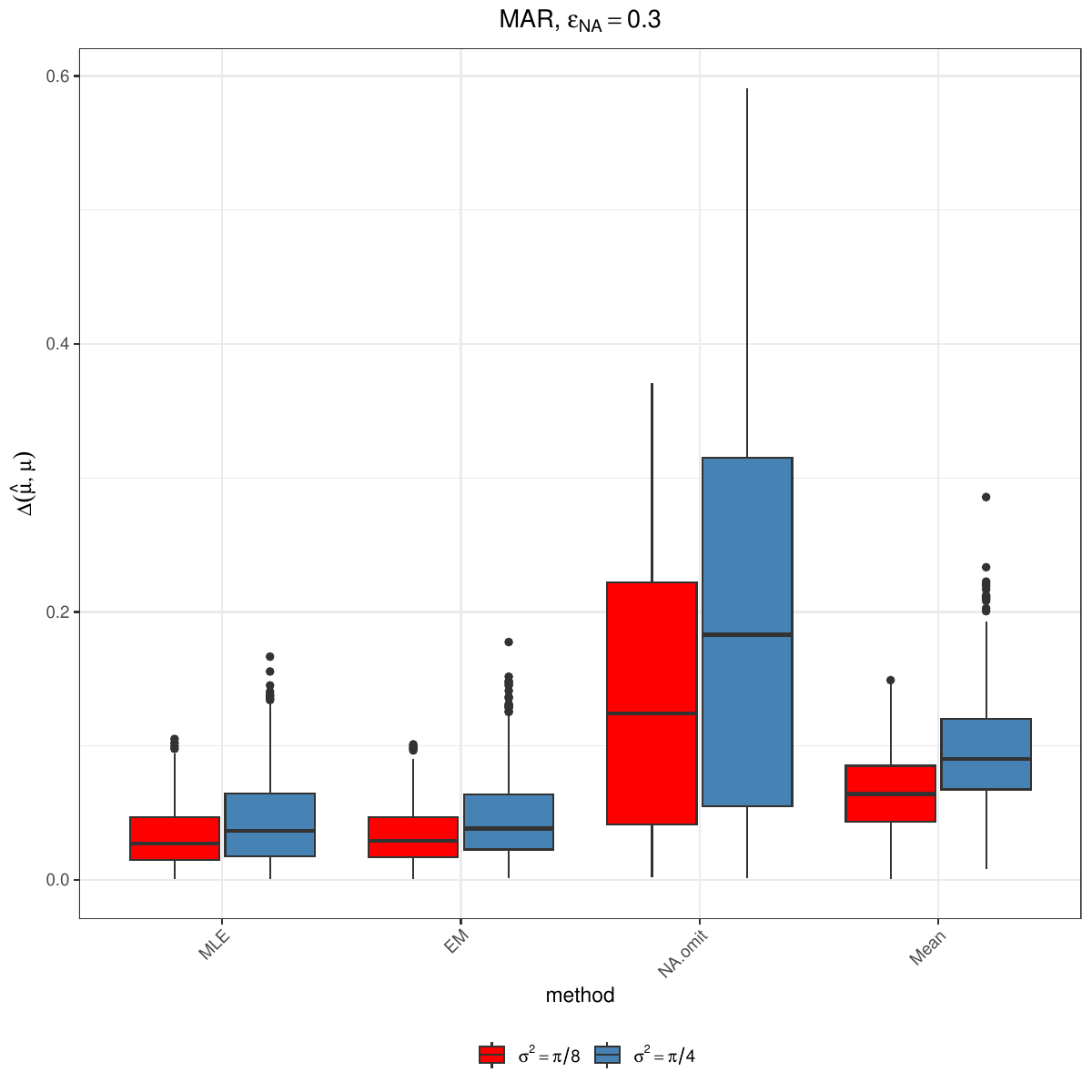}
	\includegraphics[scale=0.325]{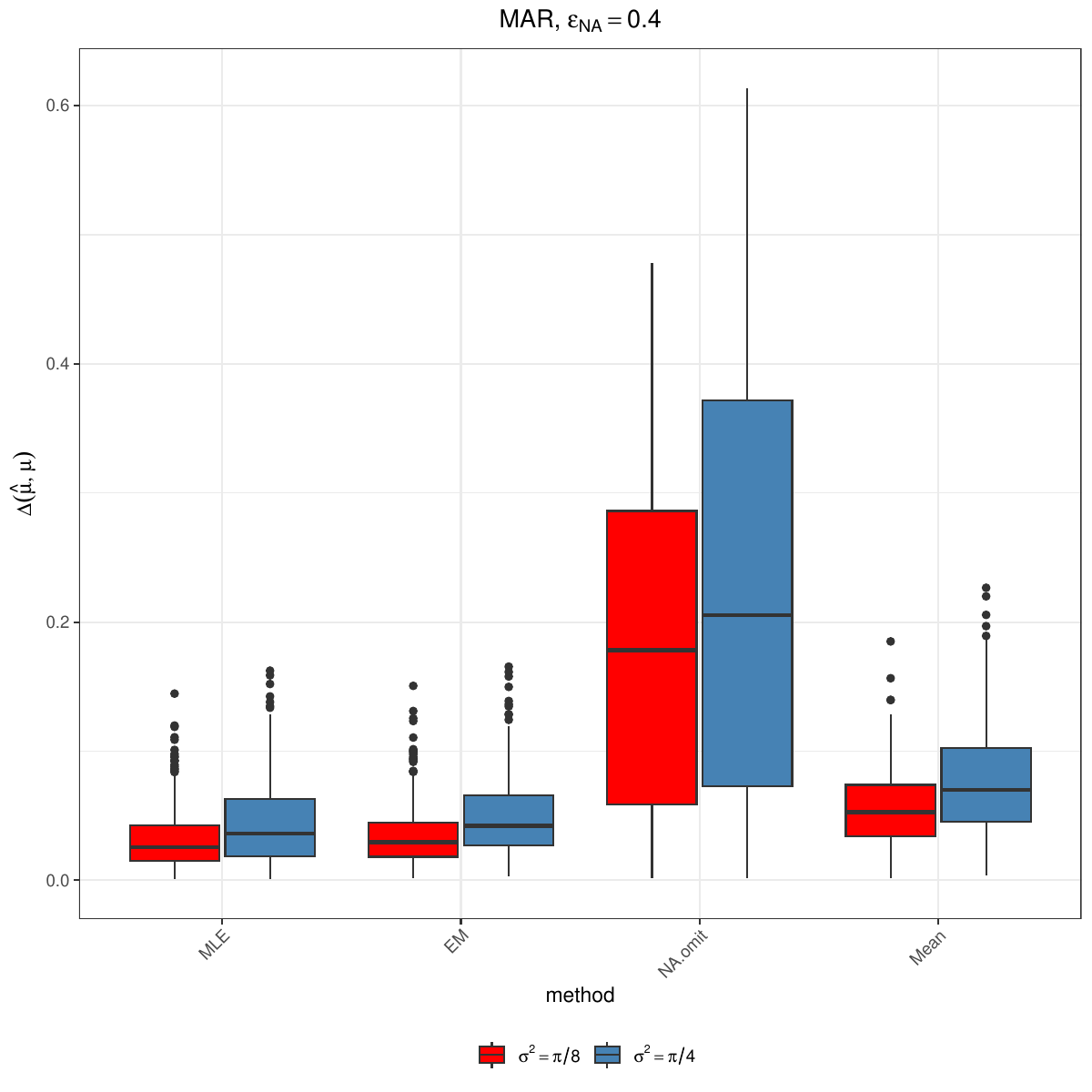}
	\caption{Numerical studies. Distribution of $\Delta(\hat\bmu)$ for different methods uder MAR when $p=2$. Missingness rate is cellwise.}
	\label{fig:3a}
\end{figure}

\begin{figure}[!h]
	\includegraphics[scale=0.325]{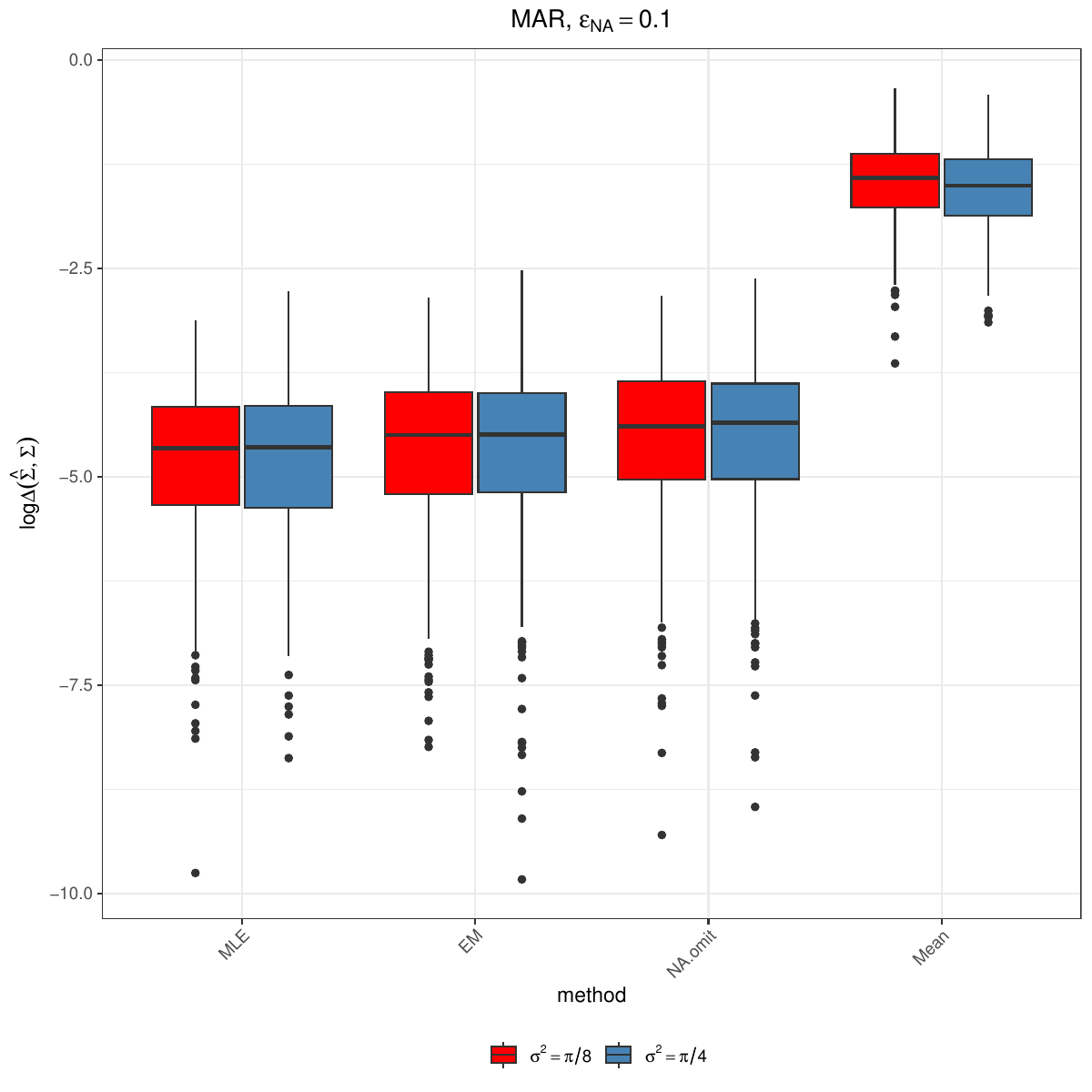}
	\includegraphics[scale=0.325]{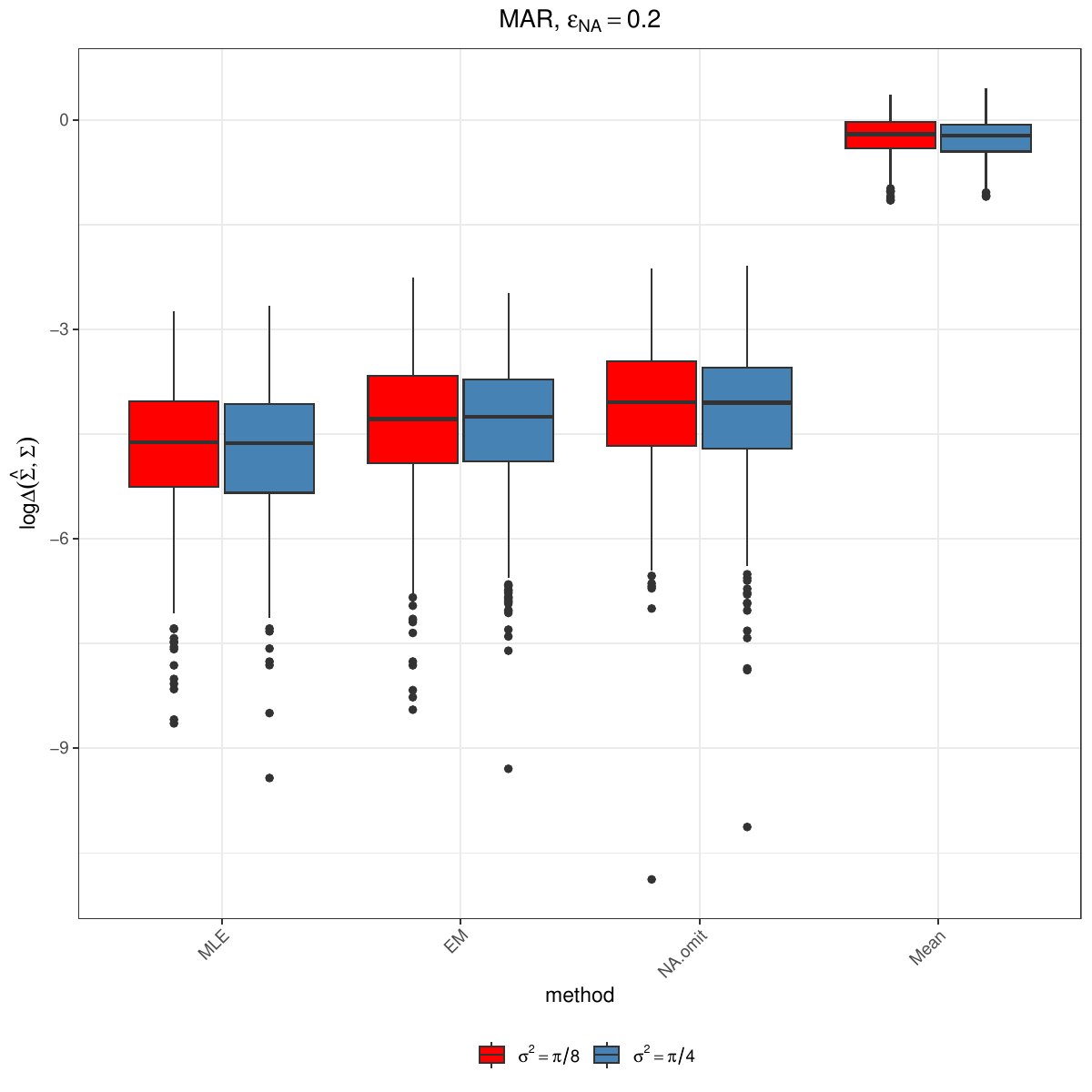}\\
	\includegraphics[scale=0.325]{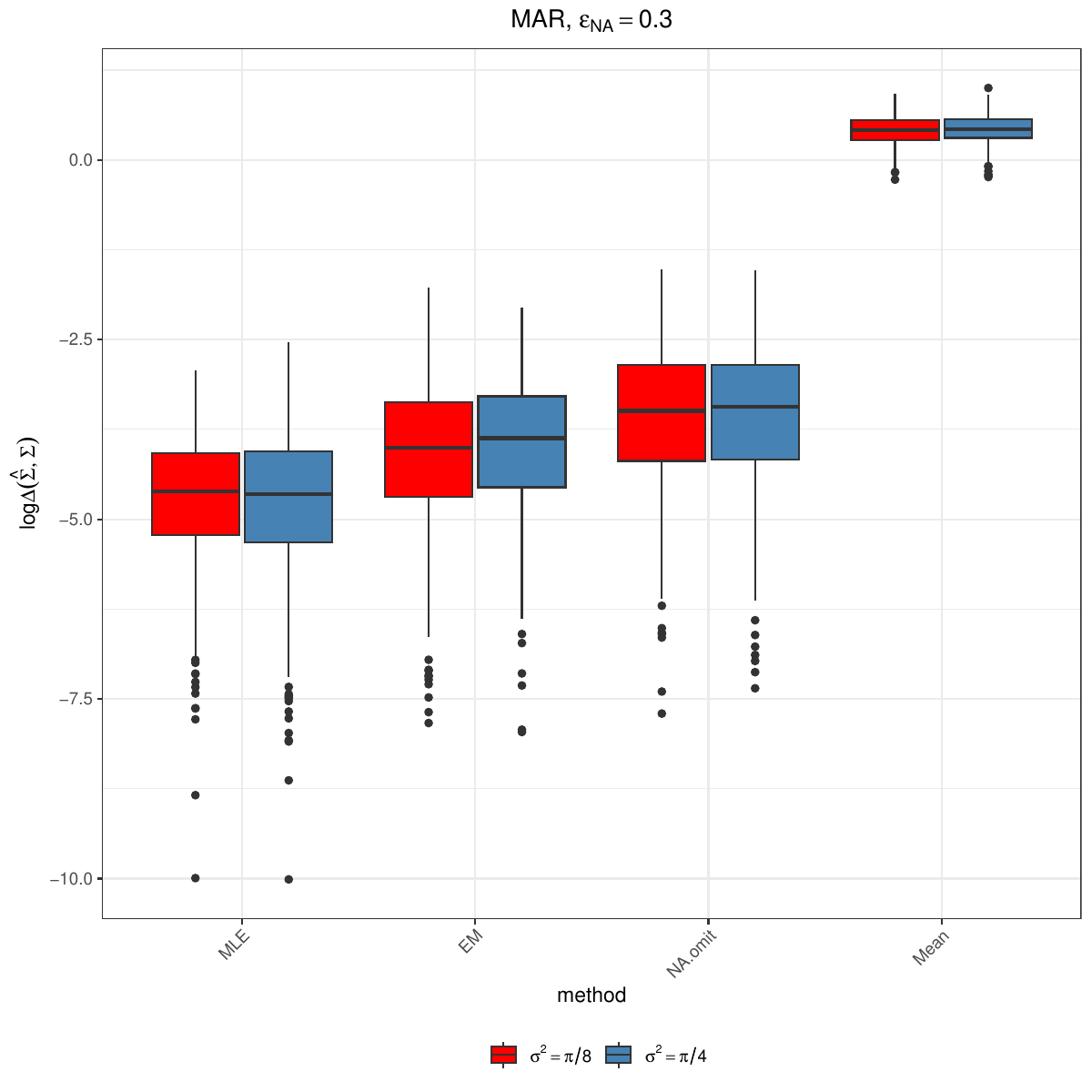}
	\includegraphics[scale=0.325]{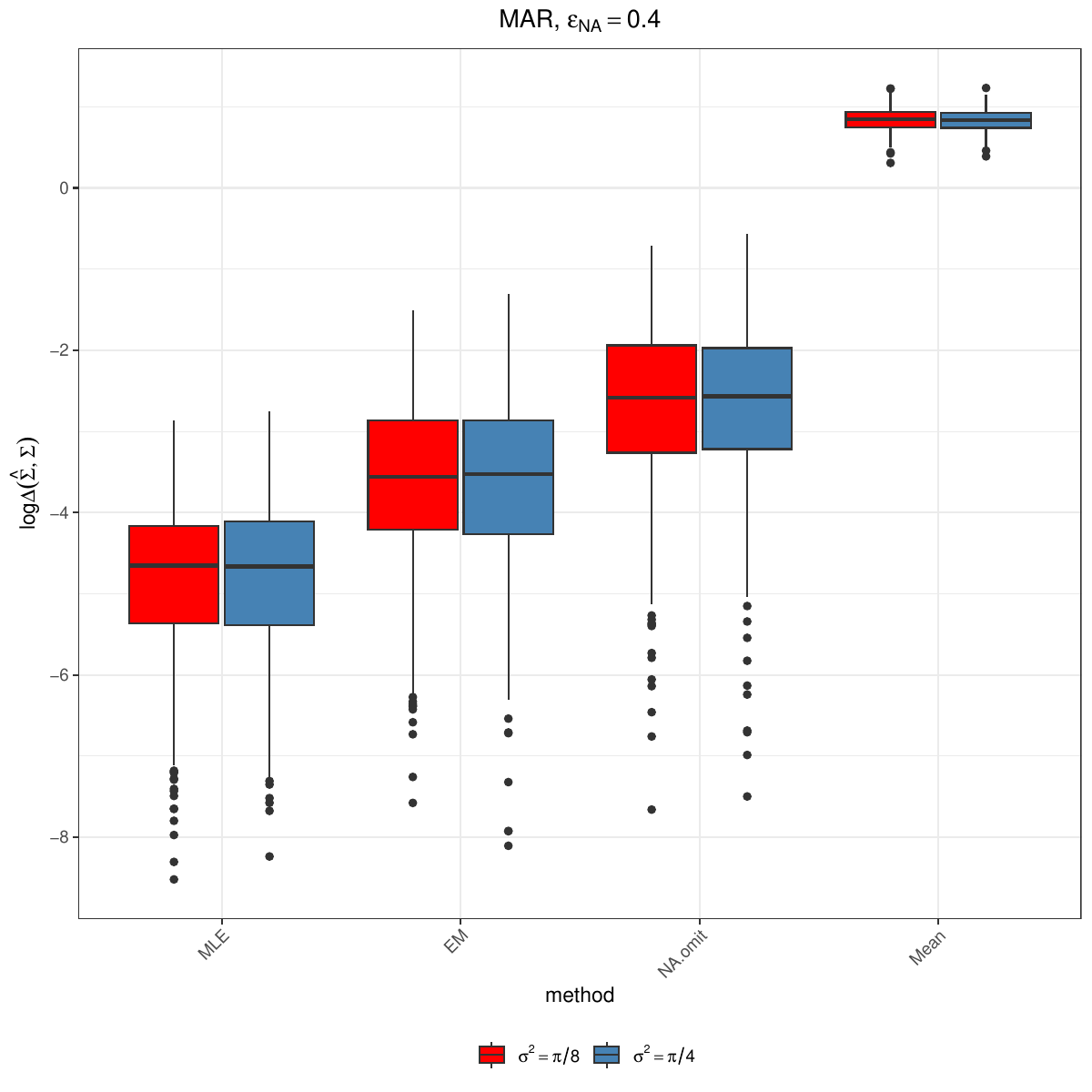}
	\caption{Numerical studies. Distribution of $\log\Delta(\hat\bSigma, \bSigma)$ for different methods uder MAR when $p=2$. Missingness rate is cellwise.}
	\label{fig:4a}
\end{figure}

\begin{figure}[!h]
	\includegraphics[scale=0.325]{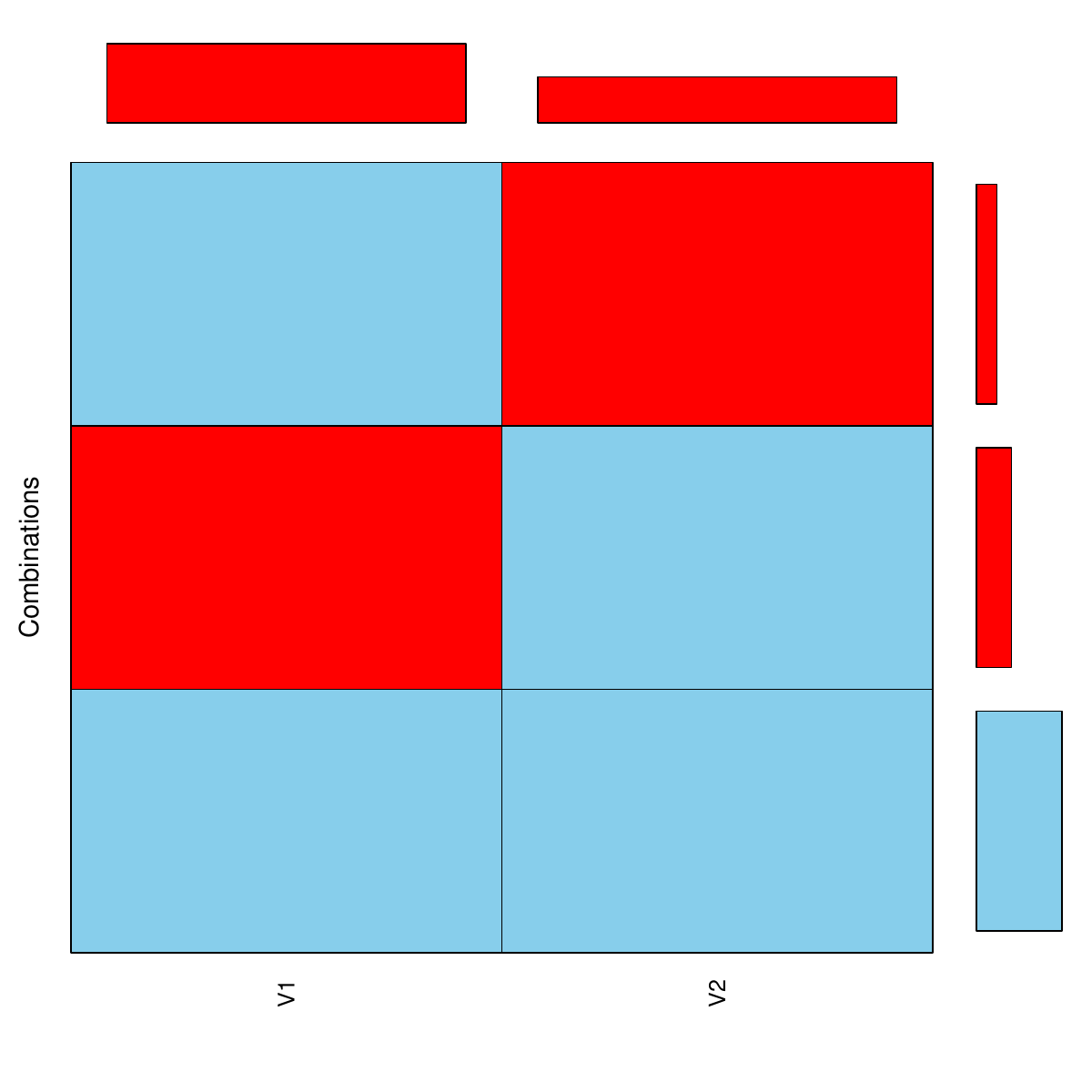}
	\includegraphics[scale=0.325]{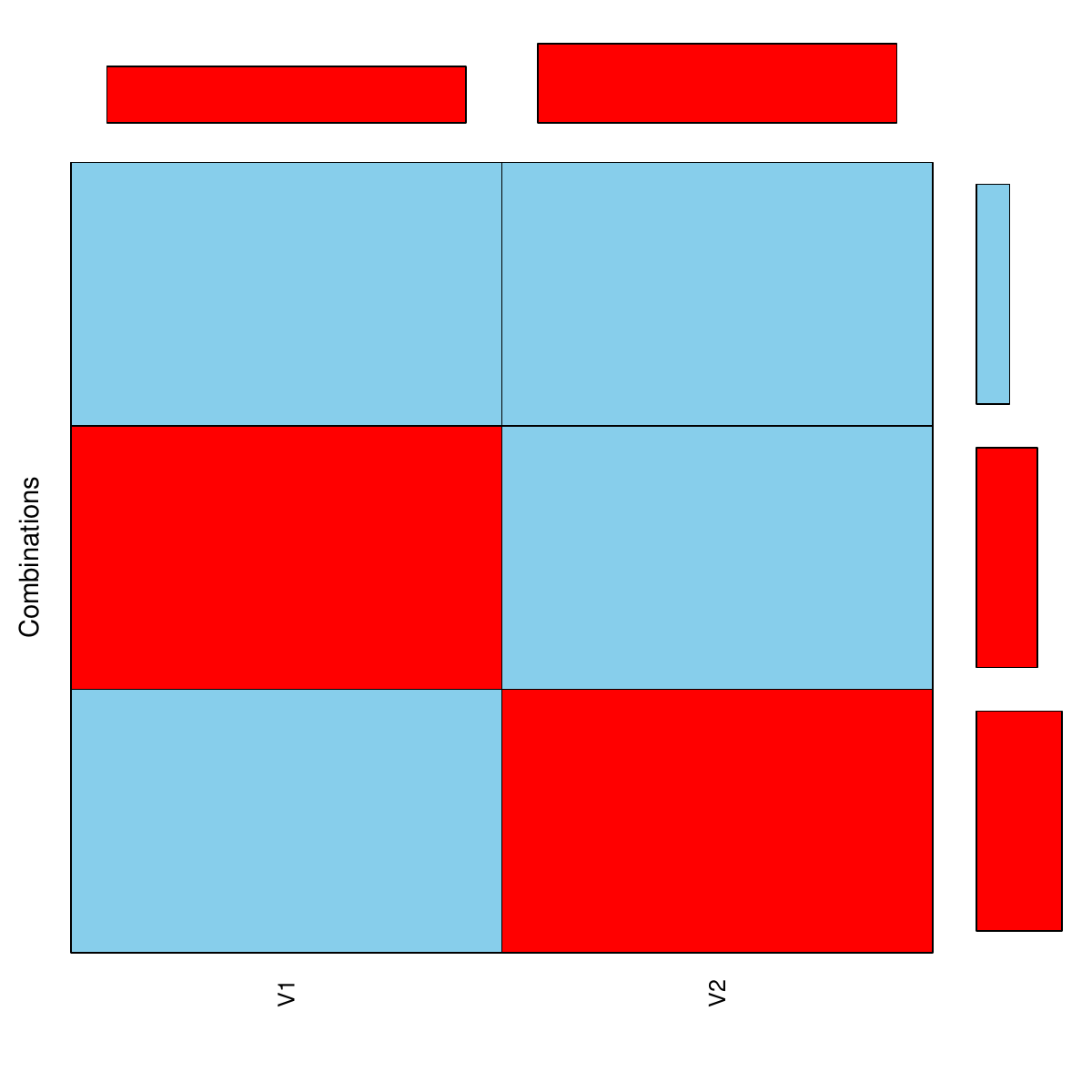}\\
	\includegraphics[scale=0.325]{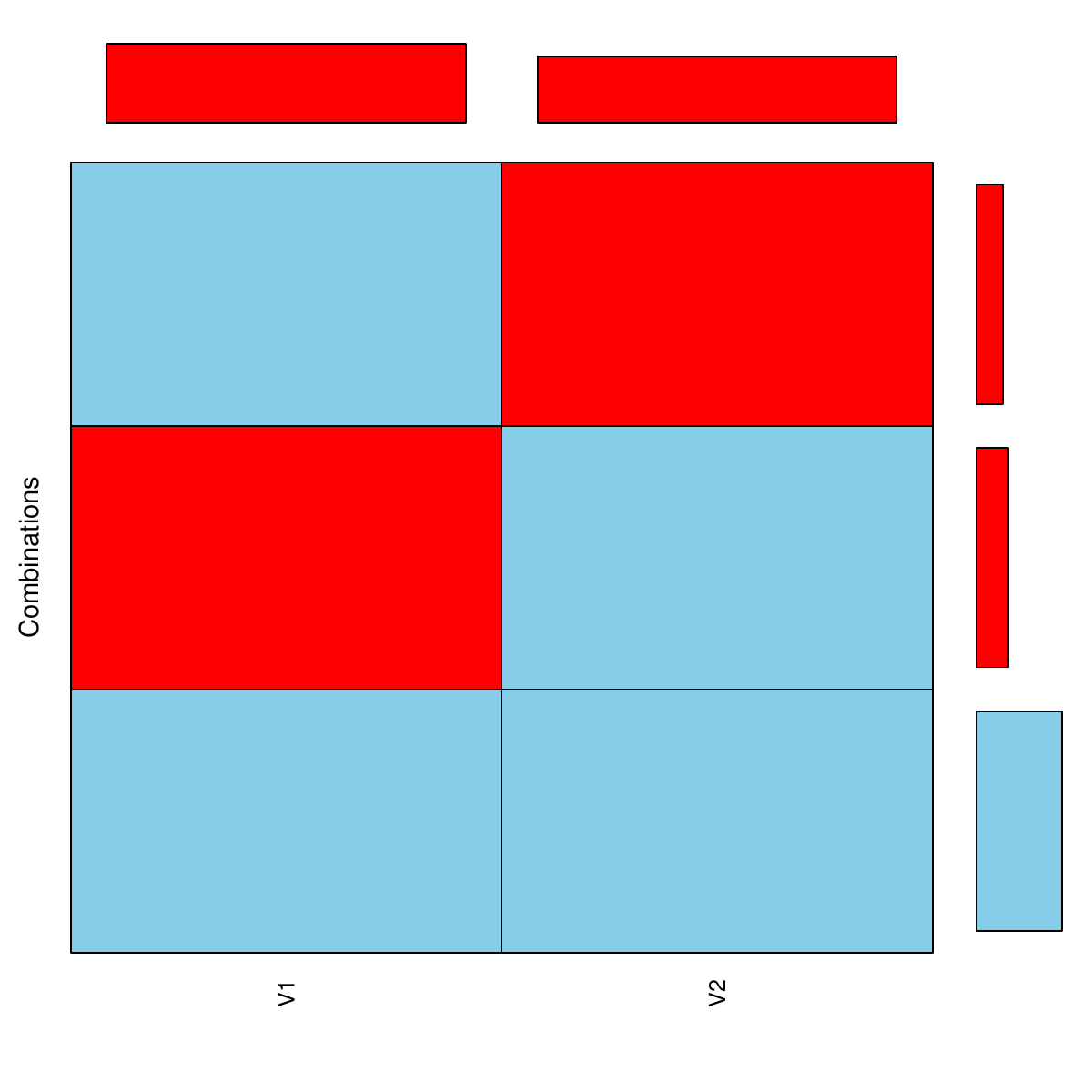}
	\includegraphics[scale=0.325]{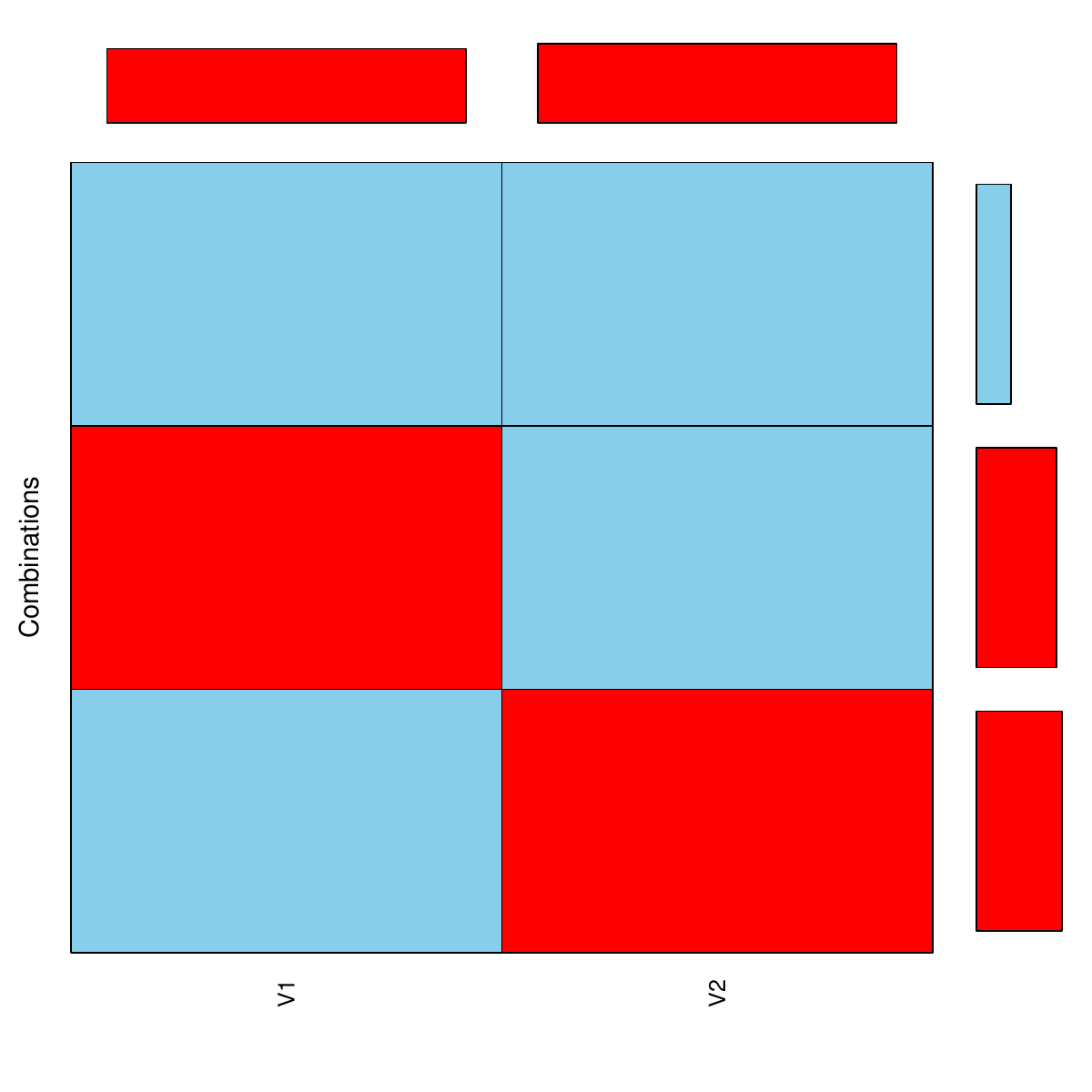}
	\caption{Numerical studies. Missingness patterns when $p=2$, $\epsilon_{NA}=0.4$ under MCAR (top)  and MAR (bottom), casewise (left) and cellwise (right) missingness rate. The top bars give the frequency of missingness by variables; the side bars give the frequency of missingness by pattern. Missing values in red.}
	\label{fig:pattern_p2}
\end{figure}


\begin{figure}[!h]
	\includegraphics[scale=0.325]{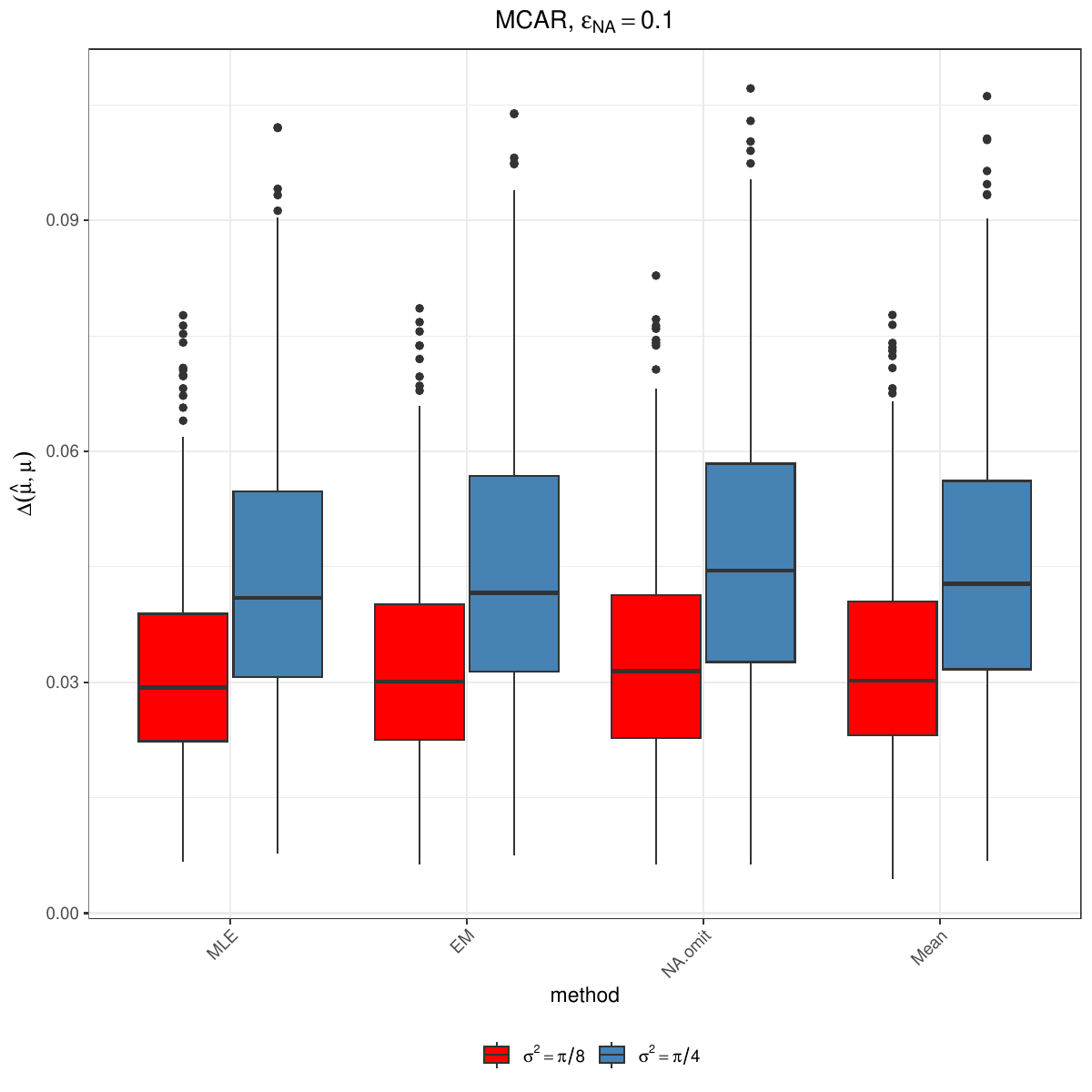}
	\includegraphics[scale=0.325]{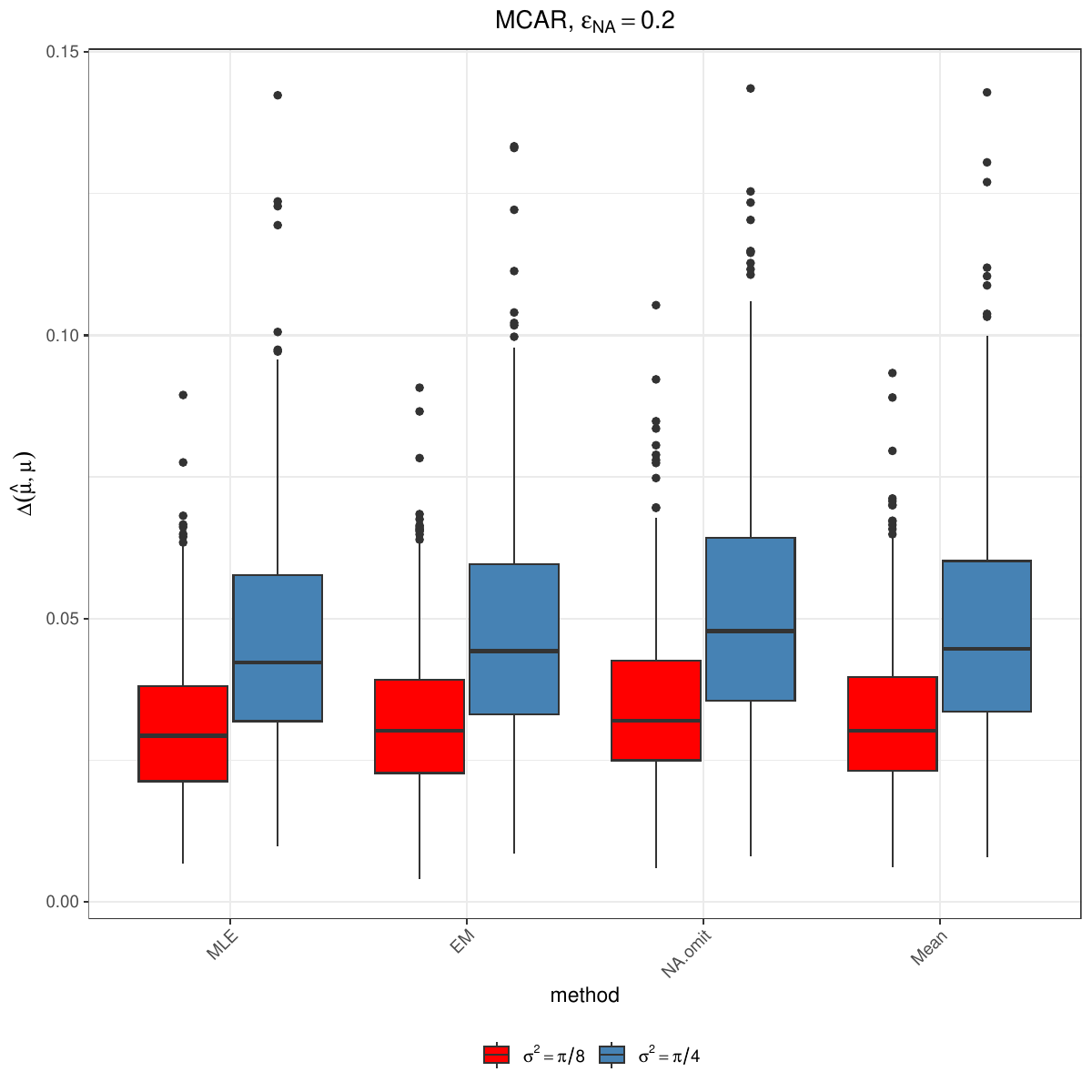}\\
	\includegraphics[scale=0.325]{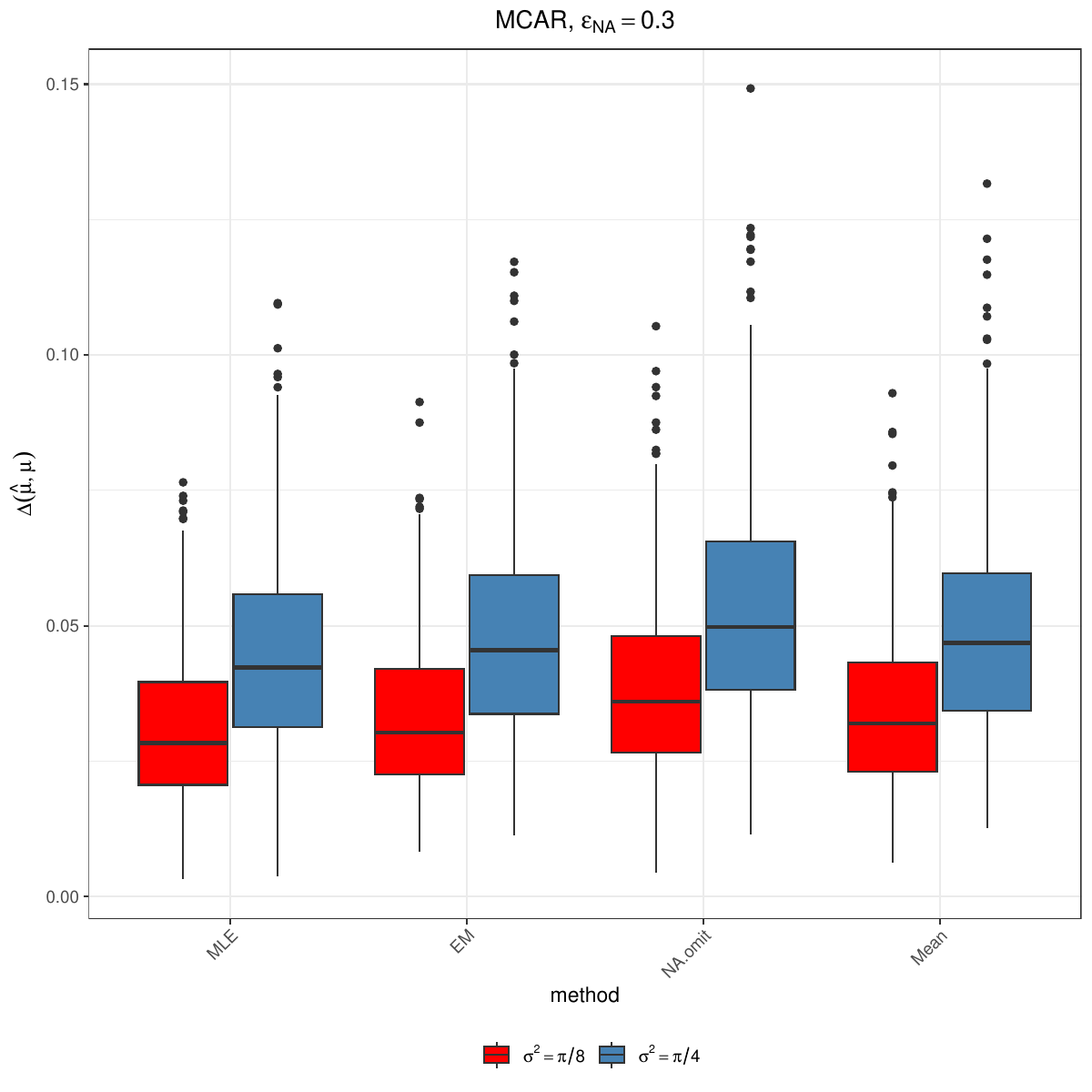}
	\includegraphics[scale=0.325]{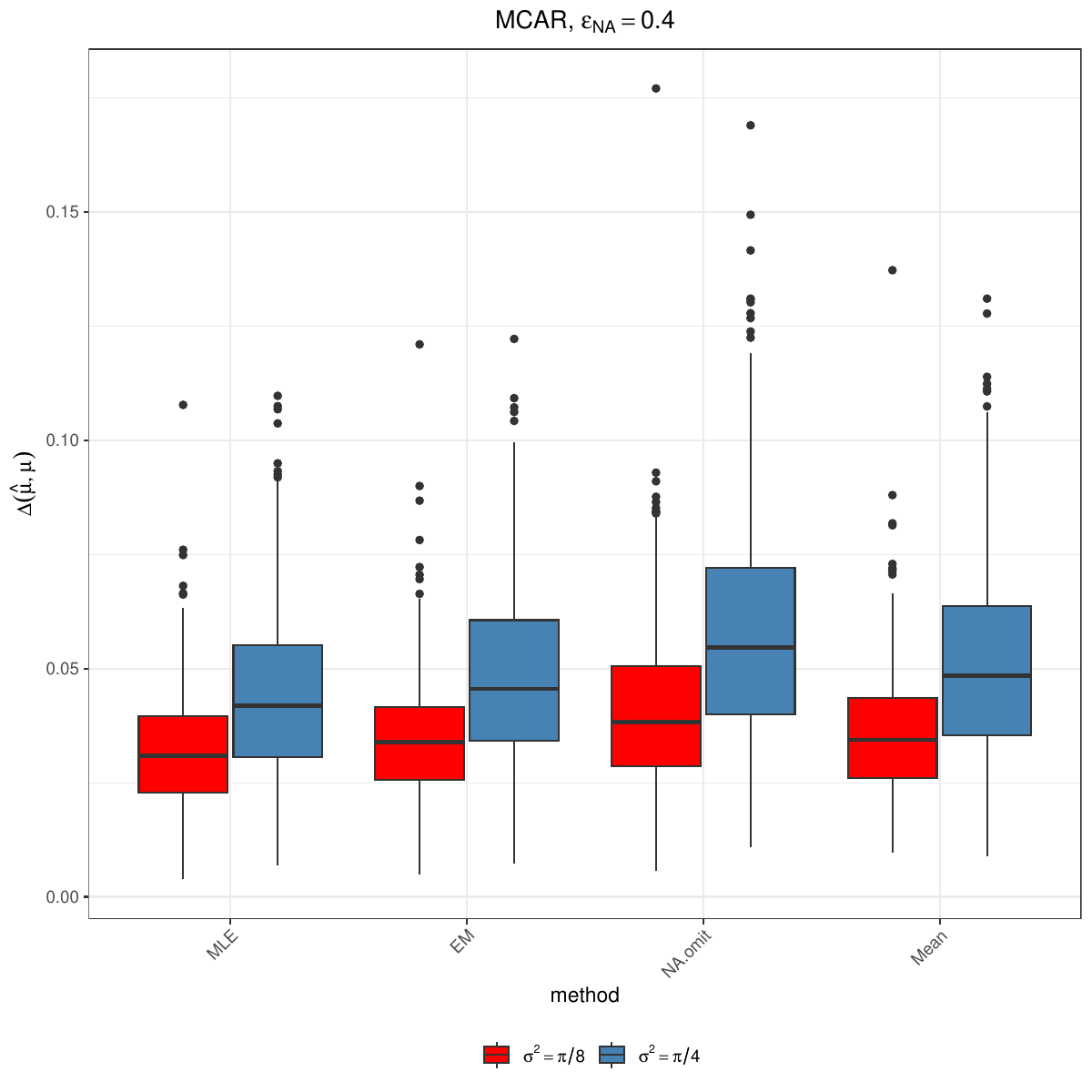}
	\caption{Numerical studies. Distribution of $\Delta(\hat\bmu)$ for different methods uder MCAR when $p=5$. Missingness rate is casewise.}
	\label{fig:5}
\end{figure}

\begin{figure}[!h]
	\includegraphics[scale=0.325]{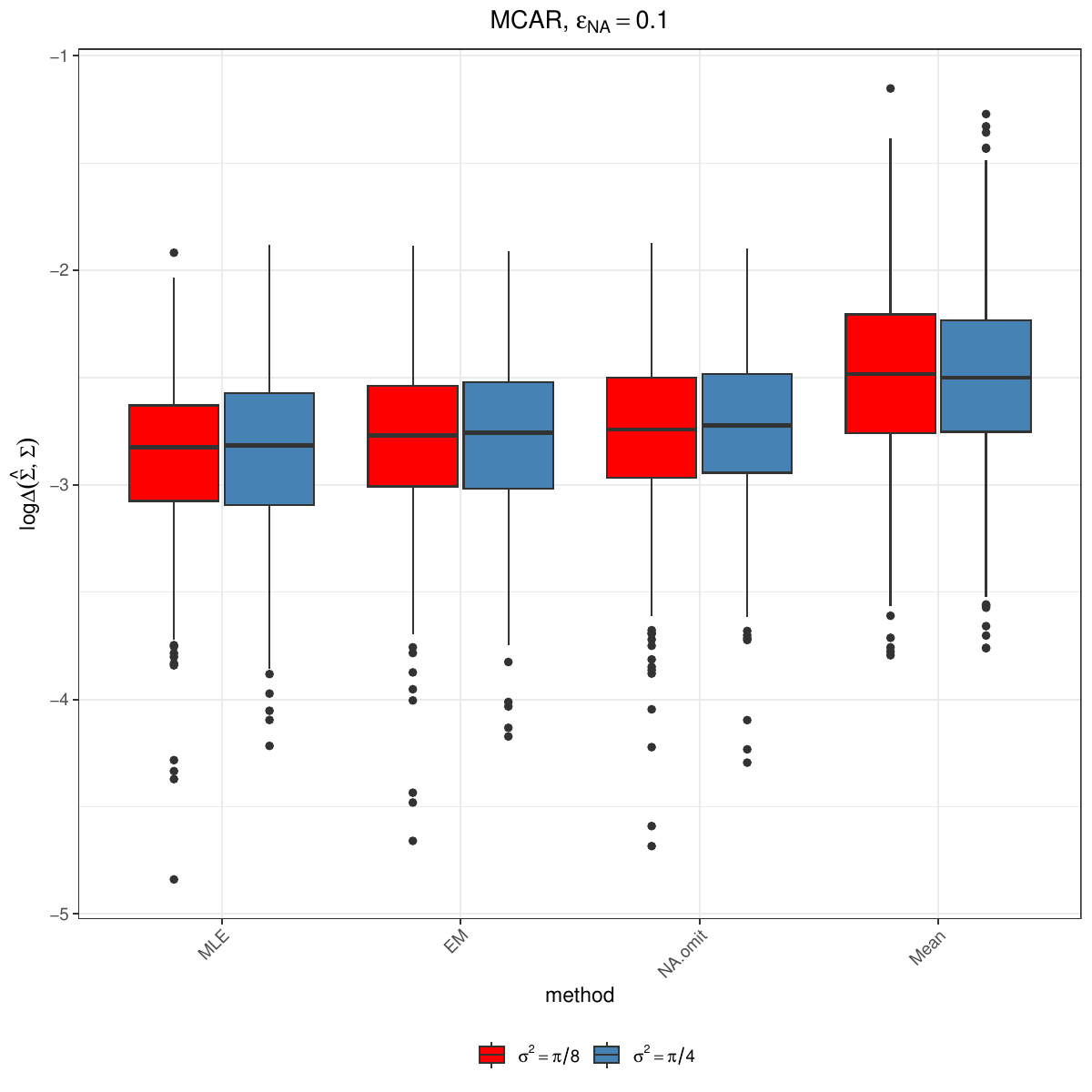}
	\includegraphics[scale=0.325]{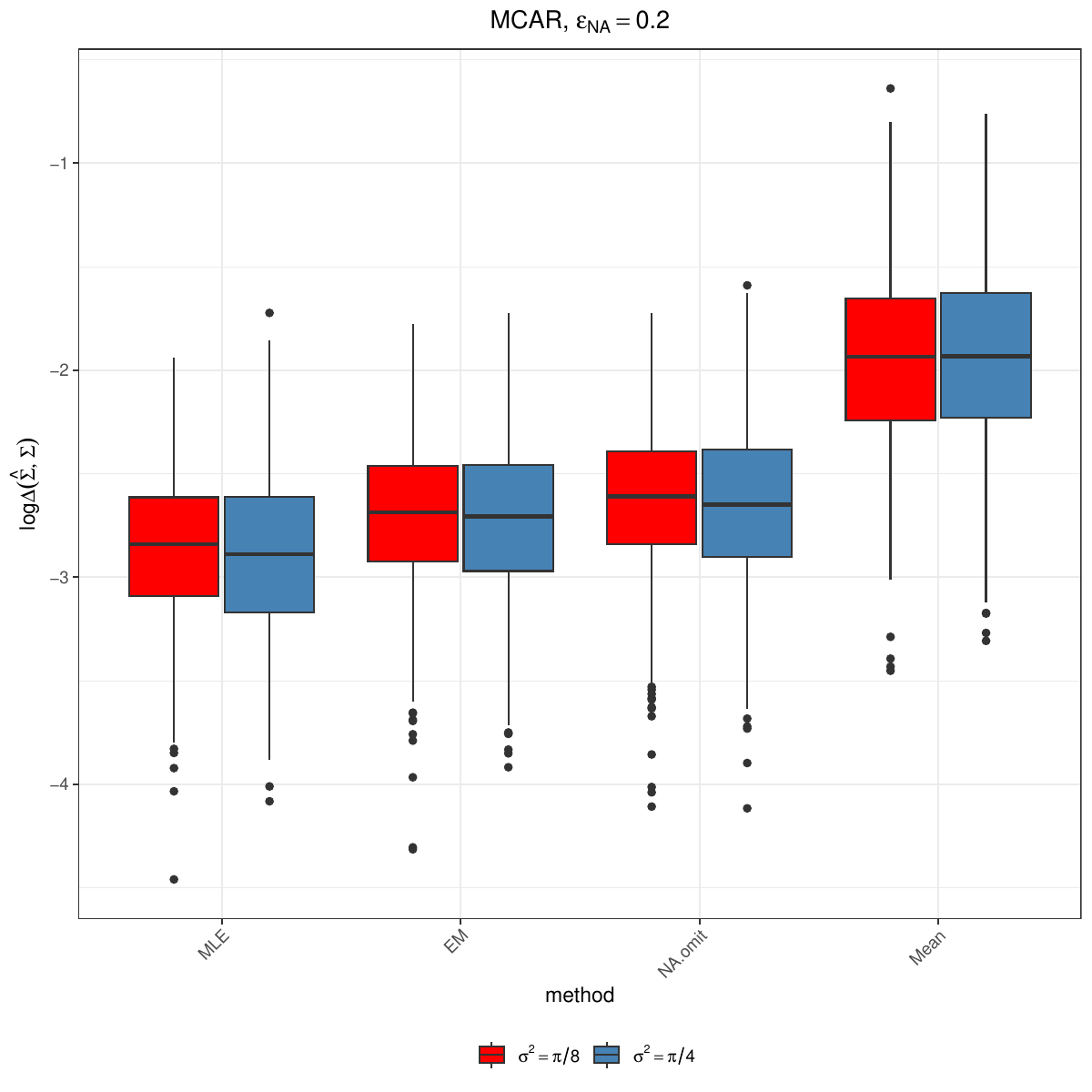}\\
	\includegraphics[scale=0.325]{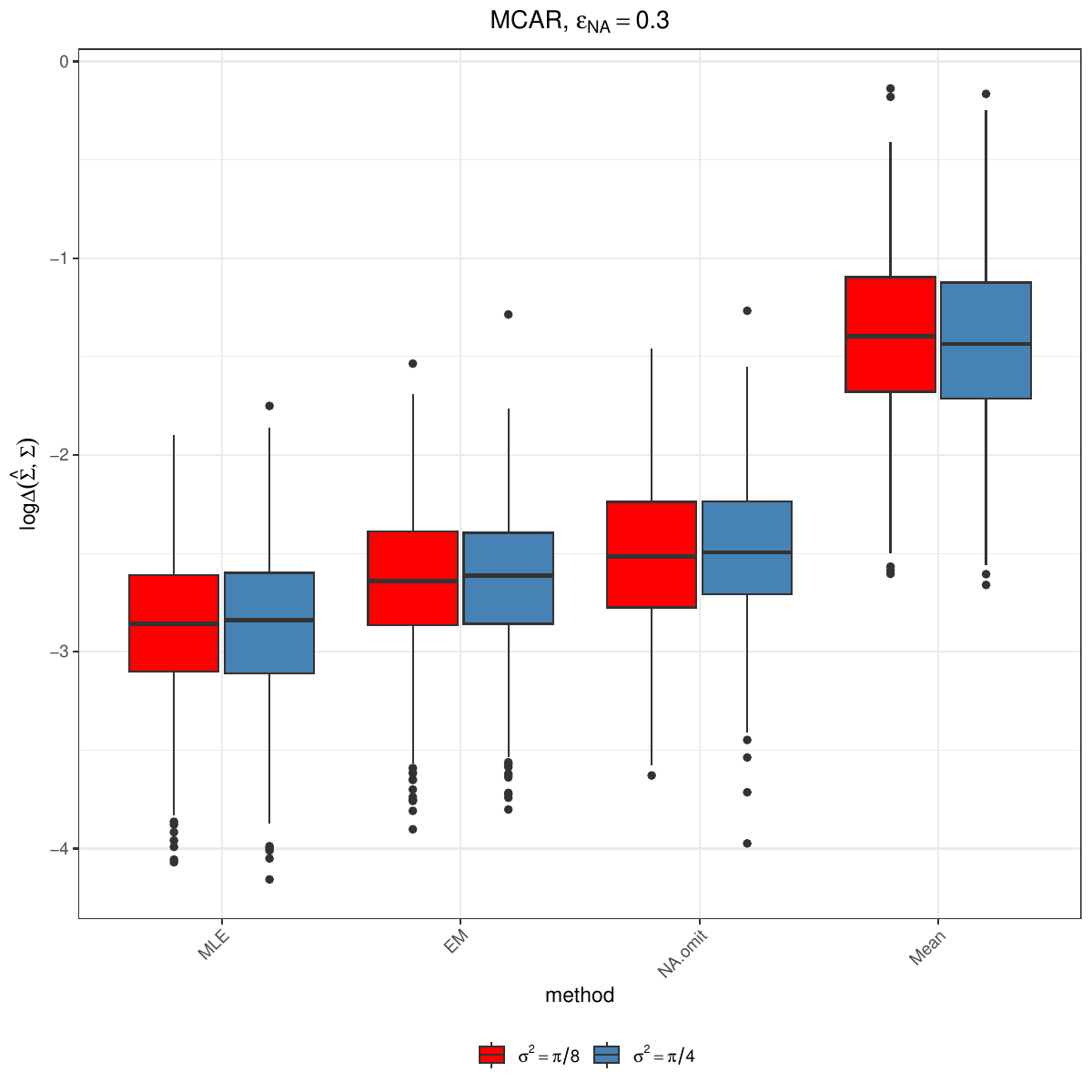}
	\includegraphics[scale=0.325]{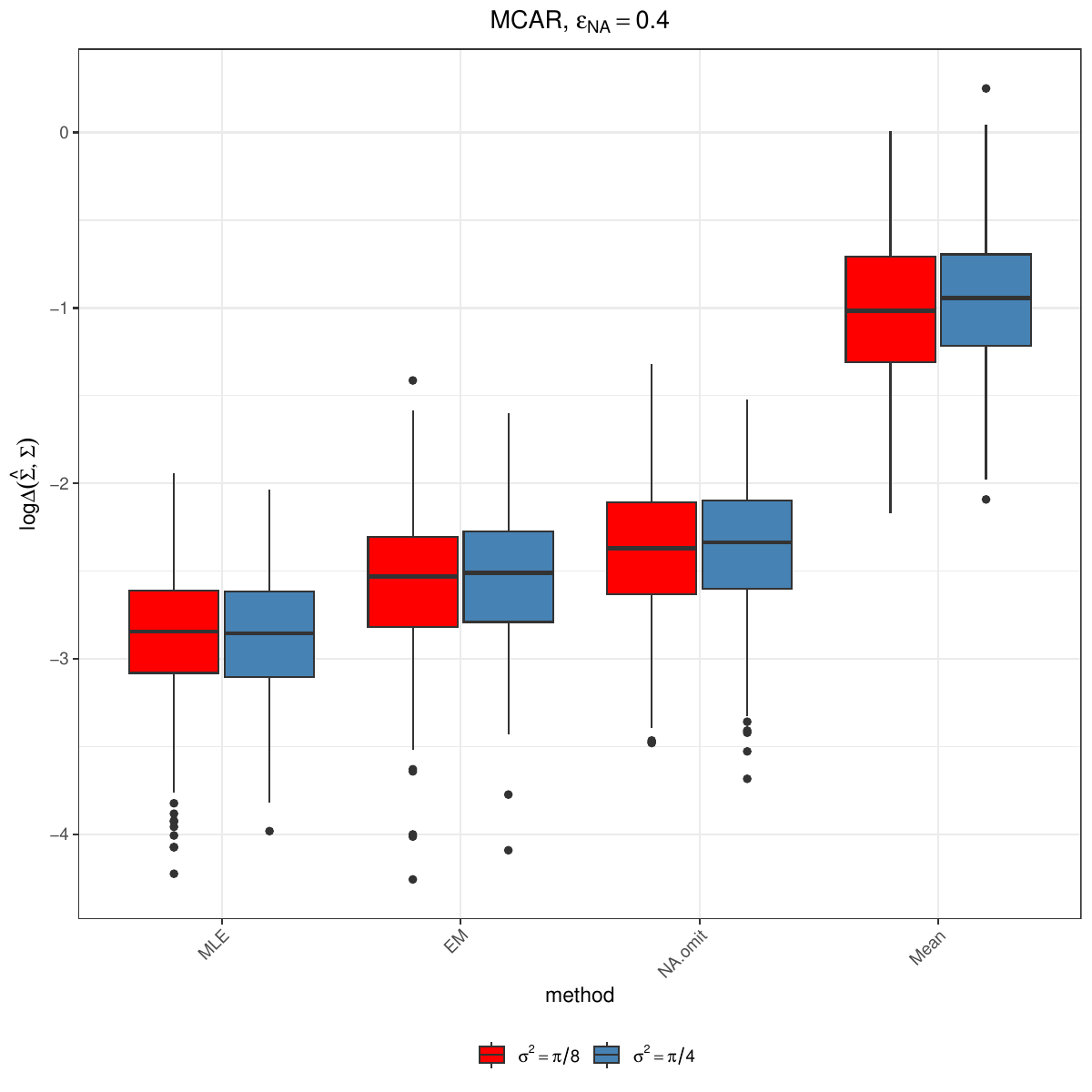}
	\caption{Numerical studies. Distribution of $\log\Delta(\hat\bSigma, \bSigma)$ for different methods uder MCAR when $p=5$. Missingness rate is casewise.}
	\label{fig:6}
\end{figure}

\begin{figure}[!h]
	\includegraphics[scale=0.325]{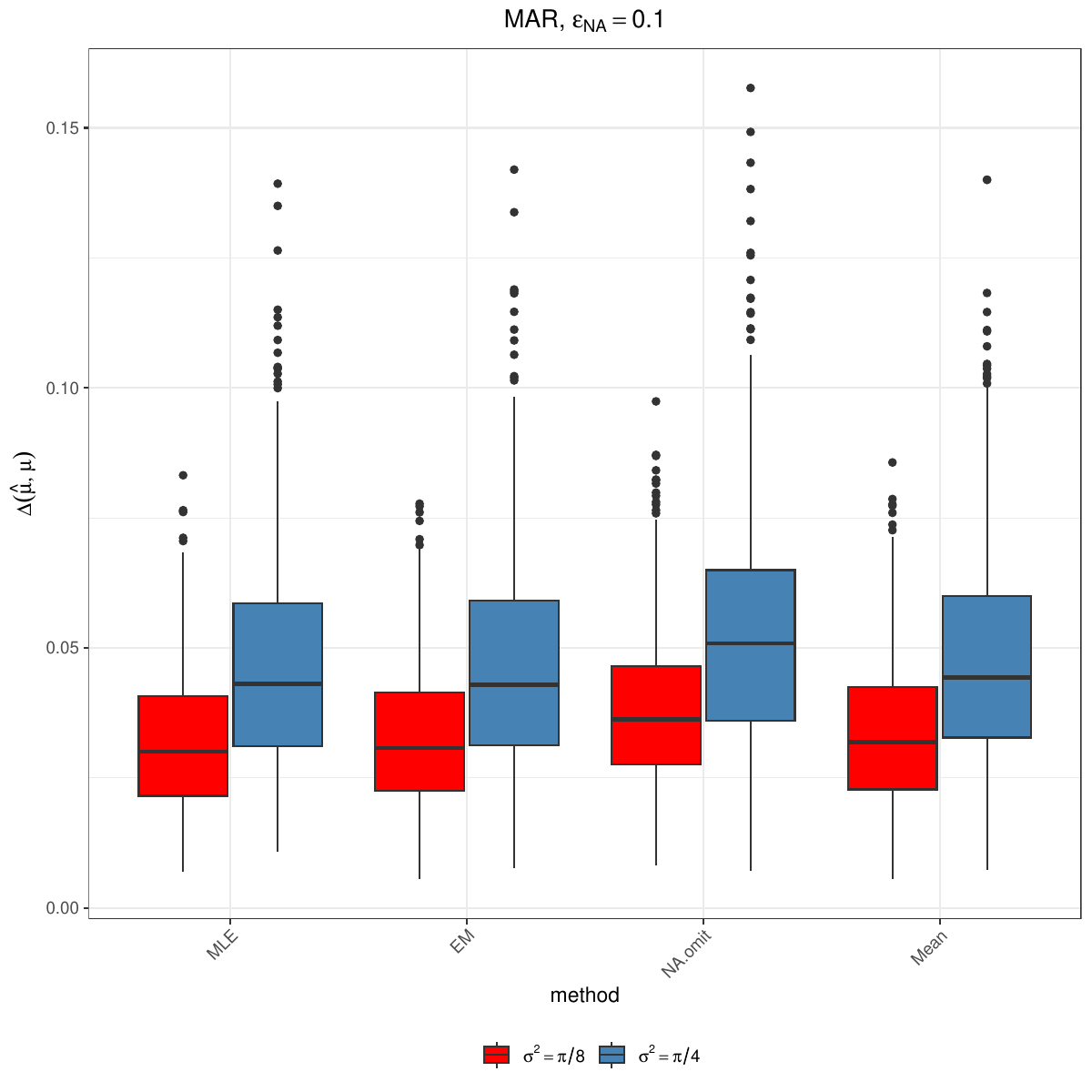}
	\includegraphics[scale=0.325]{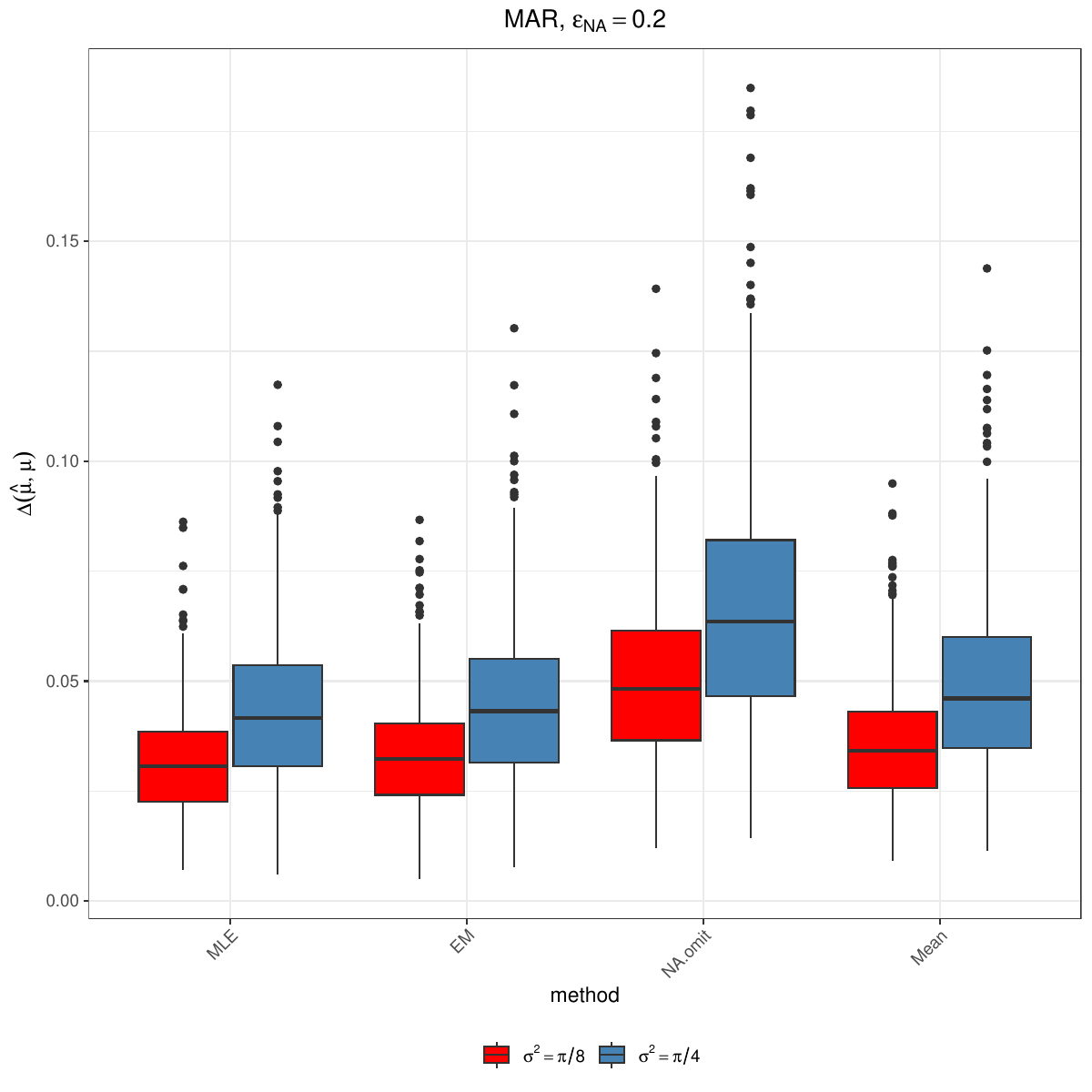}\\
	\includegraphics[scale=0.325]{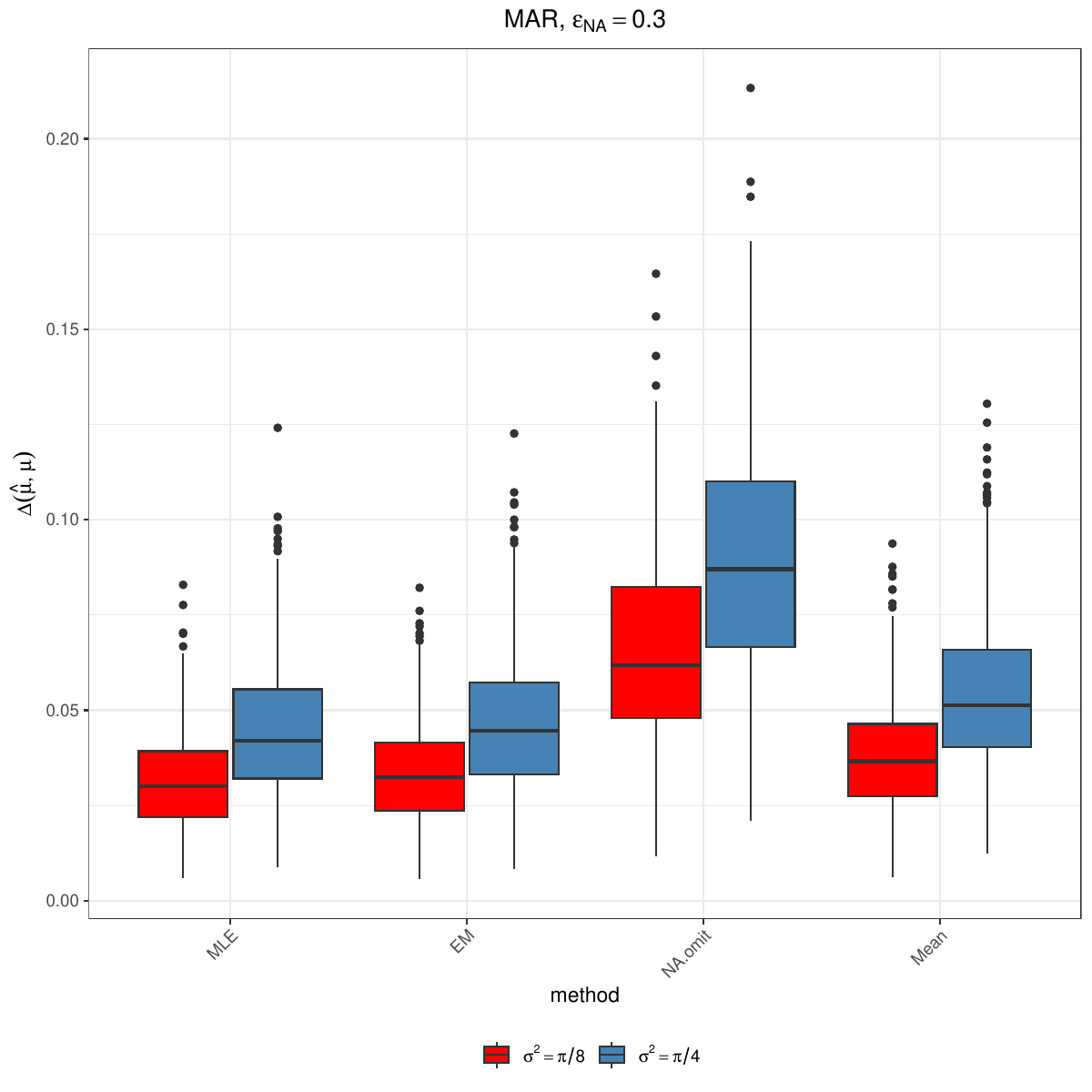}
	\includegraphics[scale=0.325]{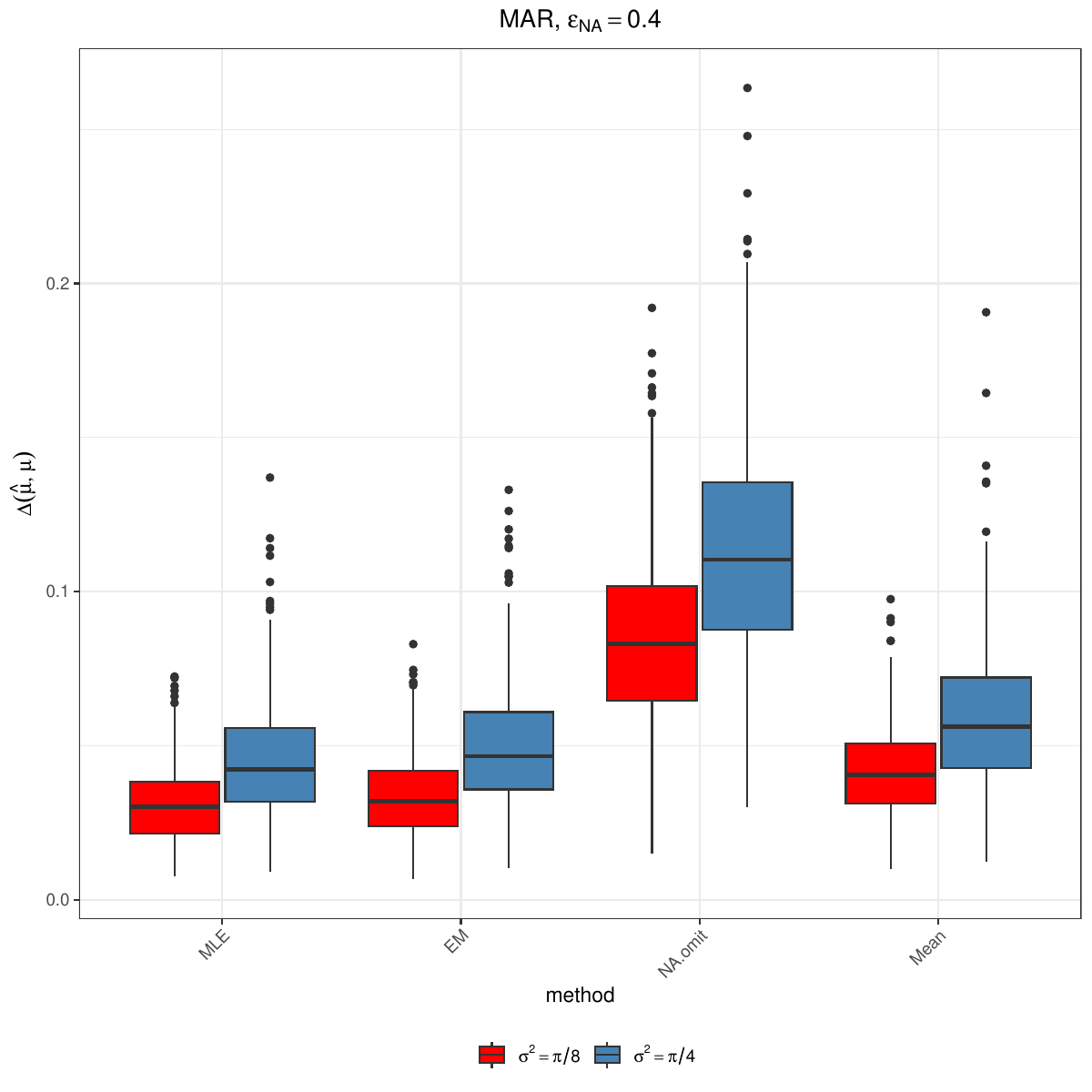}
	\caption{Numerical studies. Distribution of $\Delta(\hat\bmu)$ for different methods uder MAR when $p=5$. Missingness rate is casewise.}
	\label{fig:7}
\end{figure}

\begin{figure}[!h]
	\includegraphics[scale=0.325]{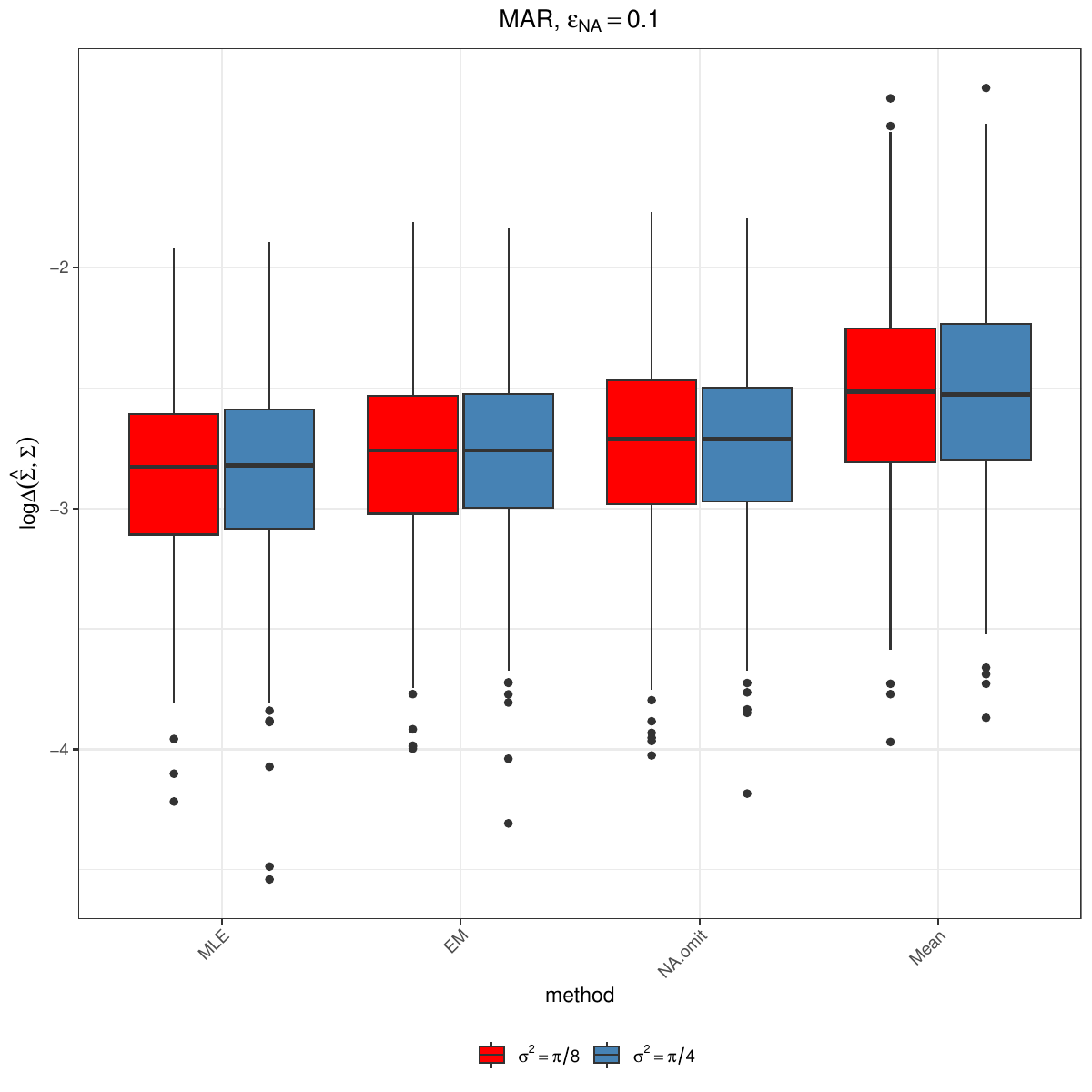}
	\includegraphics[scale=0.325]{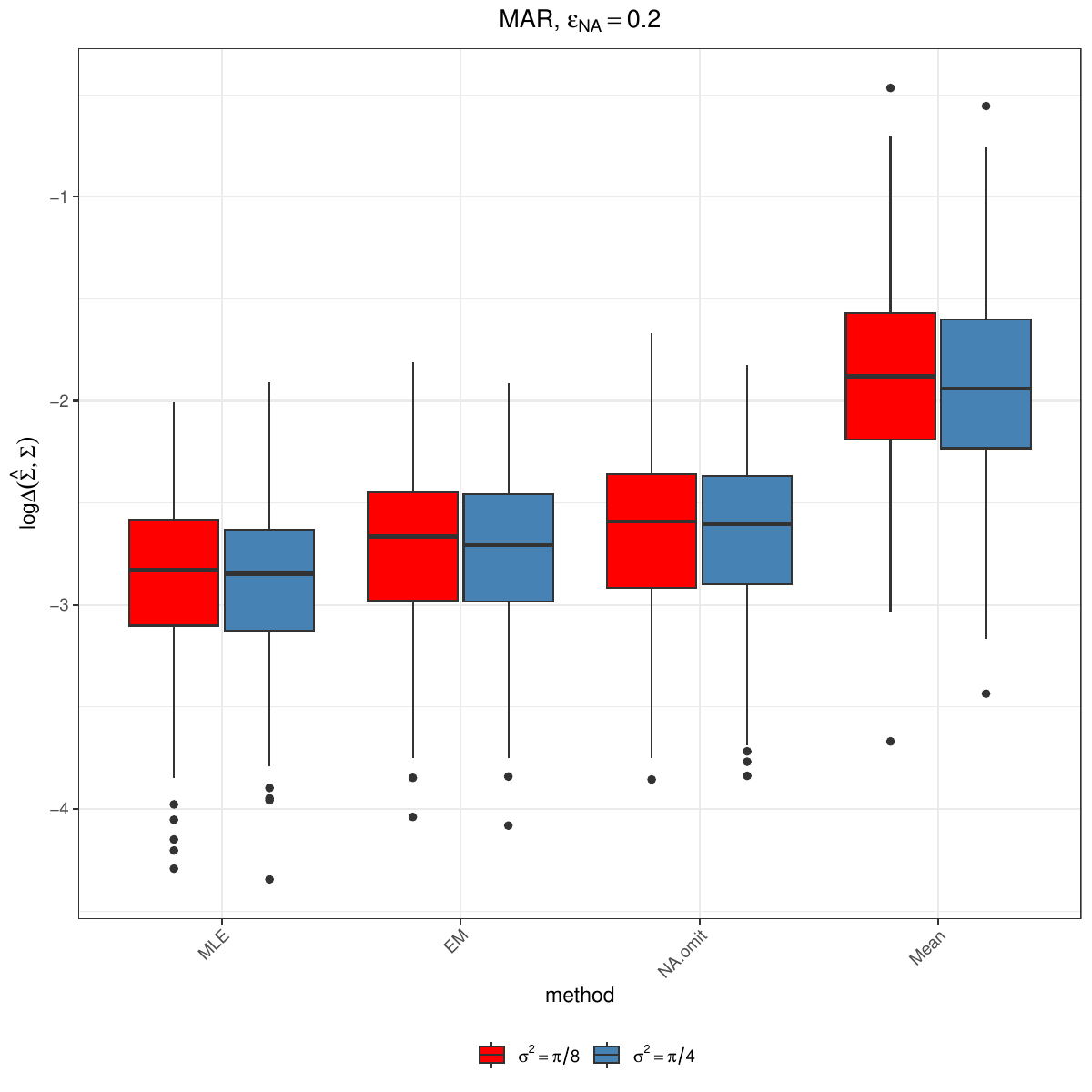}\\
	\includegraphics[scale=0.325]{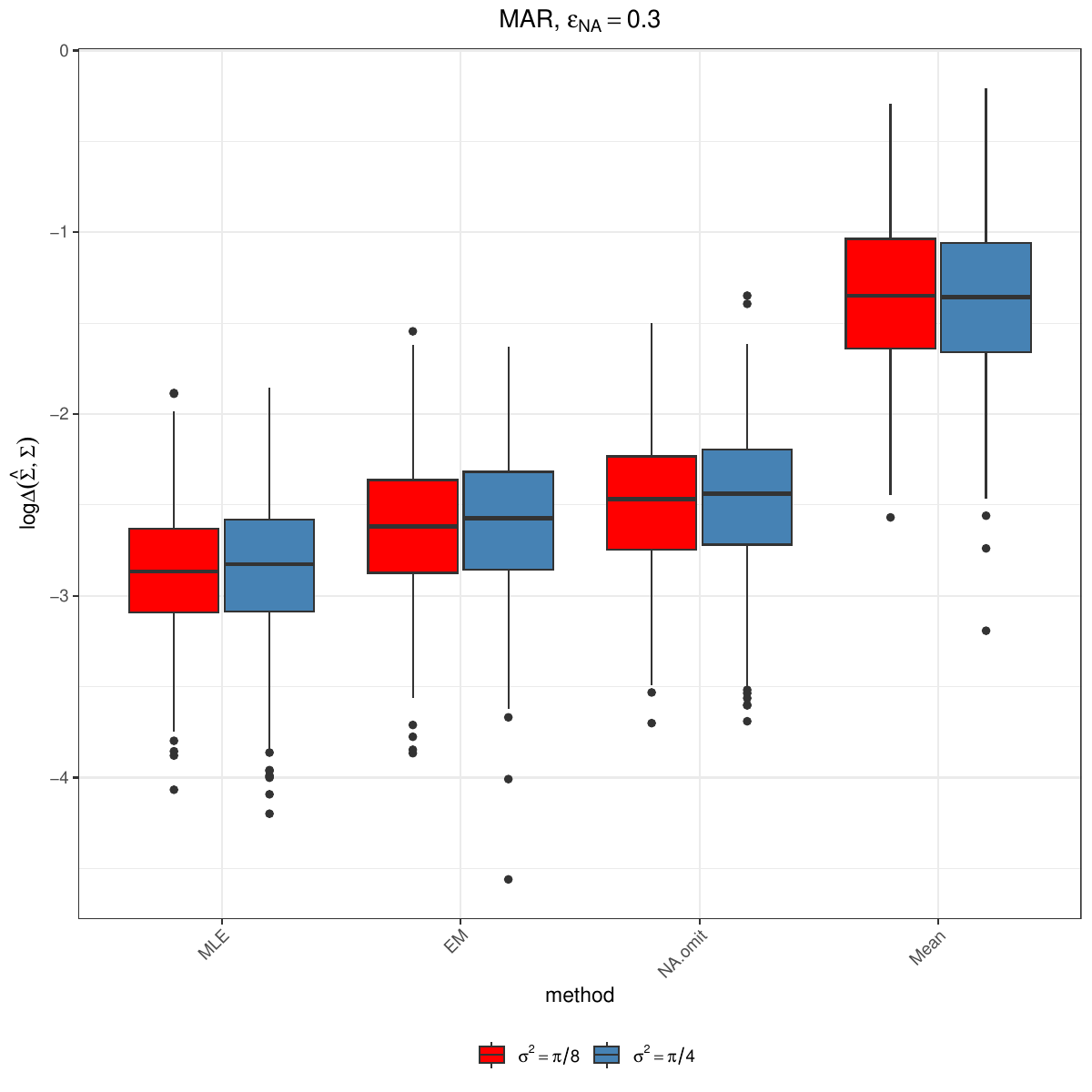}
	\includegraphics[scale=0.325]{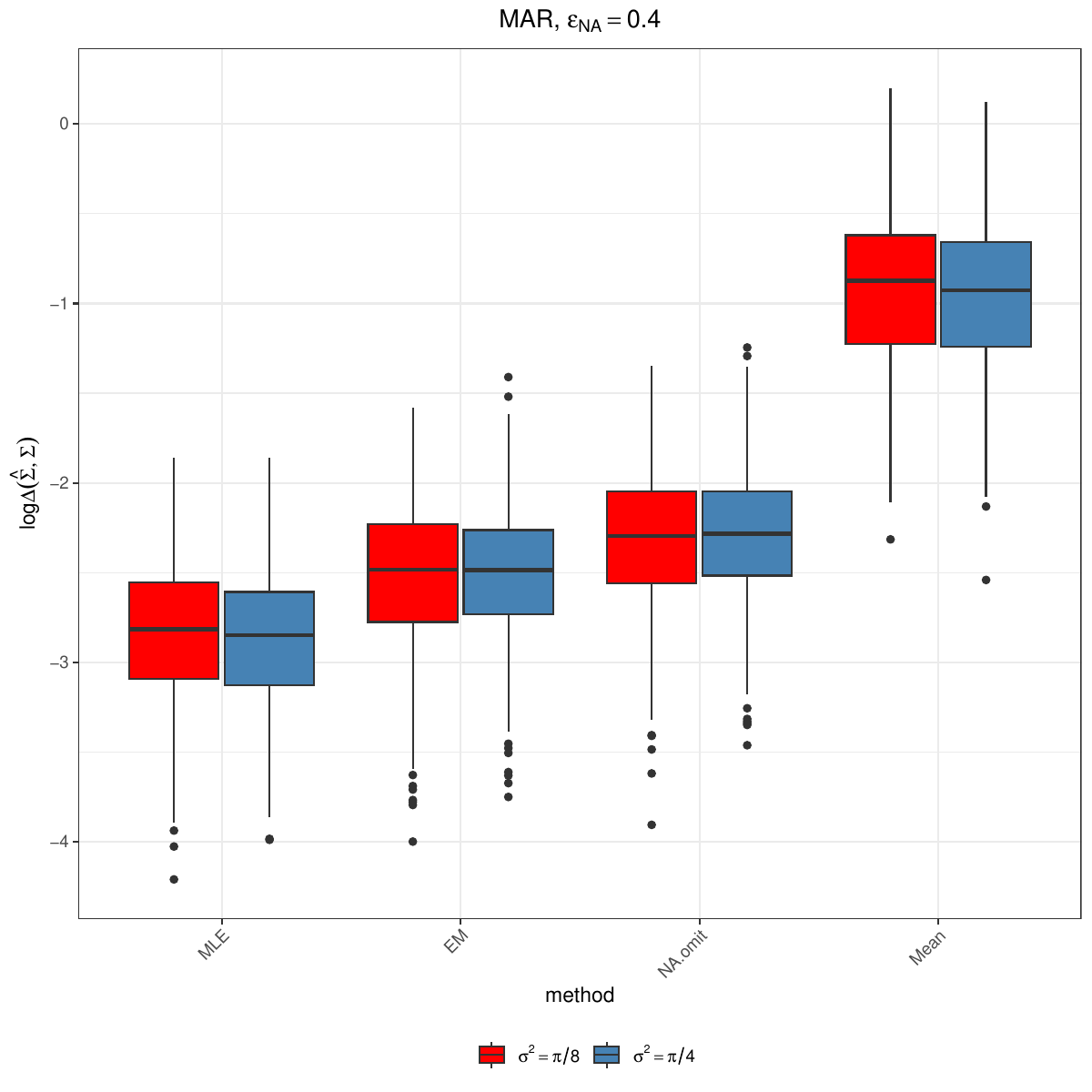}
	\caption{Numerical studies. Distribution of $\log\Delta(\hat\bSigma, \bSigma)$ for different methods uder MAR when $p=5$. Missingness rate is casewise.}
	\label{fig:8}
\end{figure}


\begin{figure}[!h]
	\includegraphics[scale=0.325]{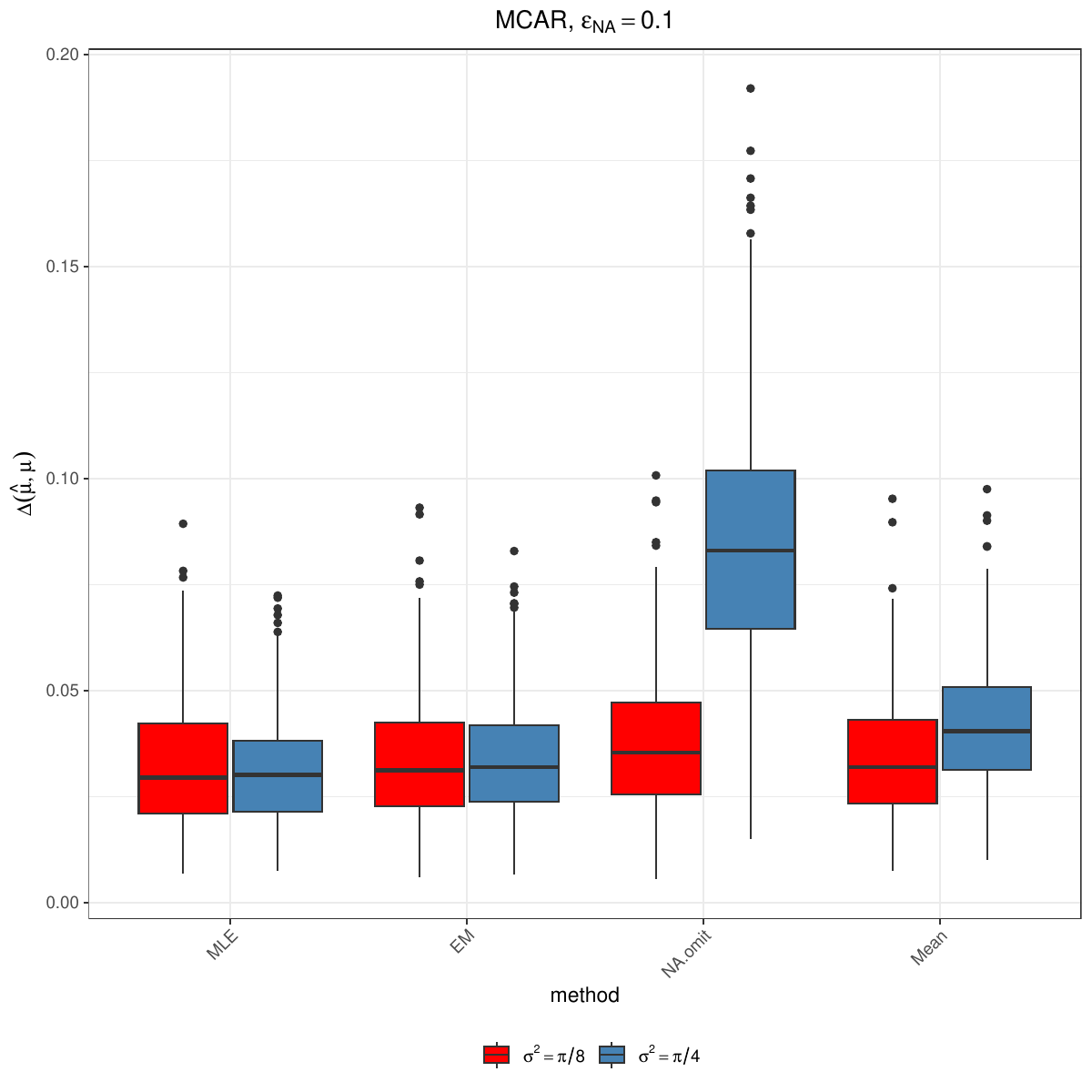}
	\includegraphics[scale=0.325]{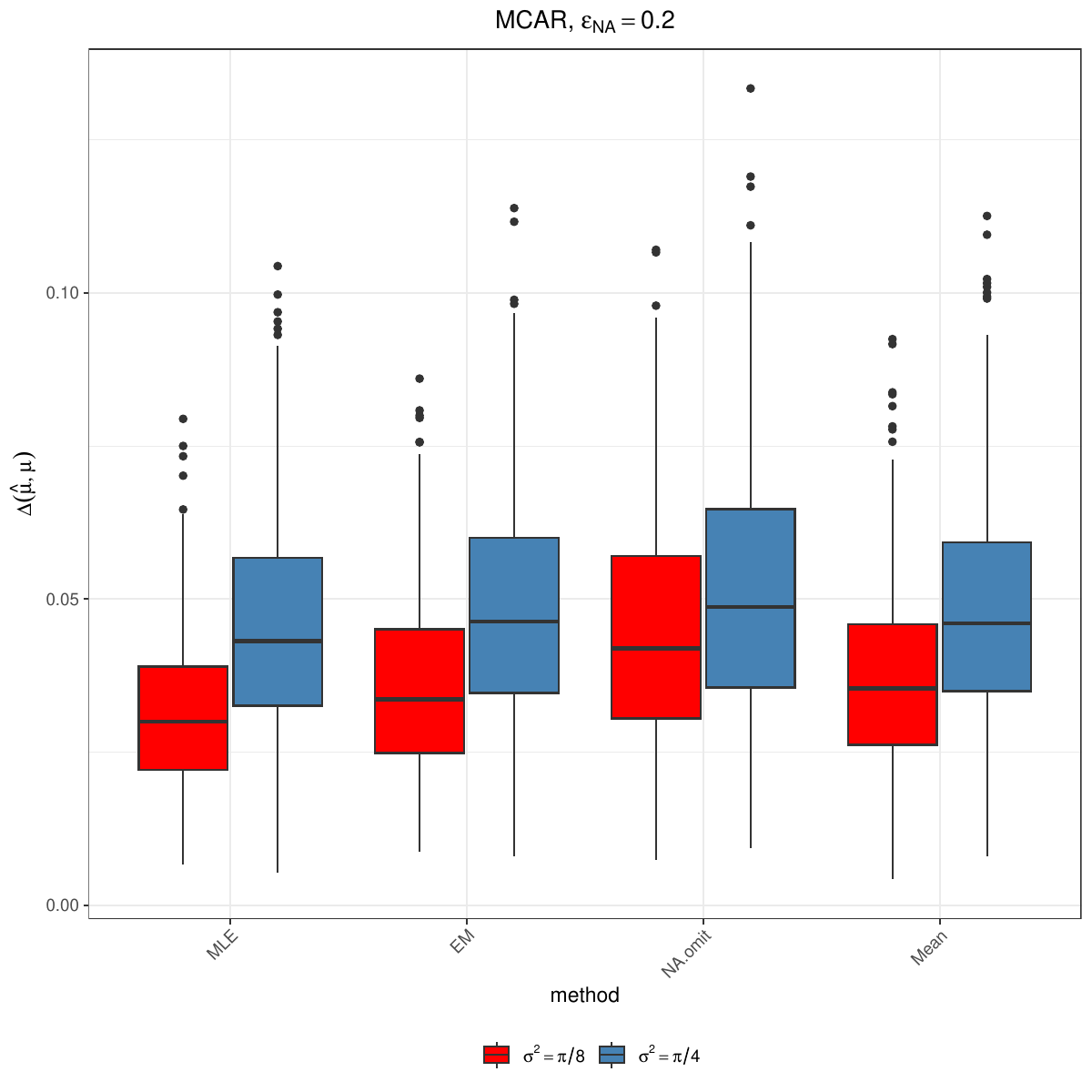}
	\includegraphics[scale=0.325]{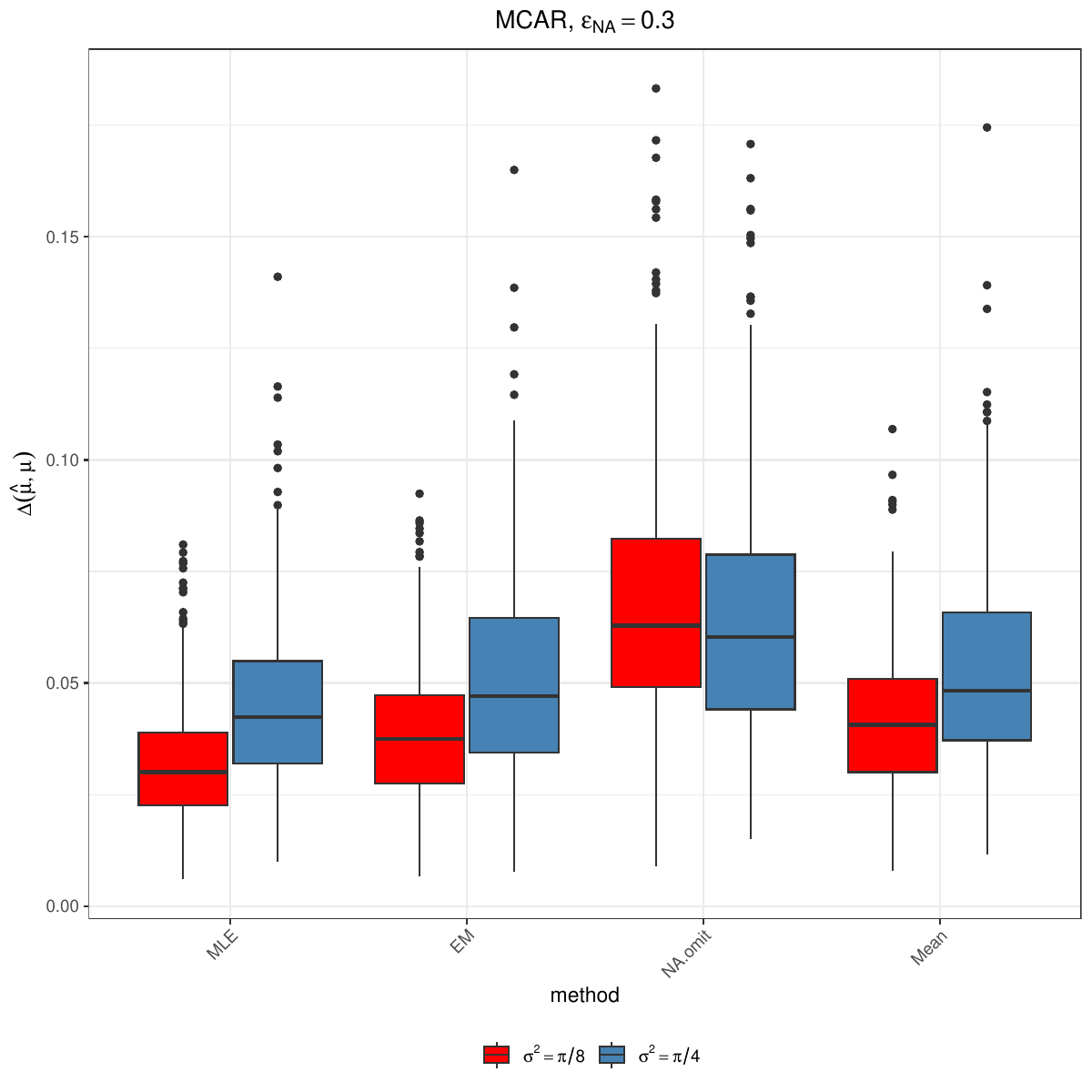}
	\caption{Numerical studies. Distribution of $\Delta(\hat\bmu)$ for different methods uder MCAR when $p=5$. Missingness rate is cellwise.}
	\label{fig:5a}
\end{figure}

\begin{figure}[!h]
	\includegraphics[scale=0.325]{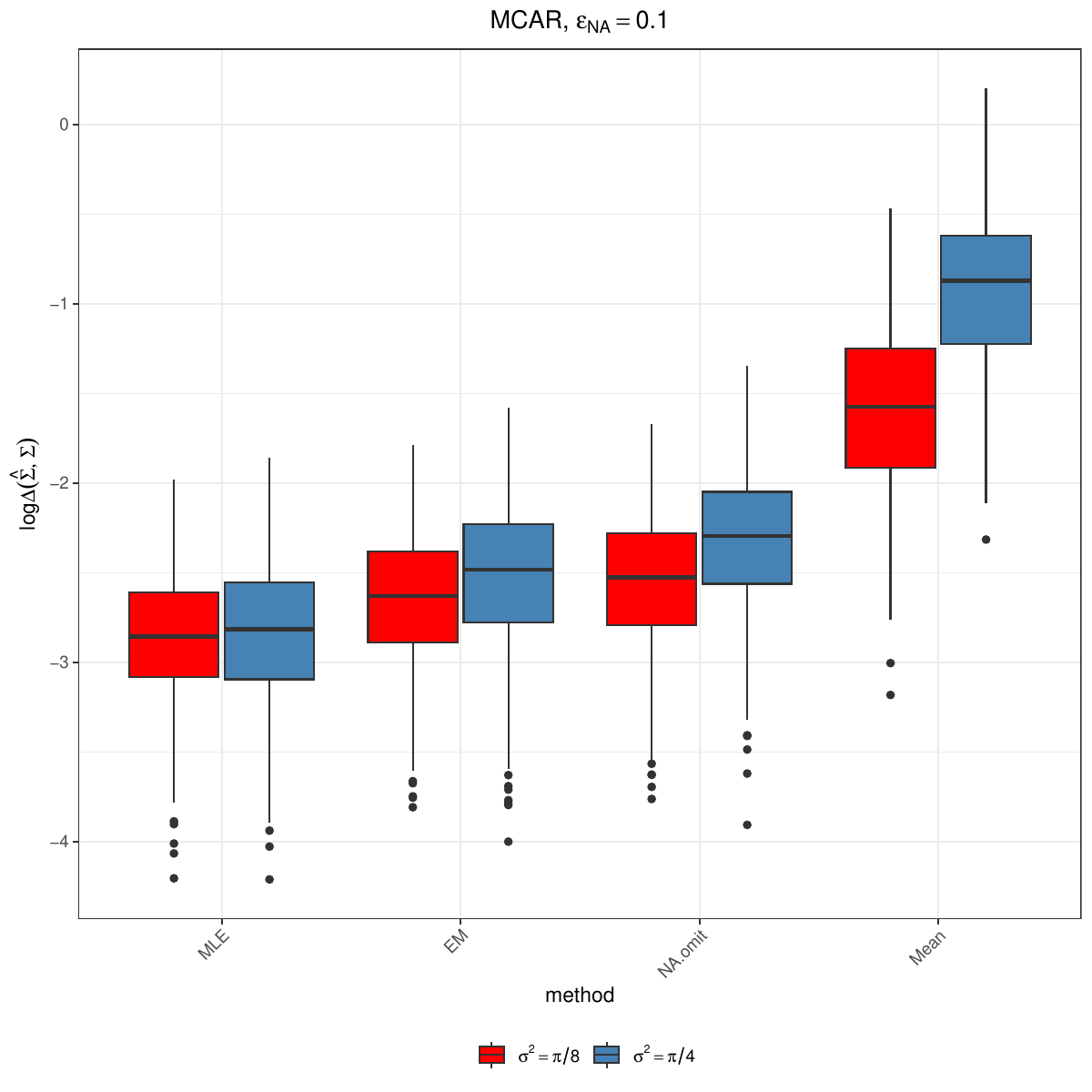}
	\includegraphics[scale=0.325]{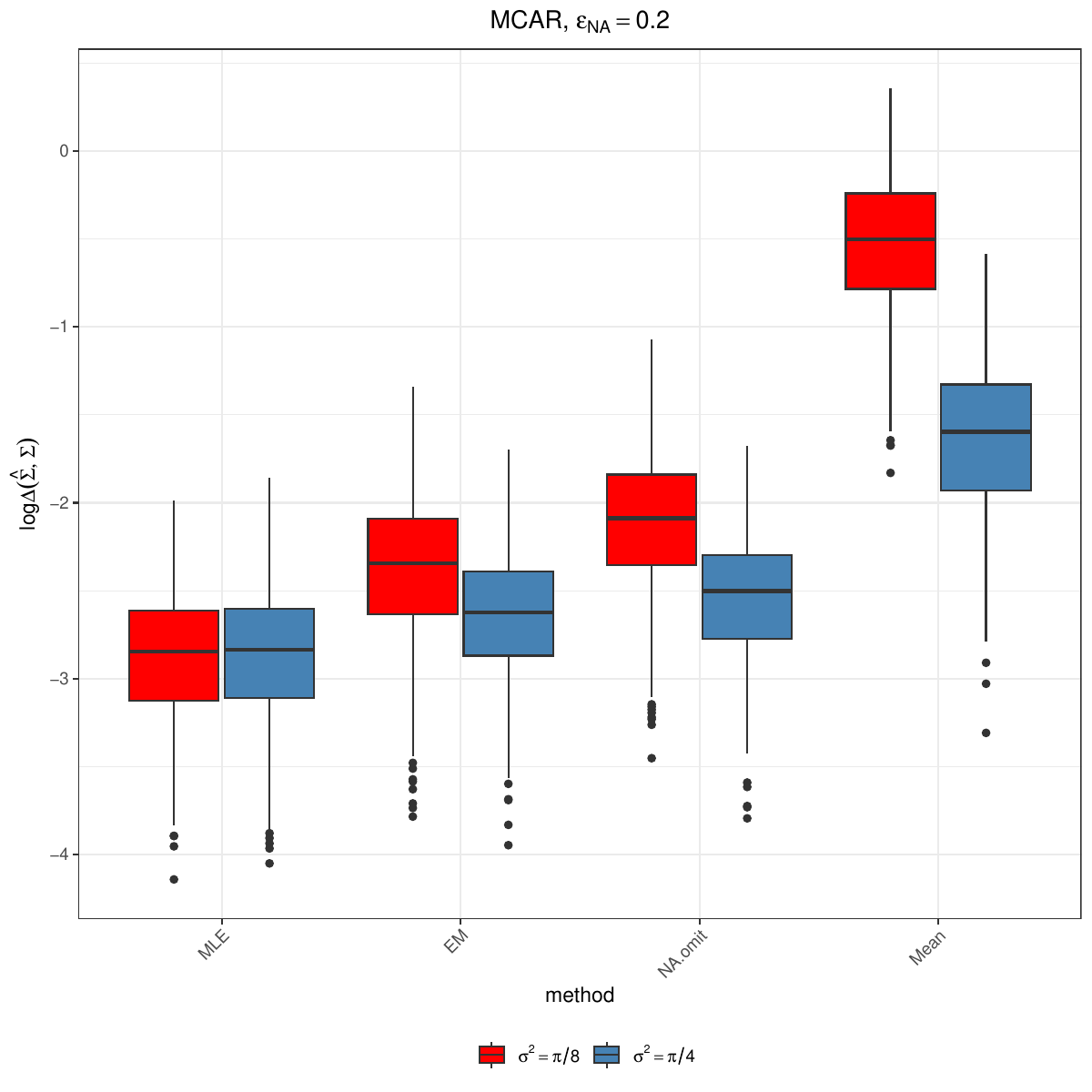}
	\includegraphics[scale=0.325]{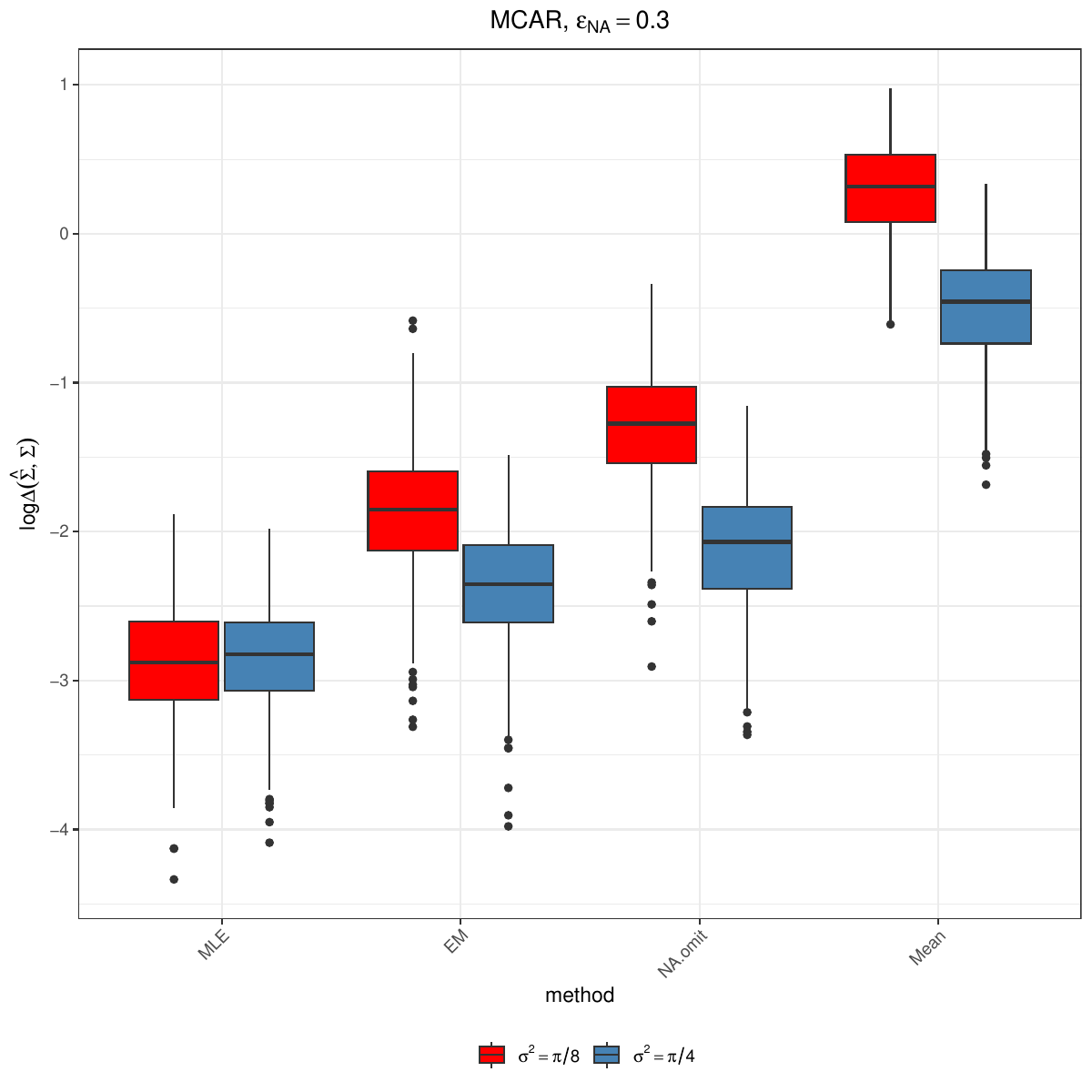}
	\caption{Numerical studies. Distribution of $\log\Delta(\hat\bSigma, \bSigma)$ for different methods uder MCAR when $p=5$. Missingness rate is cellwise.}
	\label{fig:6a}
\end{figure}

\begin{figure}[!h]
	\includegraphics[scale=0.325]{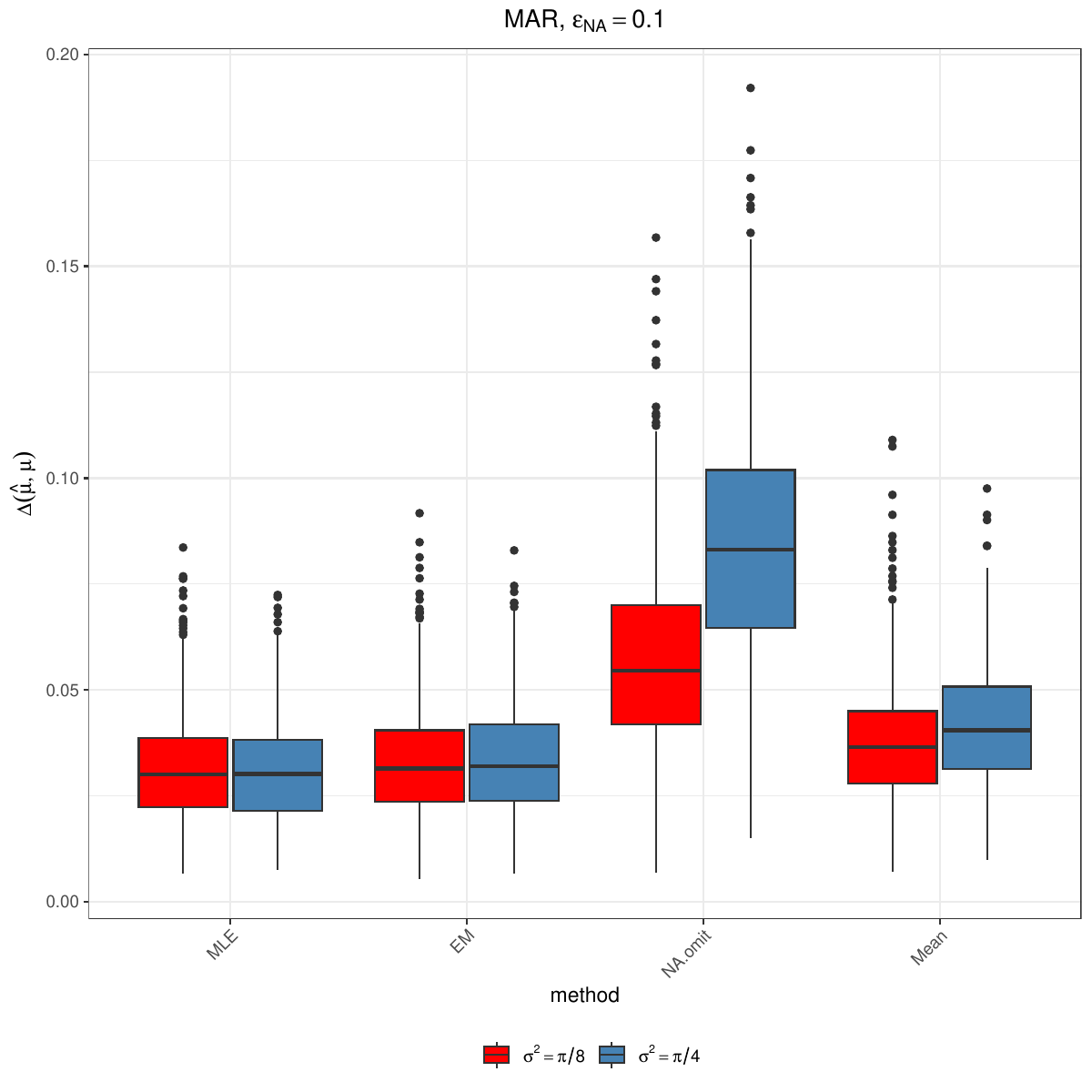}
	\includegraphics[scale=0.325]{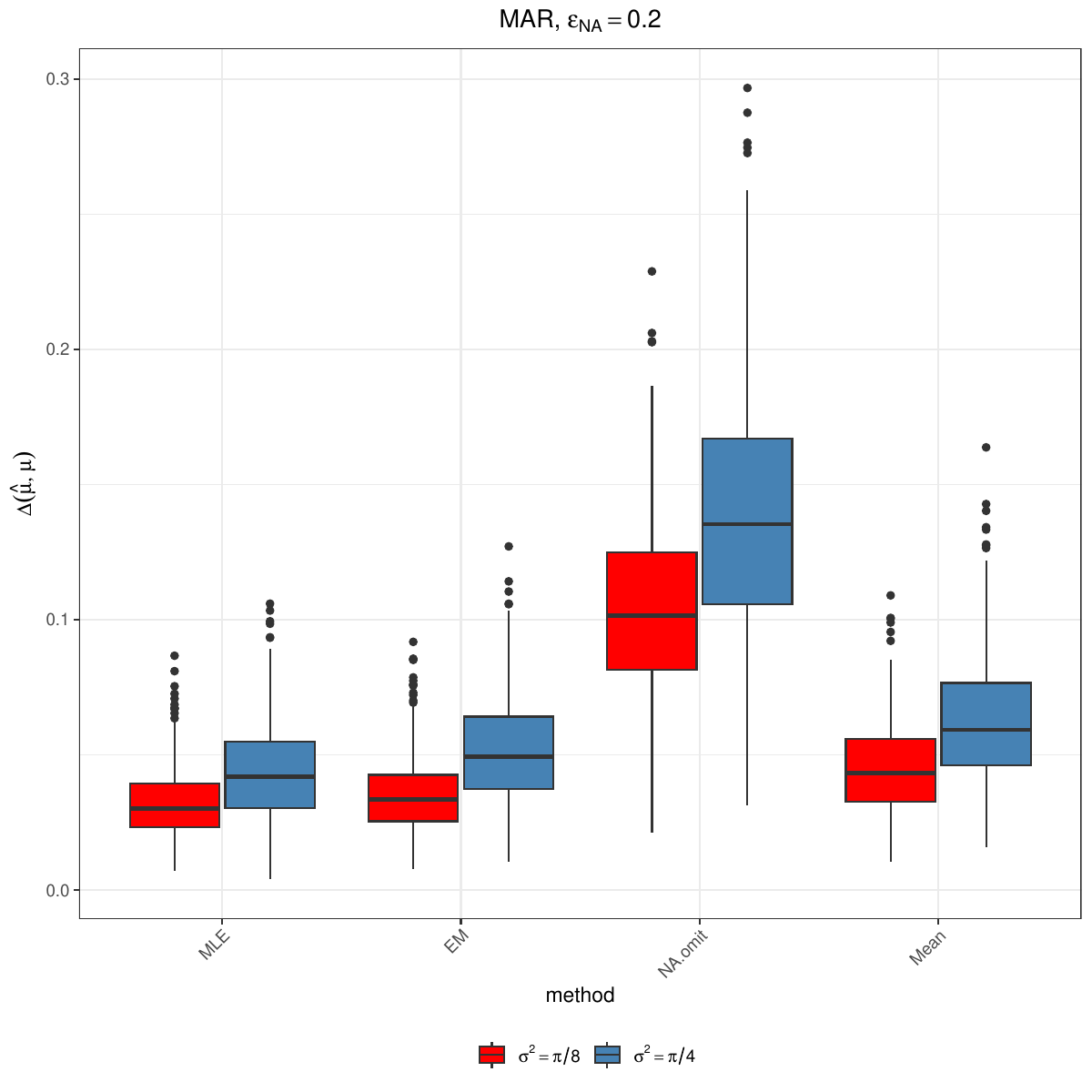}
	\includegraphics[scale=0.325]{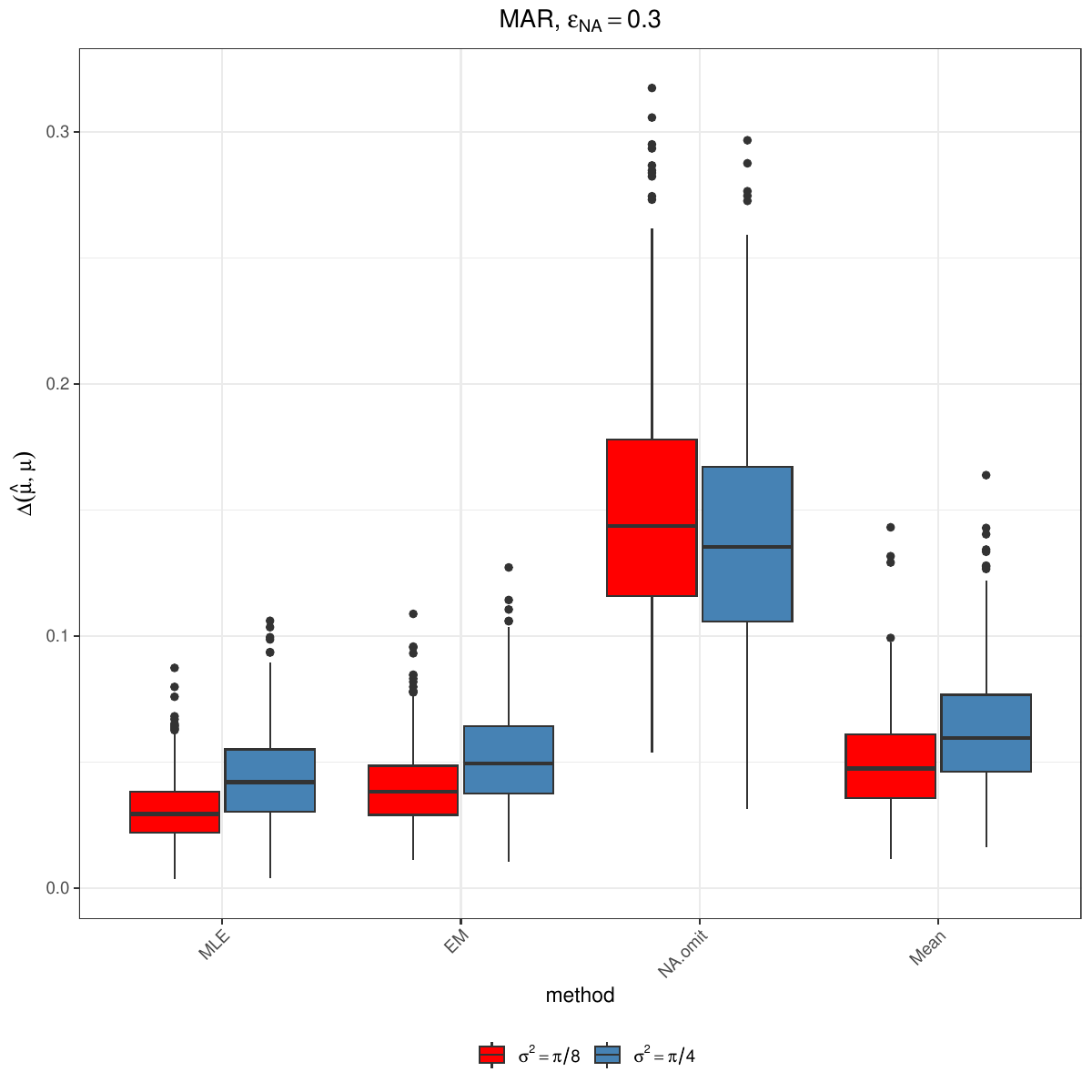}
	\caption{Numerical studies. Distribution of $\Delta(\hat\bmu)$ for different methods uder MAR when $p=5$. Missingness rate is cellwise.}
	\label{fig:7a}
\end{figure}

\begin{figure}[!h]
	\includegraphics[scale=0.325]{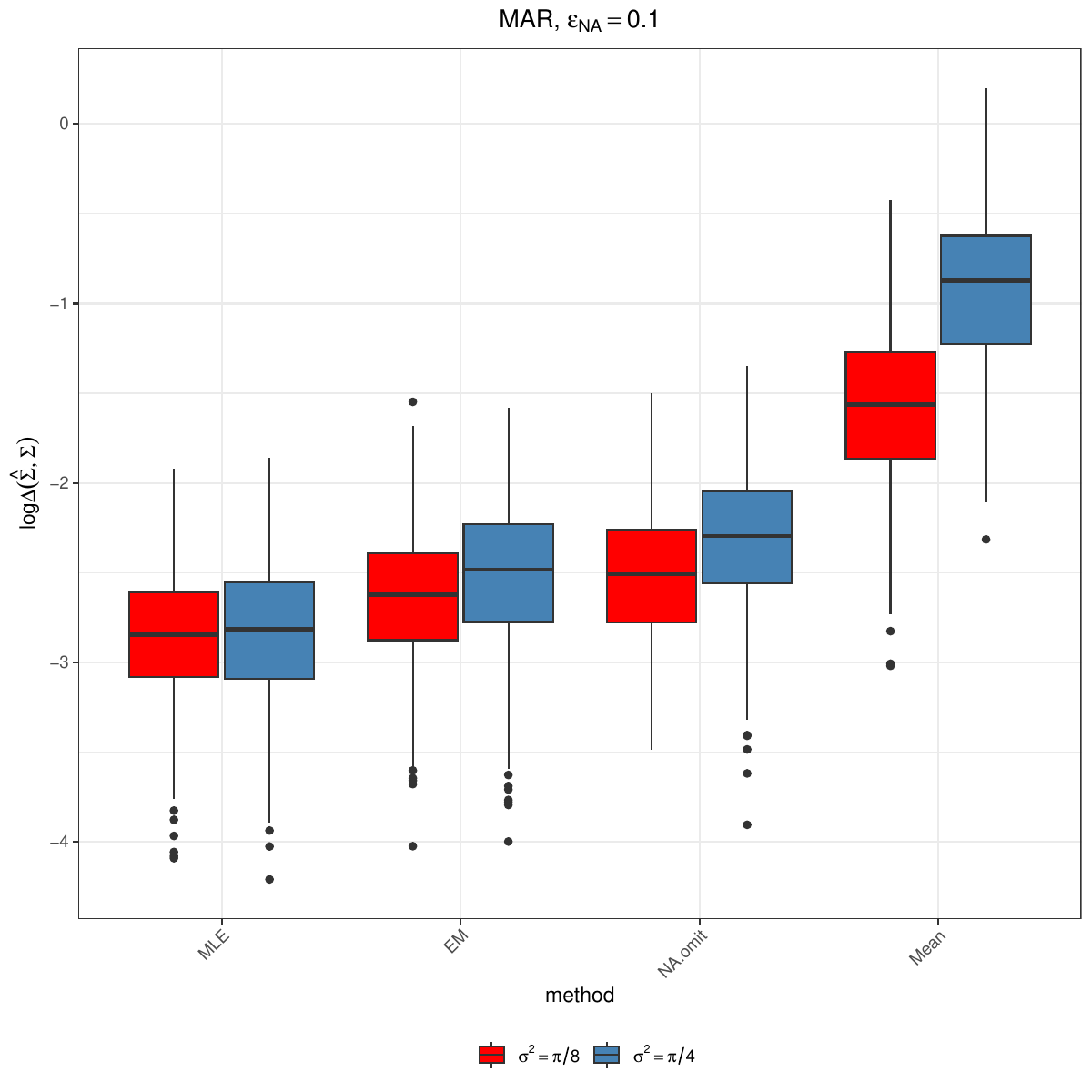}
	\includegraphics[scale=0.325]{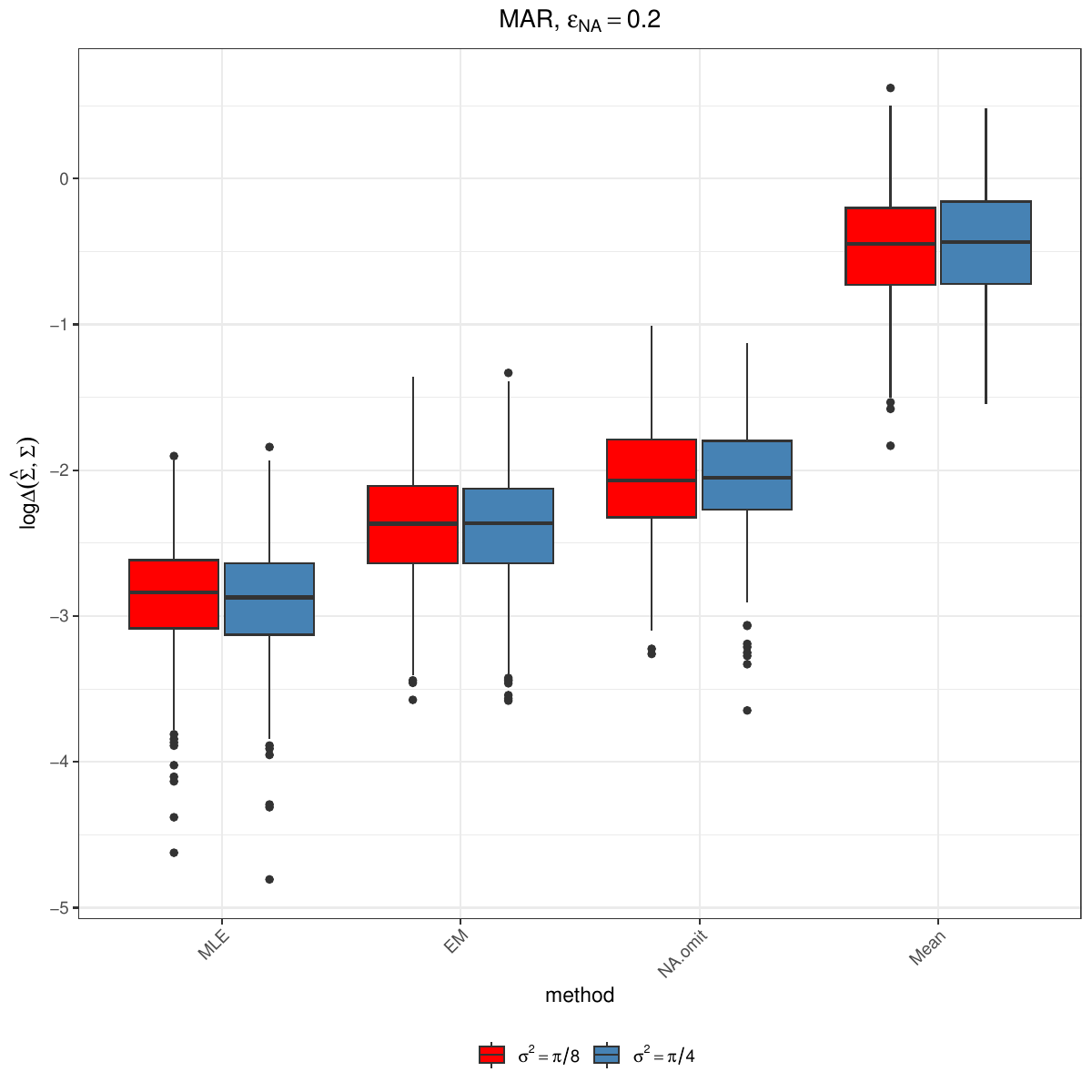}
	\includegraphics[scale=0.325]{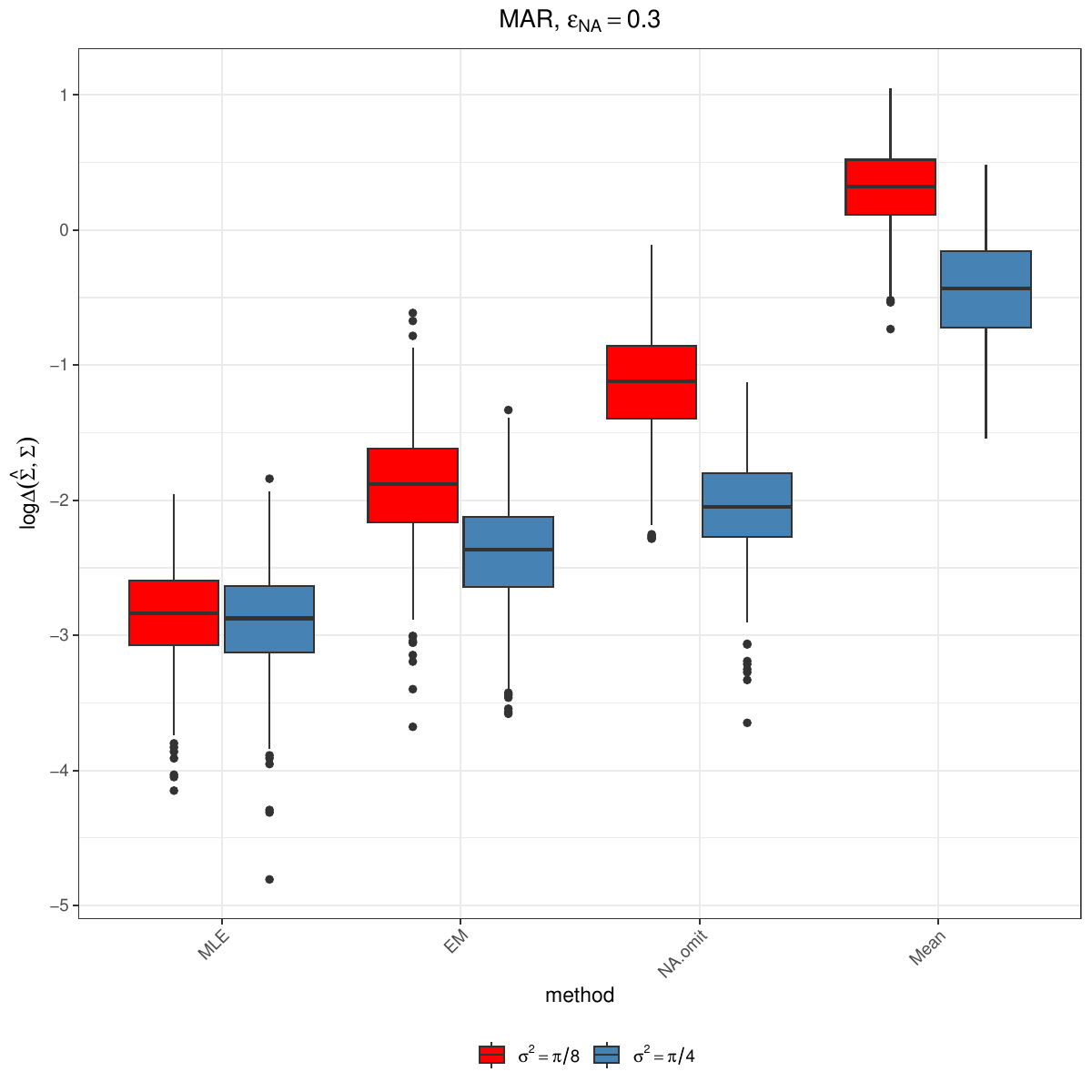}
	\caption{Numerical studies. Distribution of $\log\Delta(\hat\bSigma, \bSigma)$ for different methods uder MAR when $p=5$. Missingness rate is cellwise.}
	\label{fig:8a}
\end{figure}

\begin{figure}[!h]
	\includegraphics[scale=0.325]{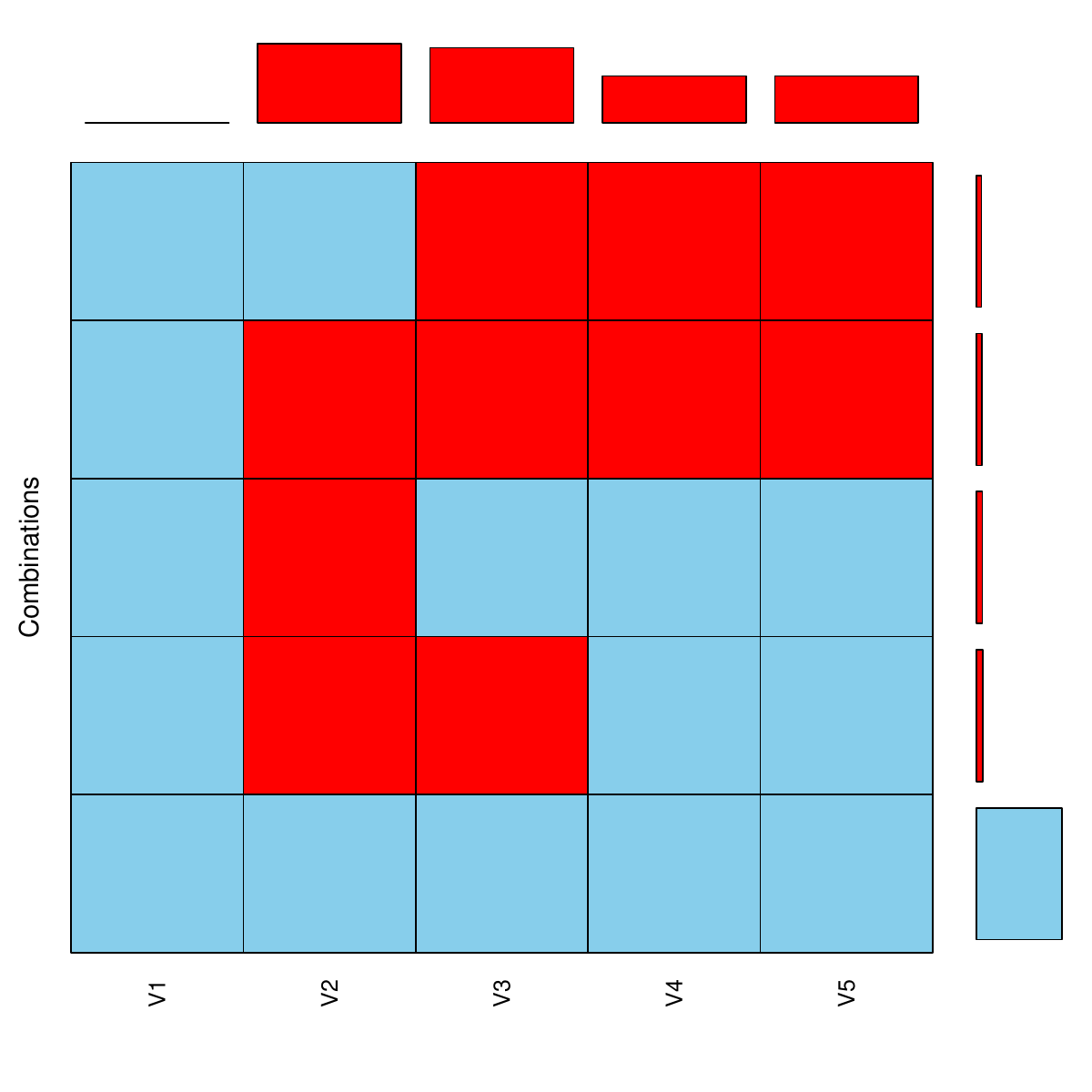}
	\includegraphics[scale=0.325]{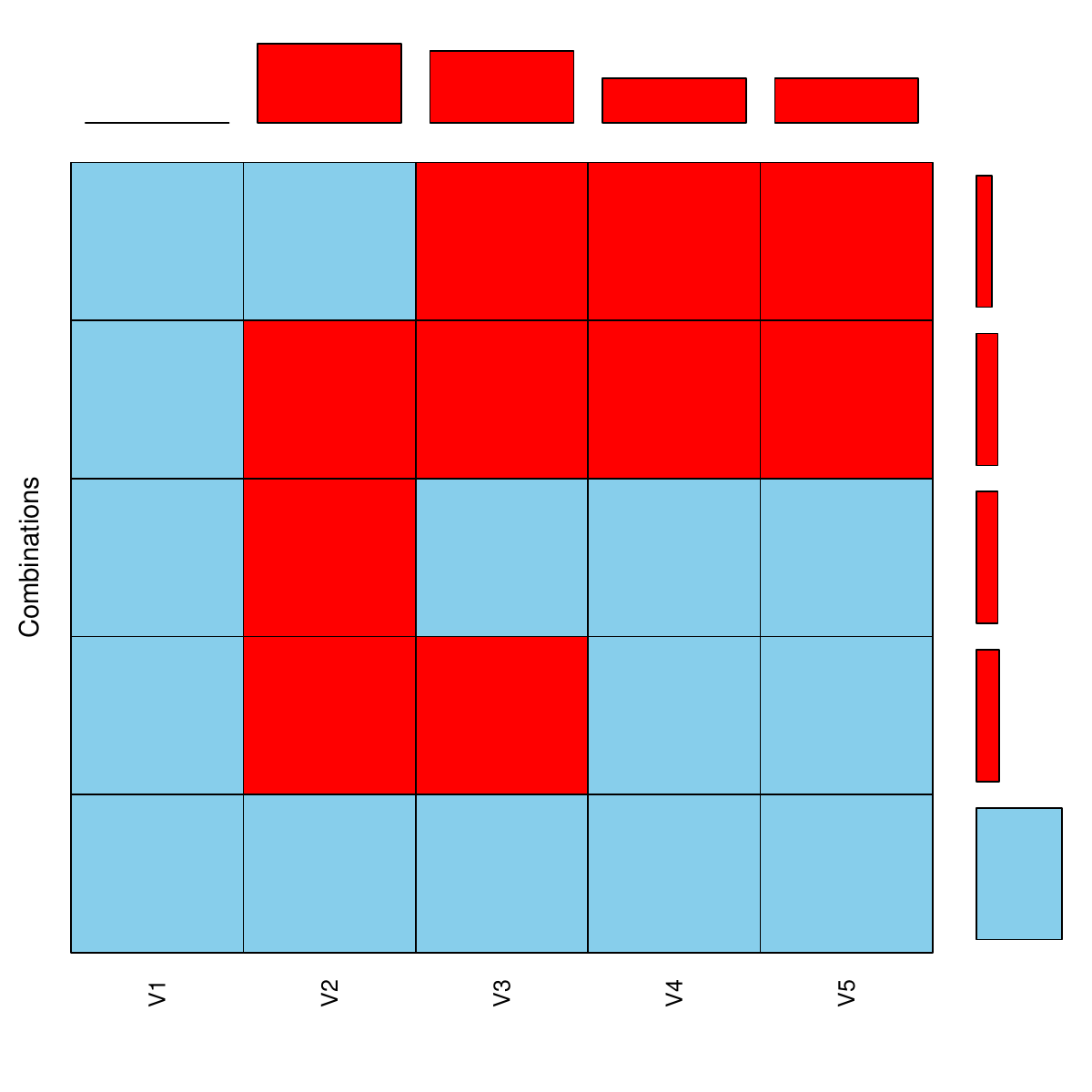}\\
	\includegraphics[scale=0.325]{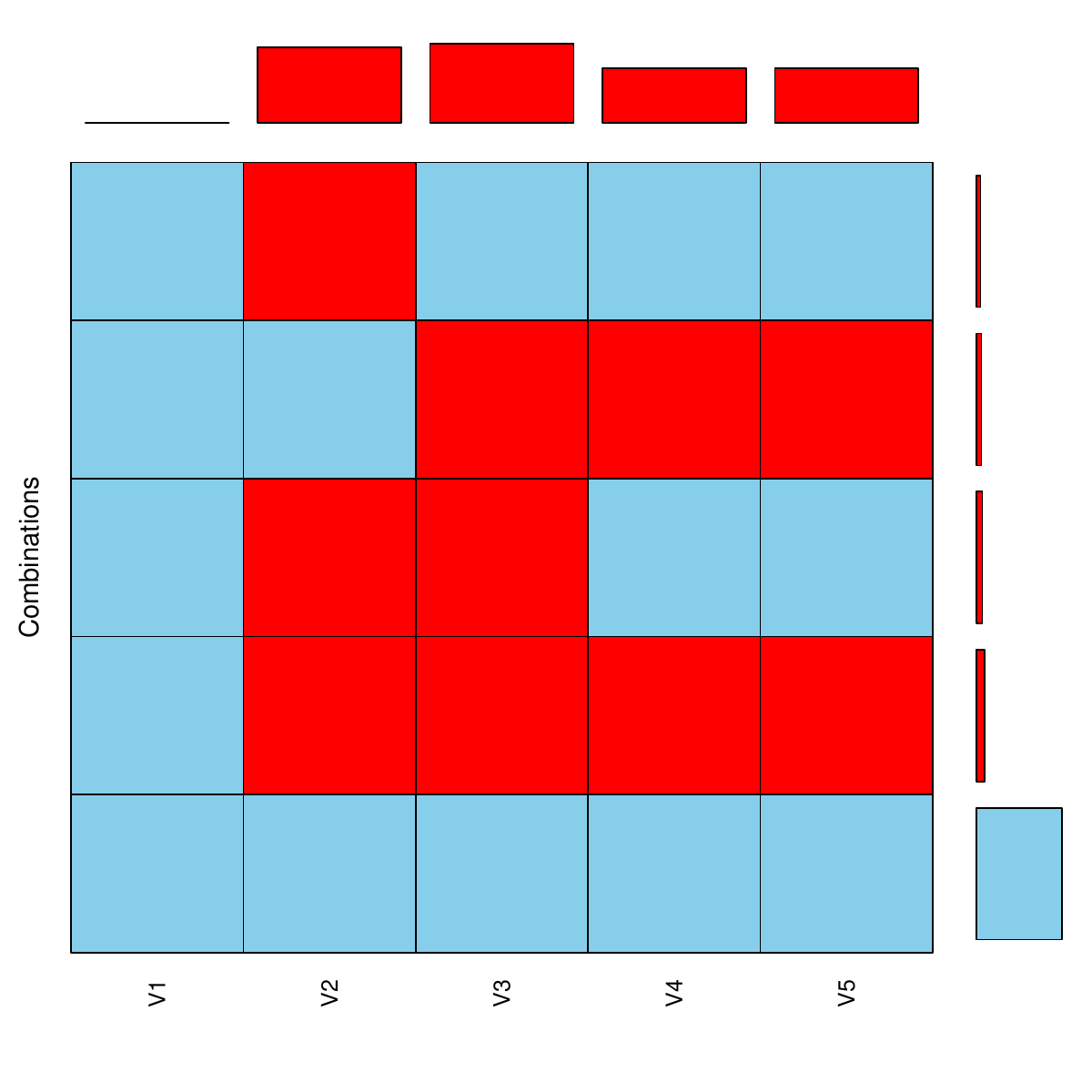}
	\includegraphics[scale=0.325]{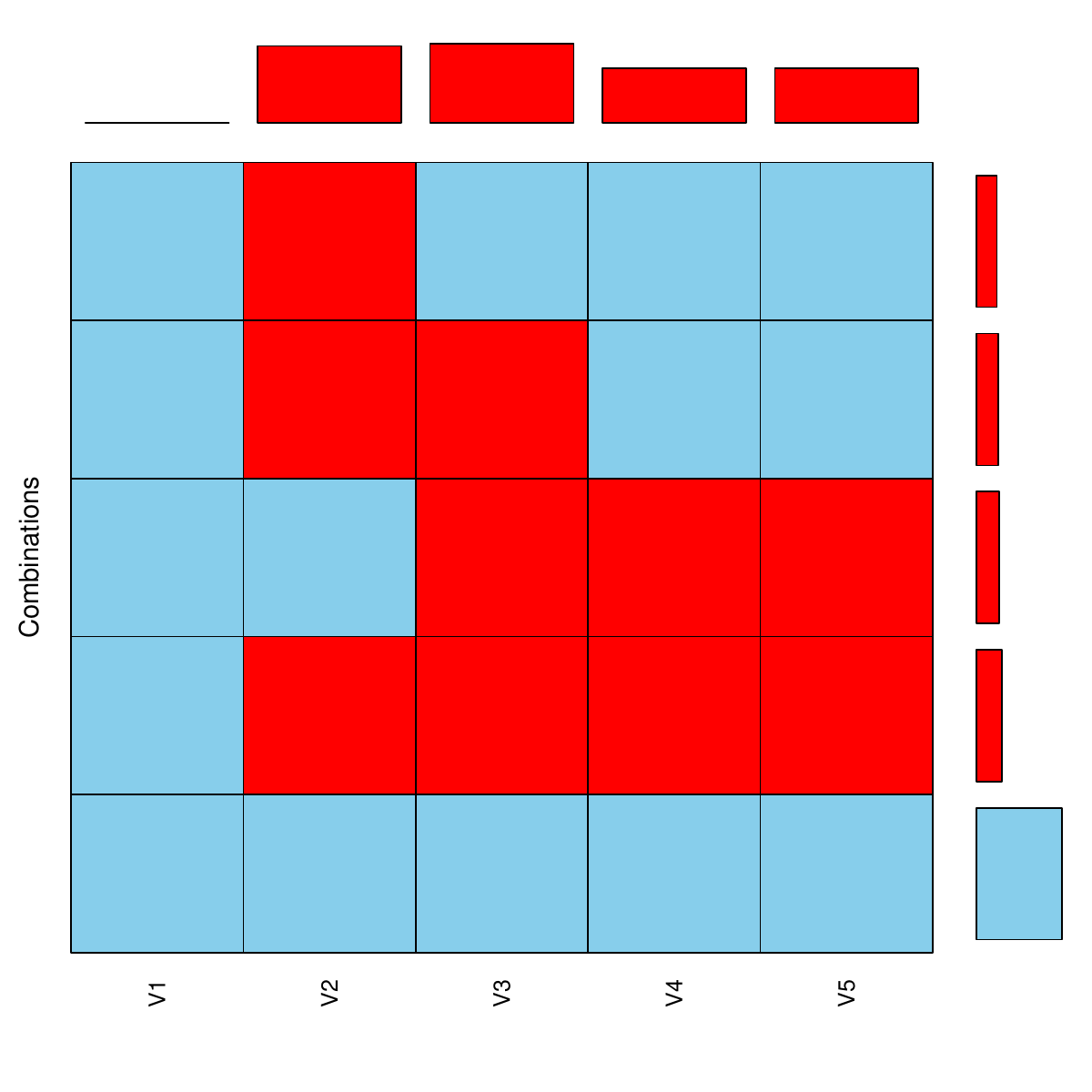}
	\caption{Numerical studies. Missingness patterns when $p=5$, $\epsilon_{NA}=0.2$ under MCAR (top)  and MAR (bottom), casewise (left) and cellwise (right) missingness rate. The top bars give the frequency of missingness by variables; the side bars give the frequency of missingness by pattern. Missing values in red.}
	\label{fig:pattern_p5}
\end{figure}

\end{document}